\documentclass{amsart}
\usepackage[utf8]{inputenc}
\usepackage[T1]{fontenc}
\usepackage{lmodern}
\usepackage{microtype}
\usepackage{booktabs}

\usepackage{amssymb}
\usepackage{mathtools}
\usepackage{url}
\usepackage[plainpages=false,pdfpagelabels,colorlinks=true,citecolor=blue,hypertexnames=false,pdftex=true,unicode]{hyperref}

\newtheorem{thm}{Theorem}[section]
\newtheorem{prop}[thm]{Proposition}
\newtheorem{lem}[thm]{Lemma}
\newtheorem{cor}[thm]{Corollary}

\theoremstyle{definition}
\newtheorem{ex}[thm]{Example}

\theoremstyle{remark}
\newtheorem{rem}[thm]{Remark}

\numberwithin{equation}{section}

\usepackage{caption}
\usepackage{calc}
\usepackage[shortlabels]{enumitem}
\newlist{algospec}{itemize}{1}
\setlist[algospec]{font=\normalfont\itshape,itemsep=0ex,partopsep=0ex}
\usepackage{float}
\restylefloat{table}
\newfloat{algofloat}{t}{lop}
\floatname{algofloat}{Algorithm}
\newcommand{\inputs}[1]{%
  \begin{algospec}[nosep,align=right,labelwidth=\widthof{Output:}]
    \item[Input:]
    #1%
  }
\newcommand{\outputs}[1]{%
    \item[Output:] #1
  \end{algospec}%
  \rule[.5\baselineskip]{\textwidth}{.05em}%
}

\newenvironment{algo}
{
  \setcounter{totalnumber}{1}
  \setcounter{topnumber}{1}
  \begin{algofloat}
  \begin{center}\begin{minipage}{.9\linewidth}%
  \rule{\textwidth}{.08em}
}
{
  \rule{\textwidth}{.08em}
  \end{minipage}\end{center}
  \end{algofloat}
}
\newenvironment{algop}
{
  \setcounter{totalnumber}{4}
  \setcounter{topnumber}{4}
  \begin{algofloat}[tp]
  \begin{center}\begin{minipage}{.9\linewidth}%
  \rule{\textwidth}{.08em}
}
{
  \rule{\textwidth}{.08em}
  \end{minipage}\end{center}
  \end{algofloat}
}

\newcommand{\bK}{\mathbb{K}}
\newcommand{\bN}{\mathbb{N}}
\newcommand{\bQ}{\mathbb{Q}}
\newcommand{\bZ}{\mathbb{Z}}
\newcommand{\bC}{\mathbb{C}}

\newcommand\ftmu{\lfloor\tilde\mu\rfloor}
\newcommand\ftnu{\lfloor\tilde\nu\rfloor}

\newcommand{\IM}{R}
\newcommand{\Sylv}{S}

\newcommand{\bigO}{\operatorname{O}}
\newcommand{\softO}{\operatorname{\tilde{O}}}
\newcommand{\Mult}{\operatorname{M}}

\newcommand{\divides}{\mid}

\newcommand{\Ser}[1]{\bK((x^{#1}))}
\newcommand{\Puiseux}{\Ser{1/*}}
\newcommand{\ratram}{\bK(x^{1/*})}

\newcommand{\sqrfree}{\operatorname{sqrfree}}
\newcommand{\Spoly}{R}

\DeclareMathOperator{\rk}{rk}
\DeclareMathOperator{\id}{id}

\usepackage{tikz}
\usetikzlibrary{shapes.geometric}

\begin{document}

\title{Computing solutions of linear Mahler equations}

\author[F.~Chyzak]{Frédéric Chyzak}
\address{%
  Frédéric Chyzak,
  INRIA, Université Paris-Saclay (France)
}
\email{frederic.chyzak@inria.fr}

\author[Th.~Dreyfus]{Thomas Dreyfus}
\address{%
  Thomas Dreyfus,
  CNRS (France),
  Institut de Recherche Mathématique Avancée, UMR 7501,
  Université de Strasbourg,
  7 rue René Descartes 67084 Strasbourg
}
\email{dreyfus@math.unistra.fr}

\author[Ph.~Dumas]{Philippe Dumas}
\address{%
  Philippe Dumas,
  INRIA, Université Paris-Saclay (France)
}
\email{philippe.dumas@inria.fr}

\author[M.~Mezzarobba]{Marc Mezzarobba}
\address{%
  Marc Mezzarobba,
  Sorbonne Université,  CNRS, Laboratoire d'Informatique de Paris~6, LIP6,
  F-75005 Paris, France
}
\email{marc@mezzarobba.net}

\thanks{
This project has received funding from the European Research Council (ERC)
under the European Union's Horizon 2020 research and innovation programme under
the Grant Agreement No 648132.
MM was supported in part by ANR grant ANR-14-CE25-0018-01 (FastRelax).
}


\date{\today}

\begin{abstract}
Mahler equations relate evaluations of the same function~$f$ at iterated
$b$th powers of the variable.
They arise in particular in the study of automatic sequences
and in the complexity analysis of divide-and-conquer algorithms.
Recently, the problem of solving Mahler equations in closed form has
occurred in connection with number-theoretic questions.
A difficulty in the manipulation of Mahler equations
is the exponential blow-up of degrees
when applying a Mahler operator to a polynomial.
In this work, we present
algorithms for
solving linear Mahler equations for series, polynomials, and rational
functions,
and get polynomial-time complexity under a mild assumption.
Incidentally, we develop an algorithm for computing the gcrd
of a family of linear Mahler operators.
\end{abstract}

\maketitle

\section{Introduction}\label{sec:intro}

\subsection{Context}

Our interest in the present work is in computing various classes of solutions
to \emph{linear Mahler equations} of the form
\begin{equation*}\tag{\ref{eq:mahler-eqn}}
\ell_r (x) y (x^{b^r}) + \cdots + \ell_1 (x) y (x^b) + \ell_0 (x) y (x) = 0 ,
\end{equation*}
where $\ell_0, \ldots, \ell_r$ are given polynomials, $r > 0$ is the order
of the equation, and $b \geq 2$ is a fixed integer.

Mahler equations were first studied by Mahler himself
in a nonlinear context~\cite{Mahler-1929-AEL}.
His aim was to develop a general method to prove the transcendence
of values of certain functions.
Roughly speaking, the algebraic relations over $\bar{\bQ}$
between certain of these values come from algebraic
relations over $\bar\bQ(x)$ between the functions themselves.
This direction was continued by several authors.
We refer to Pellarin's introduction~\cite{Pellarin-2011-IMM}
for a historical and tutorial presentation,
and to the references therein;
see also Nishioka~\cite{Nishioka-1996-MFT} for a textbook.

Mahler equations are closely linked with automata theory:
the generating series of any $b$-automatic sequence
is a Mahler function, that is, a solution of a linear Mahler equation;
see~\cite{Christol-1979-EPP,ChristolKamaeMendesFranceRauzy-1980-SAA}.
Mahler functions also appear in many areas
at the interface of mathematics and  computer science, including
combinatorics of partitions, enumeration of words, and
the analysis of divide-and-conquer algorithms.

Very recently, functional relations between Mahler functions have been further studied
with a bias to effective tests and procedures
\cite{AdamczewskiFaverjon-2016-MMT,AdamczewskiFaverjon-2017-MMR,BellCoons-2017-TTM,DreyfusHardouinRoques-2015-HSM,Roques-2018-ARB}.
Such studies motivate the need for algorithms that solve Mahler equations
in various classes of functions.
For instance, testing transcendence of a Mahler series
by the criterion of Bell and Coons~\cite{BellCoons-2017-TTM}
requires to compute truncations of Mahler series to suitable orders.
So does the algorithm by Adamczewski and Faverjon~%
\cite{AdamczewskiFaverjon-2016-MMT,AdamczewskiFaverjon-2017-MMR}
for the explicit computation of all linear dependence relations over~$\bQ$
between evaluations of Mahler functions at algebraic numbers.
Besides, Mahler functions being either rational or transcendental%
---but never algebraic---,
solving Mahler equations for their rational functions
is another natural approach to testing transcendence,
and an alternative to Bell and Coons'
(see further comments on this in~\S\ref{sec:transcendence}).
Similarly, the hypertranscendence criterion
by Dreyfus, Hardouin, and Roques~\cite{DreyfusHardouinRoques-2015-HSM}
relies on determining if certain Mahler equations possess
ramified rational solutions.

\subsection{Related work}

Mahler equations are a special case of difference equations,
in the sense of functional equations
relating iterates of a ring endomorphism~$\sigma$
applied to the unknown function.

Algorithms dealing with difference equations have been widely studied.
In particular, the computation of rational solutions of linear
difference equations with coefficients polynomial
in the independent variable~$x$
is an important basic
brick coming up repeatedly in other algorithms.
Algorithms in the cases of the usual shift~$\sigma(x) = x + 1$
and its $q$-analogue~$\sigma(x) = q x$
have been given by Abramov~\cite{Abramov-1989-RSL,Abramov-1995-RSL}
for equations with polynomial coefficients:
in both cases, the strategy is to compute a denominator bound
before changing unknown functions
and computing the numerator as a polynomial solution
of an auxiliary difference equation.
Bronstein~\cite{Bronstein-2000-SLO} provides
a similar study for difference equations over more general coefficient domains;
his denominator bound is however stated
under a restriction (\emph{unimonomial extensions})
that does not allow for the Mahler operator~$\sigma(x) = x^b$.

Mahler equations can also be viewed as difference equations
in terms of the usual shift~$\sigma(t) = t + 1$
after performing the change of variables $t = \log_b \log_b x$.
This reduction from Mahler to difference equation, however,
does not preserve polynomial coefficients, which means
that neither Abramov's nor Bronstein's algorithm can be used in this setting.

There has been comparatively little
interest in algorithmic aspects specific to Mahler equations.
To the best of our knowledge, the only systematic study is
by Dumas in his PhD thesis~\cite{Dumas-1993-RMS}.
In particular, he describes procedures
for computing various types of solutions of linear Mahler equations
\cite[Chapter~3]{Dumas-1993-RMS}.
However, beside a few gaps of effectiveness,
that work does not take
computational complexity issues into account.
To a large extent, the results of the present work can be viewed as
refinements of it,
with a focus on efficiency and complexity analysis.
More recently, Bell and Coons~\cite{BellCoons-2017-TTM} give degree bounds
that readily translate into algorithms for polynomial and rational
solutions based on undetermined coefficients.
With regard to series solutions, van der Hoeven~\cite[§4.5.3]{vanderHoeven-2002-RBD}
suggests an algorithm that applies, under hypotheses,
to certain equations of the form~\eqref{eq:mahler-eqn} as well as
to certain nonlinear generalizations,
and computes the first~$n$ terms of
a power series solution in~$\softO(n)$ arithmetic operations.
At least in the linear case and in analogy to the case of difference equations,
this leaves the open question of an algorithm in complexity~$\bigO(n)$.

\subsection{Setting}

Our goal in this article is to present algorithms that compute complete
sets of polynomial solutions, rational function solutions,
truncated power series solutions, and
truncated Puiseux series solutions of~\eqref{eq:mahler-eqn}.
More precisely, let~$\bK$ be a (computable) subfield of~$\bC$, and
suppose $\ell_0, \ldots, \ell_r \in \bK[x]$.
Denote by~$\Puiseux$ the field
$\bigcup_{n=1}^{+\infty} \Ser{1/n}$
of formal Puiseux series with coefficients in~$\bK$.
Let~$M$ denote the \emph{Mahler operator of radix~$b$}, that is the
automorphism of~$\Puiseux$ that substitutes~$x^b$ for~$x$ and reduces
to the identity map on~$\bK$.
Writing $x$~again for the operator of multiplication of a series by~$x$,
$M$ and~$x$ follow the commutation rule~$M x = x^b M$.
Equation~\eqref{eq:mahler-eqn} then rewrites as $L y = 0$ where
\begin{equation*}\tag{\ref{eq:mahler-opr}}
L = \ell_r M^r + \dots + \ell_0
\end{equation*}
in the algebra generated by $M$ and~$x$.
We are interested in the
algebraic complexity of computing the kernel of~$L$ in each of
$\bK[x]$, $\bK(x)$, $\bK[[x]]$, and $\Puiseux$.

\begin{table}
\begin{tabular}{lll}
\toprule
Kind of solutions & Algorithm & Complexity \\
\midrule
$\bK[[x]]$, to order $\lfloor \nu \rfloor + 1$
  & Alg.~\ref{algo:solve-singular-part}
  & $\bigO(rd v_0^2 + r^2 \Mult(v_0))$ \\
$\bK[[x]]$, to order $n$, when $r = \bigO(d)$
  & Alg.~\ref{algo:solve-nonsingular-part}
  & $\bigO(r d^3 + n r d)$ \\
$\bK[x]$
  & Alg.~\ref{algo:poly-sols-basis}
  & $\softO(b^{-r} d^2 + \Mult(d))$ \\
$\bK ((x^{1/N}))$
  & Alg.~\ref{algo:PuiseuxSolutions}
  & $\softO(r^2N d(d+n))$ \\
$\Puiseux$
  & Alg.~\ref{algo:PuiseuxSolutions}
  & $\softO(r^2b^r d(d+n))$ \\
$\bK(x)$, when $b = 2$
  & Alg.~\ref{algo:rat-sol}
  & $\softO(d^3)$ \\
$\bK(x)$, when $b \geq 3$
  & Alg.~\ref{algo:rat-sol}
  & $\softO(b^{-r}d^2)$ \\
\bottomrule
\end{tabular}
\caption{Complexity of the solving algorithms presented in the paper,
assuming $\ell_0 \neq 0$.}
\end{table}

We always assume that $\ell_r$~is nonzero.
Except where otherwise noted, we also assume $\ell_0 \neq 0$.
From a decidability viewpoint,
the latter assumption is no loss of generality
thanks to the following result~\cite[Cor.~6, p.~36]{Dumas-1993-RMS}.

\begin{prop} \label{prop:nonzero-identity-term}
Given a linear Mahler equation of the form~\eqref{eq:mahler-eqn},
one can compute an equation of the same form, with $\ell_0 \neq 0$,
that has exactly the same formal Laurent series solutions%
---and therefore, the same polynomial solutions
and the same rational-function solutions.
\end{prop}

Note however that this result does not say anything about the cost of
reducing to the case~$\ell_0 \neq 0$.
We give a complexity bound for this step in~\S\ref{sec:ell0neq0}.
As it turns out, this bound often dominates our complexity estimates for the
actual solving algorithms.
Let us therefore stress that all other complexity results are stated under the
assumption that $\ell_0$~is nonzero.

For $0 \leq k \leq r$, we denote by $v_k \in \bN \cup \{+\infty\}$
and $d_k \in \bN \cup \{-\infty\}$
the valuation and degree of the coefficient~$\ell_k$.
Let $d \geq \max_{0\leq k \leq r} d_k$.
Polynomials are implicitly represented in dense form, so that
polynomials of degree~$d$ in~$\bK[x]$ have size~$d+1$.
All complexity estimates are given in terms of arithmetical operations
in~$\bK$, which we denote “ops”.
The complexity of multiplying two polynomials of degree at most~$n$ is
denoted by~$\Mult(n)$; we make the standard assumptions that
$\Mult(n) = \bigO(n^2)$ and that $n \mapsto \Mult(n)/n$ is nondecreasing.

Given two integers or polynomials $a$ and~$b$, we denote their gcd
by~$a\wedge b$ and their lcm by~$a\vee b$;
we use $\bigwedge$ and~$\bigvee$ for $n$ary forms.

The following identities are used repeatedly in the text.
We gather and repeat them here for easier reference:
\begin{gather}
\label{eq:mahler-eqn}\tag{\text{\sc eqn}}
\ell_r (x) y (x^{b^r}) + \cdots + \ell_1 (x) y (x^b) + \ell_0 (x) y (x) = 0 , \\
\label{eq:mahler-opr}\tag{\text{\sc opr}}
L = \ell_r M^r + \dots + \ell_0 , \\
\label{eq:def-sizes}\tag{\text{\sc mu-nu}}
  \nu = \max_{k \geq 1} \frac{v_0 - v_k}{b^k - 1} ,
  \qquad
  \mu = v_0 + \nu.
\end{gather}

\subsection{General strategy and outline}

The article is organized as follows.
In \S\ref{sec:series}, we develop algorithms to compute truncated series solutions of equations of the form~\eqref{eq:mahler-eqn}.
We start with an example that illustrates the structure of the solution space and some of the main ideas behind our algorithms~(\S\ref{sec:ex}).
Then, we introduce a notion of Newton polygons, and use it to prove that the
possible valuations (resp.\ degrees) of the solutions of~\eqref{eq:mahler-eqn}
in~$\Puiseux$ (resp.~$\bK[x]$) belong to a finite set
that we make explicit~(\S\ref{sec:structure}).
We compute a suitable number of initial coefficients by solving a linear
system~(\S\ref{sec:approximate-series-solutions}),
then prove that the following ones can be obtained iteratively in linear time,
and apply these results to give a procedure that computes a complete set of
truncated series solutions~(\S\ref{sec:series-sols}).
Finally, we extend the same ideas to the case of solutions in~$\bK[x]$
(\S\ref{sec:poly}) and in $\Puiseux$ (\S\ref{sec:Puiseux}).

The next section, \S\ref{sec:rat-sols}, deals with solutions in~$\bK(x)$.
The general idea is to first obtain a denominator bound, that is a
polynomial~$q$ such that $Lu=0$ with $u\in \bK(x)$ implies $qu\in \bK[x]$ (\S\ref{sec:den-intro}).
Based on elementary properties of the action of~$M$ on elements of~$\bK[x]$
(\S\ref{sec:MG}), we give several algorithms for computing such bounds
(\S\ref{sec:den-algo}--\S\ref{sec:den-alt}).
This reduces the problem to computing a set of polynomial solutions
with certain degree constraints,
which can be solved efficiently
using the primitives developed in~\S\ref{sec:series},
leading to an algorithm for solving linear Mahler equations
in~$\bK(x)$ (\S\ref{sec:num}).
We briefly comment on a comparison, in terms of complexity, of Bell and Coons'
transcendence test and the approach by solving the Mahler equation for rational
functions (\S\ref{sec:transcendence}).
The net result is that the new approach is faster.

Finally, in~\S\ref{sec:ell0neq0}, we generalize our study
to the situation where the coefficient~$\ell_0$
in~\eqref{eq:mahler-eqn}
is zero.
This makes us develop an unexpected algorithm for computing the gcrd
of a family of operators,
which we analyze and compare to the more traditional approach
via Sylvester matrices and subresultants.

\subsection{Acknowledgment}
The authors are indebted to Alin Bostan for helpful discussions
and for pointing us to the work of Grigor'ev~\cite{Grigoriev-1990-CIT}.

\section{Polynomial and series solutions}
\label{sec:series}

\subsection{A worked example}\label{sec:ex}

The aim of this section is to illustrate our solving strategy in $\bK[[x]]$ and $\Puiseux$
on an example that we treat straightforwardly.

In radix $b = 3$, consider the equation $L y = 0$ where
\begin{multline}\label{eq:running-example}
L = x^3(1-x^3+x^6)(1-x^7-x^{10}) \, M^2 \\
-  (1 - x^{28} - x^{31} - x^{37} - x^{40}) \, M
+ x^6(1+x)(1-x^{21}-x^{30}).
\end{multline}
Assume that $y\in \Puiseux$ is a solution whose valuation is a rational number~$v$.
The valuations of $\ell_k M^ky$, for $k = 0,1,2$,
are respectively equal to $6+v, 3v, 3+9v$.
If one of these rational numbers was less than the other two,
then the valuation of the sum $\sum_{k=0}^2 \ell_k M^ky$
would be this smaller number,
and $L y$~could not be zero.
Consequently, at least two of the three rational numbers $6+v, 3v, 3+9v$
have to be equal to their minimum.
After solving, we find $v\in \{-1/2,3 \}$.

First consider the case $v=3$, and write $y=\sum_{n\geq 3} y_nx^n$.
For~$m$ from~10 to~15, extracting the coefficients of~$x^m$
from both sides of $0 = \ell_0y+\ell_1My+\ell_2M^2y$,
we find that $y_3,\dots,y_9$ satisfy
\begin{equation}\label{eq:unroll-example}
\begin{array}{c@{\ }c@{\ }c@{\ }c@{\ }c@{\ }c@{\ }c@{\ }c@{\ }c}
0 &= &y_3 &{}+y_4 , &&&&& \\
0 &= &&\phantom{{}+{}}y_4 &{}+y_5 , &&&& \\
0 &= &&{}-y_4 &{}+y_5 &{}+y_6 , &&& \\
0 &= &&&&\phantom{{}+{}}y_6 &{}+y_7 , && \\
0 &= &&&&&\phantom{{}+{}}y_7 &{}+y_8 , & \\
0 &= &&&{}-y_5&&&{}+y_8 &{}+y_9 .
\end{array}
\end{equation}

More generally, extracting the coefficient of~$x^m$ yields the relation
\begin{multline}\label{eq:beyond-9}
  \bigl( y_{m-6} + y_{m-7} - y_{m-27} - y_{m-28} - y_{m-36} - y_{m-37} \bigr) \\
  - \bigl( y_{\frac{m}{3}} - y_{\frac{m-28}{3}} - y_{\frac{m-31}{3}} - y_{\frac{m-37}{3}} - y _{\frac{m-40}{3}} \bigr) \\
  + \bigl( y_{\frac{m-3}{9}} - y_{\frac{m-6}{9}} + y_{\frac{m-9}{9}} - y_{\frac{m-10}{9}} - y_{\frac{m-19}{9}} \bigr) = 0 ,
\end{multline}
where $y_s$~is understood to be zero
if the rational number~$s$ is not a nonnegative integer.
This equation takes different forms, depending on the residue of~$m$ modulo~9:
for example, for $m = 20$ and $m = 42$, it reduces to, respectively,
\[
  y_{14} + y_{13} = 0, \qquad
  y_{36}+y_{35}-y_{15}-2y_{14}-y_6-y_5-y_4 = 0.
\]
Despite these variations, for any~$m\geq10$ the index~$n = m-6$
is the largest integer index occurring in~\eqref{eq:beyond-9}.
It follows that for successive $m\geq 10$, we can iteratively obtain~$y_n$
from~\eqref{eq:beyond-9} in terms of already known coefficients of the series.
Conversely, any sequence $(y_n)_{n\geq 3}$ that satisfies~\eqref{eq:beyond-9}
gives a solution $y=\sum_{n\geq 3} y_nx^n$ of \eqref{eq:mahler-eqn}.

As a consequence,
the power series solution is entirely determined by the choice of~$y_3$
and the space of solutions of~\eqref{eq:mahler-eqn}
in~$\bK[[x]]$ has dimension one.
A basis consists of the single series
\begin{equation}\label{eq:power-series-sol}
  x^3-x^{4}+x^5-2x^6+2x^7-2x^8+3x^9-3
  x^{10}+3x^{11}-5x^{12}+ \dotsb.
\end{equation}

The other possible valuation, $v=-1/2$, is not a natural number.
To revert to the simpler situation of the previous case,
we perform the change of variables $x = t^2$
followed by the change of unknowns $y(t) = \tilde y(t) / t$.
The equation becomes $\tilde L \tilde y = 0$ with
\begin{multline}\label{eq:running-example-transfd}
  \tilde L =
    (1 - t^6 + t^{12}) (1 - t^{14} - t^{20}) \, M^2 \\
    - (1 - t^{56} - t^{62} - t^{74} - t^{80}) \, M
    + t^{14}(1 + t^2)(1 - t^{42} - t^{60}) .
\end{multline}
To understand this calculation,
remember that $M$~was defined on~$\Puiseux$,
so that $M(t) = M(x^{1/2}) = x^{3/2} = t^3$.

We now expect~$\tilde L$ to have
solutions~$\tilde y = \sum_{n\geq0} \tilde y_n t^n$
of valuation 0 and~7 with respect to~$t$,
and the solutions of~$\tilde L$ with valuation~0
to correspond to the solutions of~$L$ with valuation~$-1/2$.
Extracting the coefficients of~$x^m$ for~$m$ from~0 to~24
from both sides of $\tilde L \tilde y = 0$ and skipping tautologies,
we find that $\tilde{y}_0,\dots,\tilde{y}_{10}$ satisfy
\begin{equation*}
\begin{array}{c@{\ }c@{\ }c@{\ }c@{\ }c@{\ }c@{\ }c@{\ }c@{\ }c@{\ }c@{\ }c@{\ }c@{\ }c}
0 &= &&{}-\tilde{y}_1, &&&&&&&&& \\
0 &= &{}-\tilde{y}_0&&{}-\tilde{y}_2, &&&&&&&& \\
0 &= &&\phantom{{}+{}}\tilde{y}_1&&{}-\tilde{y}_3, &&&&&&& \\
0 &= &\phantom{{}+{}}\tilde{y}_0&&&&{}-\tilde{y}_4, &&&&&& \\
0 &= &&&&&&{}-\tilde{y}_5, &&&&& \\
0 &= &\phantom{{}+{}}\tilde{y}_0&&{}+\tilde{y}_2, &&&&&&&& \\
0 &= &&\phantom{{}+{}}\tilde{y}_1&&{}+\tilde{y}_3, &&&&&&& \\
0 &= &&&\phantom{\,\,\,}2\tilde{y}_2&&{}+\tilde{y}_4&&{}-\tilde{y}_6, &&&& \\
0 &= &&&&\phantom{{}+{}}\tilde{y}_3&&{}+\tilde{y}_5, &&&&& \\
0 &= &&&&&\phantom{{}+{}}\tilde{y}_4&&{}+\tilde{y}_6, && \\
0 &= &&\phantom{{}+{}}\tilde{y}_1&&&&{}+\tilde{y}_5, &&&&& \\
0 &= &&&&&&&\phantom{{}+{}}\tilde{y}_6&&{}+\tilde{y}_8, && \\
0 &= &&{}-\tilde{y}_1&&&&&&{}+\tilde{y}_7&&{}+\tilde{y}_9, & \\
0 &=&&&{}-\tilde{y}_2&&&&&&&&{}+\tilde{y}_{10}.
\end{array}
\end{equation*}
Reasoning as above, we derive that,
given~$\tilde y_0$ while enforcing~$\tilde y_7=0$,
there is exactly one power series solution to~$\tilde L$.
More specifically when~$\tilde{y}_0 = 1$ and $\tilde{y}_7=0$, we find the series
\[
  1-t^2+t^4-t^6+t^8-t^{10}+t^{12} + \dotsb.
\]

Hence, there is a 2-dimensional solution space in~$\Puiseux$
for the original equation~\eqref{eq:mahler-eqn},
with a basis consisting of the power series~\eqref{eq:power-series-sol}
and the additional Puiseux series
\[
  x^{-1/2}-x^{1/2}+x^{3/2}-x^{5/2}+x^{7/2}-x^{9/2}+x^{11/2} + \dotsb.
\]

\subsection{Valuations and degrees}
\label{sec:structure}

Let us assume that~$y \in \Puiseux$ is a solution of~\eqref{eq:mahler-eqn},
whose valuation is a rational number~$v$.
The valuation of the term $\ell_k M^k y$ is then $v_k + b^k v$.
Among those expressions, at least two must be minimal
to permit the left-hand side of~\eqref{eq:mahler-eqn} to be~0:
therefore, there exist distinct indices $k_1,k_2$ between 0 and~$r$ such that
\begin{equation}\label{eq:NewtonCondition}
    v_{k_1} + b^{k_1} v
  = v_{k_2} + b^{k_2} v
  = \min_{0 \leq k \leq r} v_k + b^k v.
\end{equation}

This necessary condition for~$L y = 0$ can be interpreted
using a \emph{Newton polygon}
analogous to that of algebraic equations~\cite[Sec.~IV.3.2-3]{Walker-1978-AC}:
to each monomial $x^j M^k$ in~$L$, we associate the point~$(b^k,j)$ in the
first quadrant of the Cartesian plane
endowed with coordinates $U$ and~$V$
(see Fig.~\ref{fig:Newt}).
We call the collection of these points the \emph{Newton diagram} of~$L$,
and the lower (resp.\ upper) boundary of its convex hull the
\emph{lower \emph{(resp.\ \emph{upper})} Newton polygon} of~$L$.
That two integers $k_1, k_2$ satisfy~\eqref{eq:NewtonCondition}
exactly means that $(b^{k_1}, v_{k_1})$ and $(b^{k_2}, v_{k_2})$ belong to an
edge~$E$ of slope~$-v$ of the corresponding lower Newton polygon.

Given an edge~$E$ as above, an arithmetic necessary condition holds
in addition to the geometric one just mentioned:
the coefficients of the monomials of~$L$
associated to points of~$E$ must add up to zero.
We call an edge with this property \emph{admissible}.

\begin{ex}
  The lower Newton polygon of the operator~\eqref{eq:running-example}
  appears in dashed lines in Figure~\ref{fig:Newt}. It contains two
  admissible edges, corresponding to the valuations~$3$ and~$-1/2$.
\end{ex}

We get the following criterion, already stated in~\cite[p.~51]{Dumas-1993-RMS}
with a slightly different proof.

\begin{lem} \label{lem:valuation-at-0}
Let $L$ be defined as in~\eqref{eq:mahler-opr}.
The valuation~$v$ of any formal Puiseux series solution of~\eqref{eq:mahler-eqn}
is the opposite of the slope
of an admissible edge of the lower Newton polygon of~$L$.
It satisfies
\[
     -\frac{v_r} {b^{r-1}(b-1)}
     \leq v
     = - \frac{v_{k_1} - v_{k_2}}{b^{k_1} - b^{k_2}}
     \leq \frac{v_0}{b-1},
\]
where $(b^{k_1}, v_{k_1})$ and $(b^{k_2}, v_{k_2})$ are the endpoints of
the implied edge.
\end{lem}

\begin{proof}
The fact that $v$~is the opposite of a slope
together with its explicit form
follow from~\eqref{eq:NewtonCondition}
and the discussion above.
There remains to prove the upper and lower bounds.
The leftmost edge of the lower Newton polygon of~$L$ provides the largest
valuation and its slope $(v_k - v_0)/(b^k - 1)$ for some $k \geq 1$
is bounded below by $- v_0/(b - 1)$.
In the same way, the rightmost edge provides the smallest valuation and
its slope, of the form $(v_r - v_k)/(b^r - b^k)$ for some $k < r$, is
bounded above by $v_r/(b^r - b^{r-1})$.
\end{proof}

\begin{figure}
\begin{tikzpicture}[baseline=(current bounding box.center),
  x={(8mm,0mm)},y={(0mm,2.2mm)}]
\fill[fill=black!5!white] (3,10) (1,6)--(3,0)--(9,3)--(9,19)--(3,40)--(1,37);
\draw[->](0,0)--(10,0);
\draw(10,0) node[above]{$U$};
\foreach \x in {0,...,9} \draw(\x,-.2)--(\x,.2);
\foreach \x in {0,...,9} \draw(\x,0) node[below,font=\footnotesize]{\x};
\draw[->](0,0)--(0,41);
\draw(0,41) node[right]{$V$};
\foreach \y in {0,5,...,40} \draw(-.05,\y)--(.05,\y);
\foreach \y in {0,10,...,40} \draw(0,\y) node[left,font=\footnotesize]{\y};
\draw(-.03,9)--(.03,9);
\draw(0,9) node[right]{$\mu=9$};
\fill(1,6) circle[radius=2pt];
\fill(1,7) circle[radius=1pt];
\fill(1,27) circle[radius=1pt];
\fill(1,28) circle[radius=1pt];
\fill(1,36) circle[radius=1pt];
\fill(1,37) circle[radius=1pt];
\fill(3,0) circle[radius=1pt];
\fill(3,28) circle[radius=1pt];
\fill(3,31) circle[radius=1pt];
\fill(3,37) circle[radius=1pt];
\fill(3,40) circle[radius=1pt];
\fill(9,3) circle[radius=1pt];
\fill(9,6) circle[radius=1pt];
\fill(9,9) circle[radius=1pt];
\fill(9,10) circle[radius=1pt];
\fill(9,13) circle[radius=1pt];
\fill(9,16) circle[radius=1pt];
\fill(9,19) circle[radius=1pt];
\draw[style=dashed,-] (1,6)--(3,0)--(9,3);
\draw[style=densely dotted,-] (1,37)--(3,40)--(9,19);
\draw[style=loosely dotted,-] (1,6)--(0,9);
\draw(1,6) node[right]{$P_0=(1,v_0)$};
\path (0,9)
  -- node [sloped, below, font=\footnotesize]{$\Lambda:V=-\nu U+\mu$} (3, 0);
\end{tikzpicture}
\caption{The Newton diagram of the equation treated in~\S\ref{sec:ex}
  for radix~$b = 3$,
  with corresponding lower Newton polygon (dashed line)
  and upper Newton polygon (dotted line).}
\label{fig:Newt}
\end{figure}
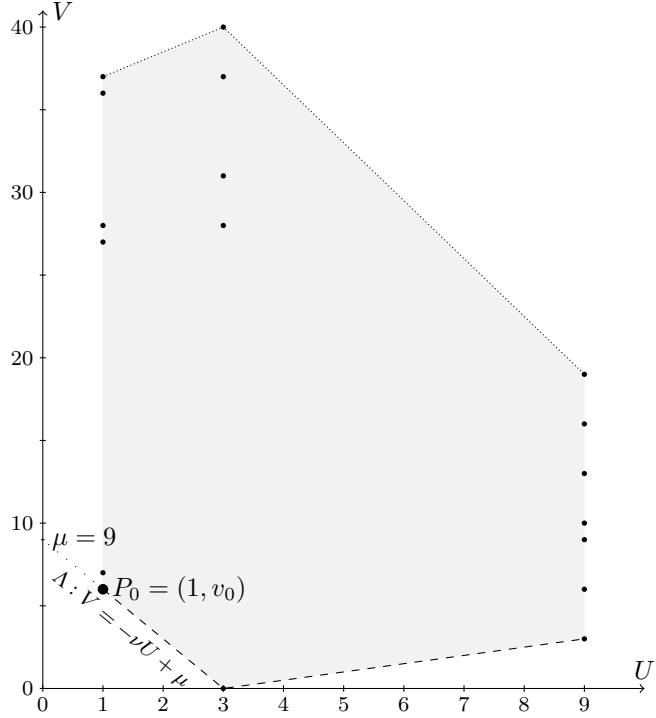

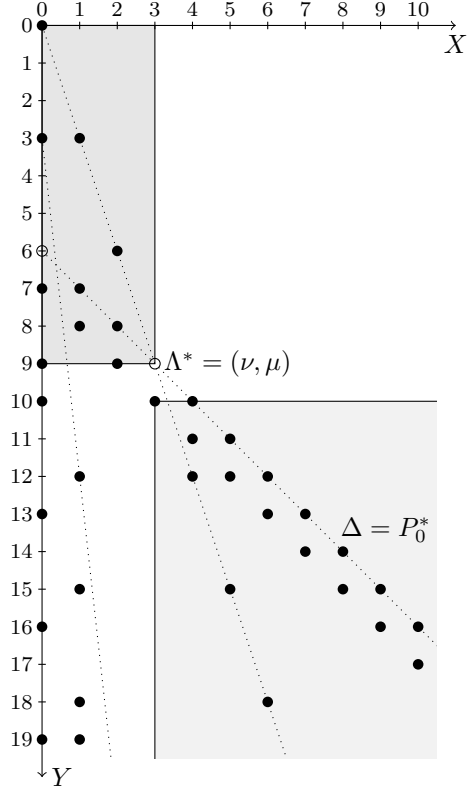
\begin{figure}
\begin{tikzpicture}[baseline=(current bounding box.center),
  x={(5mm,0mm)},y={(0mm,-5mm)}]
\draw[fill=black!10!white] (0,0) rectangle (3,9);
\begin{scope}
  \clip (0,0) rectangle(10.5,19.5);
  \draw[fill=black!5!white] (3,10) rectangle (11,20);
  \draw[style=dotted,-] (0,0)--(7,21);
  \draw[style=dotted,-] (0,3)--(2,21);
  \draw[style=dotted,-] (0,6)--(11,17);
\end{scope}
\draw[->](0,0)--(11,0);
\draw(11,0) node[below]{$X$};
\foreach \x in {0,...,10} \draw(\x,-.1)--(\x,.1);
\foreach \x in {0,...,10} \draw(\x,0) node[above,font=\footnotesize]{\x};
\draw[->](0,0)--(0,20);
\draw(0,20) node[right]{$Y$};
\foreach \y in {0,...,19} \draw(-.1,\y)--(.1,\y);
\foreach \y in {0,...,19} \draw(0,\y) node[left,font=\footnotesize]{\y};
\fill(0,0) circle[radius=2pt];
\fill(0,3) circle[radius=2pt];
\fill(1,3) circle[radius=2pt];
\fill(2,6) circle[radius=2pt];
\fill(0,7) circle[radius=2pt];
\fill(1,7) circle[radius=2pt];
\fill(1,8) circle[radius=2pt];
\fill(2,8) circle[radius=2pt];
\fill(0,9) circle[radius=2pt];
\fill(2,9) circle[radius=2pt];
\fill(0,10) circle[radius=2pt];
\fill(3,10) circle[radius=2pt];
\fill(4,10) circle[radius=2pt];
\fill(4,11) circle[radius=2pt];
\fill(5,11) circle[radius=2pt];
\fill(1,12) circle[radius=2pt];
\fill(4,12) circle[radius=2pt];
\fill(5,12) circle[radius=2pt];
\fill(6,12) circle[radius=2pt];
\fill(0,13) circle[radius=2pt];
\fill(6,13) circle[radius=2pt];
\fill(7,13) circle[radius=2pt];
\fill(7,14) circle[radius=2pt];
\fill(8,14) circle[radius=2pt];
\fill(1,15) circle[radius=2pt];
\fill(5,15) circle[radius=2pt];
\fill(8,15) circle[radius=2pt];
\fill(9,15) circle[radius=2pt];
\fill(0,16) circle[radius=2pt];
\fill(9,16) circle[radius=2pt];
\fill(10,16) circle[radius=2pt];
\fill(10,17) circle[radius=2pt];
\fill(1,18) circle[radius=2pt];
\fill(6,18) circle[radius=2pt];
\fill(0,19) circle[radius=2pt];
\fill(1,19) circle[radius=2pt];
\draw(0,6) circle[radius=2pt];
\draw(3,9) circle[radius=2pt];
\draw(8,14) node[above right, xshift=-1ex]{$\Delta=P_0^*$};
\draw(3,9) node[right]{$\Lambda^*=(\nu,\mu)$};
\end{tikzpicture}
\caption{The infinite matrix~$\IM$ corresponding
  to the example treated in~\S\ref{sec:ex}:
  solid circles denote nonzero entries,
  hollow circles denote recombinations to zero.}
\label{fig:inf-mat}
\end{figure}

\begin{prop}\label{prop:esp-sol}
  The dimension of the space of solutions of the homogeneous equation
  $L y = 0$ in~$\Puiseux$ is bounded by the order~$r$ of~$L$.
\end{prop}

\begin{proof}
The space of solutions admits a basis consisting of Puiseux series
with pairwise distinct valuations.
The number of possible valuations is bounded by the edge count of the
lower Newton polygon of~$L$, which is at most~$r$.
\end{proof}

\begin{rem}
As we will see, the dimension of the solutions in $\Puiseux$ can be strictly
less than~$r$.
It is natural to ask how to construct a ``full'' system of $r$~linearly
independent formal solutions in some larger extension of~$\bK(x)$.
We will not pursue this question here and point to Roques's work
for an answer;
see \cite[Lemma~20 and Thm~35]{Roques-2018-ARB}
and \cite[Theorem~1]{Roques-2016}.
See also Remark~\ref{rem:non-Puiseux} below.
\end{rem}

In analogy with the previous discussion on valuations of solutions,
if a Puiseux series solution of~\eqref{eq:mahler-eqn} involves monomials with
maximal exponent~$\delta$, then the expression $d_k + b^k \delta$ must
reach its maximum at least twice as $k$ ranges from~0 to~$r$.
As we see by the same reasoning as above (or by changing $x$~to~$1/x$,
which exchanges the lower and upper Newton polygons),
$-\delta$ is then one of the slopes of the upper Newton polygon of~$L$.
The largest possible value corresponds to the rightmost edge.

\begin{lem}\label{lem:admissible-degrees}
The maximum exponent~$\delta$ of a monomial in a finite Puiseux series
solution, and in particular the degree of a polynomial solution,
is the opposite of the slope of an admissible edge of the upper Newton
polygon.
It satisfies
\[
  \delta = - \frac{d_{k_1} - d_{k_2}}{b^{k_1} - b^{k_2}} \leq \frac{d}{b^{r-1} (b-1)},
\]
for some ${k_1} \neq {k_2}$.
\end{lem}

The admissibility of an edge of the upper Newton polygon is defined
in analogy with admissibility in the lower Newton polygon.

\subsection{The nonhomogeneous case}\label{sec:non-hom}

One of the proofs of results about Puiseux series solutions
in~\S\ref{sec:Puiseux} makes use of extended Newton diagrams that take into
account the right-hand side of nonhomogeneous equations.

For $L$ as in~\eqref{eq:mahler-opr} and
a Puiseux series~$\ell_{-\infty}$ of valuation
$v_{-\infty} \in \bQ \cup \{+\infty\}$,
consider the nonhomogeneous equation
\begin{equation}\label{eq:mahler-eqn'}
\ell_r (x) y (x^{b^r}) + \cdots + \ell_1 (x) y (x^b) + \ell_0 (x) y (x) =
  \ell_{-\infty}(x) .
\end{equation}
Given a Puiseux series solution $y \in \Puiseux$ of this equation,
with valuation $v \in \bQ$,
we define the Newton diagram of~$(L, \ell_{-\infty})$
as the Newton diagram of~$L$,
augmented with all points~$(0, \alpha)$
for which $x^\alpha$~appears with nonzero coefficient in~$\ell_{-\infty}$.
The notion of lower Newton polygon extends correspondingly.

As in~\S\ref{sec:structure}, these definitions are motivated by analyzing
the minimum of the valuations~$v_k + b^k v$ of the terms
of the left-hand side of~\eqref{eq:mahler-eqn'}:
either this minimum is equal to~$v_{-\infty}$,
or it is less than~$v_{-\infty}$
and must be reached as least twice on the left-hand side.
In both cases, making the convention that $b^{-\infty} = 0$,
there exist distinct indices $k_1, k_2$,
now in $\{-\infty, 0, 1, \dots, r \}$,
such that the analogue
\begin{equation*}
    v_{k_1} + b^{k_1} v
  = v_{k_2} + b^{k_2} v
  = \min_{k \in \{-\infty, 0, 1, \dots, r\}} v_k + b^k v
\end{equation*}
of~\eqref{eq:NewtonCondition}~holds.
Again, this exactly means that $(b^{k_1}, v_{k_1})$ and $(b^{k_2}, v_{k_2})$
belong to an edge~$E$ of slope~$-v$ of the lower Newton polygon,
now of~$(L, \ell_{-\infty})$.

Depending on $v_{-\infty}$ and~$\hat v = \min_{0 \leq k \leq r} (v_k + b^k v)$,
the lower Newton polygon of~$(L, \ell_{-\infty})$ can:
be equal to that of~$L$,
if~$\ell_{-\infty} = 0$;
add an edge to its left,
if~$v_{-\infty} > \hat v$;
prolong its leftmost edge,
if~$v_{-\infty} = \hat v$;
or replace some of its leftmost edges,
if~$v_{-\infty} < \hat v$.
We defined the admissibility of an edge~$E$ of the lower Newton polygon of~$L$
in terms of the coefficients of those monomials~$x^{v_k} M^k$ in~$L$
associated to points on~$E$.
We extend the definition to edges of the lower Newton polygon
of~$(L, \ell_{-\infty})$ by the convention that,
if a point has to be considered for $k = -\infty$,
the corresponding coefficient
is the opposite of the coefficient of~$x^{v_{-\infty}}$ in~$\ell_{-\infty}$.
Admissibility is again a necessary condition
for $v$ to be a possible valuation of a solution of~\eqref{eq:mahler-eqn'}.

\subsection{Approximate series solutions}
\label{sec:approximate-series-solutions}

We now concentrate on the search for power series solutions
$y(x) = y_0 + y_1 x + \dotsb \in \bK[[x]]$
of~\eqref{eq:mahler-eqn}.
Extracting the coefficient of~$x^m$ in both
sides of it yields a linear equation for the
coefficients~$y_n$.
This linear equation can be viewed as a row, denoted~$\IM_m$, of an infinite
matrix~$\IM = \IM(L)$.

The matrix~$\IM$ consists of overlapping strips
with different slopes.
We view its row and column indices, starting at~0, as
continuous variables $Y$ and $X$ with the $Y$-axis oriented downwards.
Each nonzero term~$\ell_k(x)M^k$ then corresponds to matrix entries in the
strip
$b^k X + v_k \leq Y \leq b^k X + d_k$.
By definition of $v_k$ and~$d_k$, the entries lying on the lines
$Y = b^k X + d_k$ and $Y = b^k X + v_k$ that delimit the strip are
nonzero, except maybe at intersection points of such lines
(obtained for different~$k$).
Because of our assumption that $\ell_0$ is nonzero, the
smallest slope is~1, obtained for $k = 0$.

For large~$Y$, the line $Y = X + v_0$ becomes the topmost one, and each
row~$\IM_m$ determines a new coefficient~$y_n$ uniquely, for~$n=m-v_0$.
Thus, the power series solutions are characterized by a finite subsystem
of~$\IM$.
In order to state this fact more precisely
in Proposition~\ref{prop:system-size} below, define
\begin{equation*}\tag{\ref{eq:def-sizes}}
  \nu = \max_{k \geq 1} \frac{v_0 - v_k}{b^k - 1} ,
  \qquad
  \mu = v_0 + \nu.
\end{equation*}
In terms of the Newton diagram, $\nu$~and~$\mu$ are, respectively, the
opposite of the slope and the $V$-intercept of the leftmost edge of the
lower Newton polygon.
Note that, as we can deduce from the proof of Lemma~\ref{lem:valuation-at-0},
there is no nonzero power series
solution when~$\nu < 0$,
which happens if and only if $v_0$~is a strict minimum of all the~$v_k$
over~$0 \leq k \leq r$.

\begin{prop} \label{prop:system-size}
Assume that $\nu \geq 0$. A vector
$(y_0, \dots, y_{\lfloor \nu \rfloor})$
is a vector of initial coefficients of a formal power series solution
\begin{equation} \label{eq:series-from-ini}
y = y_0 + \dots + y_{\lfloor \nu \rfloor} x^{\lfloor \nu \rfloor}
+ y_{\lfloor \nu \rfloor + 1} x^{\lfloor \nu \rfloor + 1} + \dotsb
\end{equation}
of~\eqref{eq:mahler-eqn}
if and only if it satisfies the linear system given by the upper left
$(\lfloor \mu \rfloor + 1) \times (\lfloor \nu \rfloor + 1)$
submatrix of~$\IM$.
The power series solution~\eqref{eq:series-from-ini} extending
$(y_0, \dots, y_{\lfloor \nu \rfloor})$
is then unique.
\end{prop}

\begin{proof}
A series $y = y_0 + y_1 x + \dotsb$ is a solution if and only if its
coefficients satisfy the system~$(\IM_m)_{m \geq 0}$.
Whenever
\begin{equation}\label{eq:ConditionOn-n}
  v_0 + n < v_1 + b^1 n, \quad \dots, \quad v_0 + n < v_r + b^r n,
\end{equation}
the row~$\IM_{v_0 + n}$ of~$\IM$ is the first one with a nonzero entry of
index~$n$.
It then determines~$y_n$ in terms of $y_0, \dots, y_{n-1}$.
Condition~\eqref{eq:ConditionOn-n} is equivalent to $n > \nu$, hence,
for any given $(y_n)_{0 \leq n \leq \nu}$, there is a unique choice of
$(y_n)_{n > \nu}$ satisfying all the equations~$\IM_m$ for $m > v_0 + \nu = \mu$.
As, when \eqref{eq:ConditionOn-n}~holds for a given~$n$, the entries of index~$n$
of~$\IM_{m}$ with $m < v_0 + n$ are zero, the remaining equations
$(\IM_m)_{0 \leq m \leq \mu}$ only involve the unknowns $(y_n)_{0 \leq n \leq \nu}$.
\end{proof}

We note in passing the following corollary,
which is the essential argument
in the proof of~\cite[Theorem~22]{Roques-2018-ARB}.

\begin{cor}\label{coro}
In case the leftmost edge of the lower Newton polygon of~$L$
lies on the axis of abscissas and is admissible,
Equation~\eqref{eq:mahler-eqn} admits a power series solution of valuation~0.
\end{cor}

\begin{proof}
We then have $\nu = \mu = 0$, so the only condition to check is that the
first entry of~$\IM_0$ is zero.
This is equivalent to the edge being admissible.
\end{proof}

The geometric interpretation of the quantities $\mu$~and~$\nu$ defined
by~\eqref{eq:def-sizes} is a special case of a general correspondence
between the structure of the matrix~$\IM$ and the Newton diagram of~$L$
via the point-line duality of plane projective geometry.
The correspondence stems from the fact that a monomial~$x^j M^k$ of~$L$
is associated both to a point $(b^k,j)$ in the Newton diagram and, by
considering its action on powers of~$x$, to the entries of~$\IM$ lying on
the line $Y = b^k X + j$.
More generally, under projective duality, each point~$(U,V)$ in the
plane of the Newton diagram corresponds to a line $Y = U X + V$ in the
plane of the matrix~$\IM$, while, conversely, the dual of a point~$(X,Y)$
is the line $V = -XU + Y$.
A line through two points $(U_1, V_1)$ and $(U_2, V_2)$
corresponds to the intersection of their duals.

In particular, the point $P_0 = (1, v_0)$ corresponds to the right boundary
$\Delta : Y = X + v_0$ of the strip of entries of slope~$1$ in the
matrix~$\IM$ (see Figures \ref{fig:Newt} and~\ref{fig:inf-mat}).
In the $(U,V)$-plane,
the line containing the leftmost edge of the lower Newton polygon
passes through that point~$P_0 = \Delta^*$.
This line is $\Lambda : V = -\nu U + \mu$ and
corresponds to the bottommost intersection $\Lambda^* = (\nu, \mu)$
of~$\Delta$ with the right boundary of another strip.
Below this intersection, the entries of~$\IM$ lying on~$\Delta$ are the
topmost nonzero entries of their respective columns, and, at the same time,
the rightmost nonzero entries of their respective rows: as already
observed, each row~$\IM_m$ then determines a new~$y_n$.

\begin{ex}\label{example:TheExample-2}
For the operator~$L$ of~\S\ref{sec:ex}, the right
boundaries of the strips associated to the three terms of~$L$ have
equations $Y = X + 6$, $Y = 3 X$, and~$Y = 9 X + 3$ respectively
(dotted lines in Fig.~\ref{fig:inf-mat}).
The first two of them meet at $\Lambda^* = (3, 9)$ (Fig.~\ref{fig:inf-mat},
hollow circle at the bottom right corner of the gray rectangle),
and the line $\Delta : Y = X + 6$ becomes the rightmost line for~$Y > 9$.
For $m \geq 10$, the row~$\IM_m$ reflects the relation \eqref{eq:beyond-9}.
In particular, the existence of a power series solution is entirely determined by
the small linear system that uses the rows $\IM_0$~to~$\IM_9$
and the unknowns $y_0$~to~$y_3$
(gray rectangle on Figure~\ref{fig:inf-mat}).
Solving the system yields $y_0 = y_1 = y_2  = 0$ and
$y_3$~arbitrary.
We then recover the results of~\S\ref{sec:ex}:
the space of solutions of~\eqref{eq:mahler-eqn}
in~$\bK[[x]]$ has dimension one
and a basis consists of the single series~\eqref{eq:power-series-sol}.
The $V$-intercept of the leftmost edge of the lower Newton polygon
is~$\mu=9$, and the corresponding slope is~$-\nu = -3$.
In this case, it is both
the column dimension of the small system and the valuation of the solution.
Observe how the bottom right sector depicted in light gray
corresponds to the system starting with equations~\eqref{eq:unroll-example}:
as the top left rectangle imposes $y_0=y_1=y_2=0$,
the dots on the left of the sector in light gray play no role in the equations.
\end{ex}

As we will see,
in the situation of Proposition~\ref{prop:system-size},
the coefficients
$y_{\lfloor \nu \rfloor + 1}$ to $y_{\lfloor \nu \rfloor + n}$ of~$y$
can be computed from $y_0, \dots, y_{\lfloor \nu \rfloor}$ in $\bigO(n)$~ops for fixed~$L$.
This motivates to call the truncation to
order~$\bigO(x^{\lfloor \nu \rfloor + 1})$
of a series solution an \emph{approximate series solution}
of~\eqref{eq:mahler-eqn}.

\subsection{Power series solutions}
\label{sec:power-series-sols}
\label{sec:series-sols}

\begin{algo}
\inputs{
  The linear Mahler equation~\eqref{eq:mahler-eqn}.
  A transformation~$\phi$ of the form~\eqref{eq:change-var-1}.
  An integer~$w$.
  A set $E = \{m_0, m_1, \dots \}$ of row indices,
    with $m_0 < m_1 < \dots$.
}
\outputs{
  The submatrix $R_E = (R_{m,n})_{\substack{m \in E \hfill\\ 0 \leq n < w}}$ of
  the infinite matrix~$\IM(\phi(L))$.
}
\caption{\label{algo:build-mat}Matrix~$\IM$.}
\begin{enumerate}
\item Initialize a row-sparse $|E| \times w$ matrix~$R_E$.
\item \label{step:build-mat:outer-loop}
  For $i = 0, 1, \dots, |E|-1$ and $k = 0, 1, \dots, r$:
  \begin{enumerate}
  \item
    Set $B = m_i + \gamma - \alpha b^k$.
  \item \label{step:build-mat:starting-index}
    Compute $j'_0 = \beta^{-1} B \bmod{b^k}$
    (with $0 \leq j'_0 < b^k)$.
  \item \label{step:build-mat:inner-loop}
    For $j' = j'_0, j'_0 + b^k, j'_0 + 2 b^k, \dots$
    while $j' \leq d_k$ and $\beta j' \leq B$:
    \begin{enumerate}
      \item \label{step:build-mat:doit}
        If $\beta j' > B - b^k w$, then
        add $\ell_{k, j'}$ to the coefficient of index $(i, b^{-k}(B - \beta j'))$ of~$R_E$.
    \end{enumerate}
  \end{enumerate}
\item Return~$R_E$.
\end{enumerate}
\end{algo}

Our goal at this point is to describe an algorithm that computes the
formal power series solutions of~\eqref{eq:mahler-eqn}, truncated to any
specified order.
We first explain how to compute the entries of the matrix~$\IM$.
It is convenient, for expository reasons, to frame this computation as an
individual step that returns a sparse representation of a submatrix of~$\IM$
corresponding to a subset of the rows.
Indeed, in our complexity model dense matrices could not lead to good bounds.
We therefore define a matrix representation to be \emph{row-sparse}
if iterating over the nonzero entries of any given row
does not require any zero test in~$\bK$.
Then, the algorithm essentially amounts to an explicit expression for the
coefficients of recurrences similar~\eqref{eq:beyond-9}, which can as well be
computed on the fly.

In view of the computation of ramified solutions (\S\ref{sec:Puiseux}),
Algorithms \ref{algo:build-mat}~and~\ref{algo:solve-prescribed-part}
accept as input a $\bK$-linear transformation~$\phi$ to be applied to the operator~$L$.
In general, $\phi$~will take the form
\begin{equation}\label{eq:change-var-1}
\phi(x^j M^k)
  = x^{\alpha b^k + \beta j - \gamma} M^k ,
\qquad
\alpha, \gamma \in \bZ, \quad
\beta \in \bN_{> 0}, \quad
\beta \wedge b = 1,
\end{equation}
with $\alpha, \beta, \gamma$ chosen such that $\phi(L)$~has plain (as opposed to
Laurent) polynomial coefficients.
The reader only interested in polynomial, rational, and power series solutions
of~$L$ may safely assume~$\phi = \id$, i.e.,
 $\alpha=\gamma=0,\beta=1$.

\begin{lem}\label{lem:build-mat}
Algorithm~\ref{algo:build-mat} computes the submatrix $R_E$ obtained by taking
the first~$w$ entries of the rows of~$\IM(\phi(L))$ with index $m \in E$
in $\bigO\bigl((r + d) |E|\bigr)$~ops.
Each row of~$R_E$ has at most~$r + 2 d$ nonzero entries.
\end{lem}

\begin{proof}
Write $\tilde L = \sum_{k=0}^r \tilde\ell_k(x) M^k = \phi(L)$.
Recall that the row~$R_m$ is obtained by extracting the coefficient of~$x^m$
in the equality $\tilde L y = 0$, where $y = \sum_{n \geq 0} y_n x^n$.
More precisely, $R_{m,n}$~is the coefficient of~$y_n x^m$ in the
series
\[
  \tilde L y
  = \sum_{k=0}^r \sum_{j=0}^{d} \tilde \ell_{k,j} x^j
      \sum_{n=0}^{\infty} y_n x^{b^k n}
  = \sum_{m=0}^{\infty} \sum_{n=0}^{\infty}
      \Bigl( \sum_{j+b^kn = m} \tilde \ell_{k,j} \Bigr) y_n x^m.
\]
The definition of~$\phi$ translates into
$\tilde \ell_{k, j} = 0$ when
$j \not\equiv \alpha b^k - \gamma \pmod \beta$, and otherwise
$\tilde \ell_{k, j} = \ell_{k, j'}$
for $j = \alpha b^k + \beta j' - \gamma$.
Therefore, $R_{m,n}$ is equal to the sum of~$\ell_{k,j'}$ for $(k,j')$
satisfying
$\alpha b^k + \beta j' - \gamma = m - n b^k$.
For fixed $m$~and~$k$, the coefficient $\ell_{k,j'}$ only
contributes when $\beta j' \equiv m + \gamma \pmod{b^k}$.
Its contribution is then to $R_{m,n}$ with
$n = b^{-k}(B - \beta j')$ where $B = m + \gamma - \alpha b^k$,
and we are only interested in
$0 \leq n < w$, i.e., $B - b^k w  < \beta j' \leq B$.
Using the assumption that $\beta$~is coprime with~$b$, the condition
on~$\beta j' \pmod{b^k}$ rewrites as $j' \equiv j'_0 \pmod{b^k}$,
where $j'_0$ is the integer computed at
step~\ref{step:build-mat:starting-index}.
Therefore, the loop~\ref{step:build-mat:inner-loop} correctly computes the
contribution of~$\tilde \ell_k$ to the entries of index less than~$w$ of the
row~$R_{m_i}$, and hence the algorithm works as stated.

The only operations in~$\bK$ performed by the algorithm are one addition and
possibly one comparison (to update the sparse structure) at each loop pass
over step~\ref{step:build-mat:doit}.
The total number of iterations of the innermost loop for a given~$i$ is at most
\[
  \sum_{k=0}^r \left\lceil \frac{d_k}{b^k} \right\rceil
  \leq r + \frac{b}{b-1} d
  \leq r + 2 d
\]
and bounds the number of nonzero entries in the row of index~$m_i$.
The complexity in ops follows by summing over~$i$.
\end{proof}

\begin{algo}
\inputs{
  A linear Mahler operator~$L$ of order~$r$.
  A transformation~$\phi$ of the form~\eqref{eq:change-var-1}.
  Integers $h,w \in \bN$.
  A set $E = \{m_0, \dots, m_{w-1}\}$
  with $m_0 < \cdots < m_{w-1} < h$,
  such that the submatrix
  $(\IM_{m_i, j})_{ 0 \leq i,j < w }$
  of~$\IM(\phi(L))$ is lower, resp.~upper, triangular,
  with at most $r$~zeros on the diagonal.
}
\outputs{
  A vector $(f_1, \dots, f_\sigma)$ of polynomials of degree less than~$w$.
}
\caption{Solutions over prescribed monomial support.}
\label{algo:solve-prescribed-part}
\begin{enumerate}
\item\label{item:init-portion}
  Construct the row-sparse submatrix $S_E = (\IM_{m_i, j})_{ 0 \leq i,j < w }$
  by Algorithm~\ref{algo:build-mat}.
\item\label{item:candidate-basis}
  Compute a basis of $\ker S_{E}$ as a matrix
  $G = (G_{i,j}) \in \bK^{w \times \rho}$
  by forward, resp.~backward, substitution,
  using the row-sparse structure.
\item\label{item:signatures}
  For $1\leq j\leq \rho$, set
  $g_j = G_{0,j} + G_{1, j} x + \dots + G_{w-1, j} x^{w-1} \in \bK[x]$
  and compute the coefficients of
  $L g_j(x) \bmod x^h = \sum_{0\leq i<h} s'_{i,j} x^i$,
  then form the matrix $S' = (s'_{i,j}) \in \bK^{h \times \rho}$.
\item\label{item:kernel}
  Compute a basis of $\ker S'$ as a matrix~$K \in \bK^{\rho\times\sigma}$
  by the algorithm of Ibarra, Moran and Hui~\cite{IbarraMoranHui-1982-GFL}.
\item\label{item:fix-candidates}
  Compute~$F = (F_{i,j}) = G K \in \bK^{w \times \sigma}$.
\item
  Return $(f_1, \dots, f_\sigma)$
  where $f_j = F_{0,j} + \dots + F_{w-1, j} x^{w-1}$.
\end{enumerate}
\end{algo}

According to Proposition~\ref{prop:system-size}, the number of linearly
independent power series solutions and their valuations are determined
by a small upper upper left submatrix of~$\IM$.
As a direct attempt at solving the corresponding linear system
would have too high a complexity (see Remark~\ref{rem:cmp-naive-algo}),
our approach is to first find a set of candidate solutions,
spanning a low-dimensional vector space that contains the approximate
series solutions, and to refine the solving in a second step.
Geometrically, the idea to obtain a candidate solution
$g = g_0 + g_1 x + \cdots$
is to follow the ``profile'' of~$\IM$
(more precisely, the right boundary of the overlapping strips described in
the previous section),
using a single equation~$\IM_m$ to try and compute each coefficient~$g_n$
from $g_0, \dots, g_{n-1}$.
(That is, for each~$n$, we resolutely skip all but one equations
susceptible to determine~$g_n$.)
By duality, this corresponds to keeping
a varying line of increasing integer slope
in contact with the lower Newton polygon,
and having it ``pivot'' around~it.
In this process, the only case that potentially leaves a degree of freedom
in the choice of~$g_n$ is when column~$n$ contains a
``corner'' of the profile, corresponding to an edge of the Newton
polygon.
As a consequence, it is enough to construct
at most~$r$ independent candidates solutions.
The second step then consists in recombining the candidates in such a way that
the equations~$\IM_m$ that were skipped in the first phase
be satisfied.

This strategy is made more precise in Algorithm~\ref{algo:solve-prescribed-part},
which will then be specialized to power series solutions (and later to other
types of solutions) by a suitable choice of $E$, $h$ and~$w$.
By construction, Algorithm~\ref{algo:solve-prescribed-part} outputs polynomials
of degree less than~$w$ that are solutions of a subsystem of the linear system
induced by~$L$.
These polynomials need not \emph{a priori} prolong into actual solutions.

\begin{lem}\label{lem:solve-prescribed-part}
Algorithm~\ref{algo:solve-prescribed-part} runs in
$\bigO(r w d + r^2 w + r^2\Mult(h))$~ops,
and returns a basis of the kernel of the linear map induced by~$\phi(L)$
from $\bK[x]_{<w}$ to~$\bK[x]/(x^h)$.
\end{lem}

\begin{proof}
When $S_E$~is lower, respectively upper, triangular
it is possible at step~\ref{item:candidate-basis}
to compute~$G$ by forward, respectively backward, substitution,
in such a way that~$S_E G = 0$.
By interpreting the $h \times w$ upper left submatrix~$S$ of~$R$ as the matrix
of a restriction of~$L$ to suitable monomial bases,
it follows from the definition of~$S'$ that~$S' = S G$.
Step~\ref{item:kernel} computes~$K$ such that~$S' K = 0$.

The columns of~$F$, computed as~$GK$ at step~\ref{item:fix-candidates},
span the kernel of~$S$:
Indeed, assume $S f = 0$, so that by selecting rows $S_E f = 0$,
and $f$~can be written as $G \gamma$ for some~$\gamma$.
Then, $S' \gamma = S G \gamma = S f = 0$.
But this means that $\gamma = K \eta$ for some~$\eta$, so that $f = G K \eta = F \eta$.
Conversely, we have
$S F = S G K = S' K = 0$,
so that any vector of the form~$F \eta$ belongs to $\ker S$.

Additionally, since the columns of~$G$, respectively those of~$K$,
are linearly independent,
$G K \eta = 0$ implies $K \eta = 0$, which implies $\eta = 0$.
The columns of~$F = G K$ hence form a basis of $\ker S$.

By Lemma~\ref{lem:build-mat}, step~\ref{item:init-portion}
takes~$\bigO(w (r + d))$~ops.
The number of nonzero entries in each row of~$S_E$ is bounded by $r + 2d$
by Lemma~\ref{lem:build-mat}, hence the cost of
computing $\rho$~linearly independent solutions by substitution at
step~\ref{item:candidate-basis} is~$\bigO(\rho w (r + d))$.
As no more than $r$~of the diagonal entries of~$S_E$ are zero,
$\rho$~is at most~$r$.
The computation of each column of~$S'$ at step~\ref{item:signatures}
amounts to adding ${r+1}$ products of the~$\ell_k$ by the~$M^k S_i$,
truncated to order~$h$, for a total of~$\bigO(r^2 \Mult(h))$~ops.
As~$\rho\leq r$, computing the kernel of~$S'$ at step~\ref{item:kernel} via
an LSP decomposition (a generalization of the LUP decomposition) requires
$\bigO(hr^{\omega-1}) = o(r^2 \Mult(h))$~ops~\cite{IbarraMoranHui-1982-GFL}.
Finally, the recombination at step~\ref{item:fix-candidates}
takes~$\bigO(wr^{\omega-1}) = o(r^2 w)$~ops as~$\sigma\leq\rho\leq r$.
\end{proof}

\begin{rem}\label{rem:cmp-naive-algo}
Note that a direct attempt to solve~$S$, when, say, $\phi = \id$
and $w = \bigO(d)$, would result in a complexity~$\bigO(d^\omega)$
(e.g., using the LSP decomposition),
as opposed to~$\bigO(d^2)$
when using Algorithm~\ref{algo:solve-prescribed-part}
and disregarding the dependency in~$r$.
\end{rem}

Let $\tilde v_k$ be the valuation of the coefficient~$\tilde\ell_k$ of
$\phi(L) = \sum_k \tilde\ell_k(x) M^k$.
In analogy with~\eqref{eq:def-sizes}, define
\begin{equation} \label{eq:def-sizes-transf}
  \tilde \nu = \max_{k \geq 1} \frac{\tilde v_0 - \tilde v_k}{b^k - 1} ,
  \qquad
  \tilde \mu = \tilde v_0 + \tilde \nu.
\end{equation}
We now specialize the generic solver to the computation of approximate series solutions (in the sense of the previous subsection) of~$\phi(L)$.
The case $\phi = \id$ is formalized
as Algorithm~\ref{algo:solve-singular-part} on
page~\pageref{algo:solve-singular-part}.

\begin{prop}\label{prop:solve-singular-part-series}
  Assume $\tilde \nu \geq 0$.
  Algorithm~\ref{algo:solve-prescribed-part}, called with
  \[
    h = \lfloor \tilde \mu \rfloor + 1,
    \quad
    w = \lfloor \tilde \nu \rfloor + 1,
    \quad
    E = \bigl( \min_k (\tilde v_k + n b^k) \bigr)_{0 \leq n < w},
  \]
  runs in $\bigO(r d \tilde v_0 + r^2 \Mult(\tilde v_0))$~ops
  and returns a basis of approximate series solutions
  of the equation $\phi(L) \, y = 0$.
\end{prop}

\begin{proof}
First of all, when $m = m_i \in E$,
none of the terms $\tilde \ell_k M^k$ of~$\phi(L)$
contributes to the entries of~$S$ located above $S_{m, n}$.
The matrix $S_E$ is thus lower triangular.
In addition, $\IM_{m,n}$~is zero
(if and) only if $-n$ is an (admissible) slope
of the lower Newton polygon,
so that no more than $r$~of the diagonal
entries of~$S_E$ are zero.
Both preconditions of Algorithm~\ref{algo:solve-prescribed-part}
are therefore satisfied.
By Proposition~\ref{prop:system-size} and Lemma~\ref{lem:solve-prescribed-part},
it follows from the choice of $h$ and~$w$
that the~$f_j$ form a basis of approximate series solutions.
Using the inequalities
$h \leq b v_0/(b-1) + 1 = \bigO(\tilde v_0)$
and
$w \leq v_0/(b-1) + 1 = \bigO(\tilde v_0)$
in the formula of Lemma~\ref{lem:solve-prescribed-part},
the total complexity is as announced.
\end{proof}

\begin{algo}
\inputs{
  A linear Mahler operator~$L$ of order~$r$.
  A transformation~$\phi$ of the form~\eqref{eq:change-var-1}.
  A polynomial $\hat y = y_0+\dots+ y_{\ftnu }x^{\ftnu}$
  such that~$\phi(L) \, \hat y = \bigO(x^{\ftmu + 1})$,
  for $\tilde\nu$ and~$\tilde\mu$ defined by~\eqref{eq:def-sizes-transf}.
  An integer~$n$.
}
\outputs{
  A polynomial $y_0+\dots+ y_{\ftnu +n}x^{\ftnu +n}$.
}
\caption{Prolonging an approximate series solution to any order.}
\label{algo:solve-nonsingular-part}
\begin{enumerate}
\item\label{item:init-portion-nonsingular-part}
  Use Algorithm~\ref{algo:build-mat} with
  $E = \{ \ftmu + 1, \dots, \ftmu + n \}$,
  $h = \ftmu + n + 1$, and
  $w = \ftnu + n + 1$
  to construct a submatrix~$\IM_E$ of~$\IM$.
\item\label{item:solve-nonsingular-part:solve}
  Solve
  $\IM_E \, (y_0, \dots, y_{\ftnu +n})^{\operatorname T} = 0$
  for $y_{\ftnu + 1}, \dots, y_{\ftnu + n}$,
  by forward substitution, starting with the coefficients
  $y_0, \dots y_{\ftnu}$ given on input.
\item
  Return $y_0+\dots+ y_{\ftnu +n}x^{\ftnu +n}$.
\end{enumerate}
\end{algo}

Given an approximate series solution, the next terms of the corresponding
series solutions can be computed efficiently one by one using simple recurrence
formulae.

\begin{prop}\label{prop:formal-series-sols}
Given an approximate series solution
$\hat y = y_0+\dots+ y_{\ftnu}x^{\ftnu}$
of~\eqref{eq:mahler-eqn},
Algorithm~\ref{algo:solve-nonsingular-part} computes
the truncation to the order $\bigO(x^{\ftnu + n})$ of
the unique solution~$y$ of~\eqref{eq:mahler-eqn} of the form
$y = \hat y + \bigO(x^{\ftnu + 1})$
in $\bigO((r + d) \, n)$~ops.
\end{prop}

\begin{proof}
By Proposition~\ref{prop:system-size}, the system to be solved at
step~\ref{item:solve-nonsingular-part:solve} is compatible.
According to the description of~$\IM$ provided above, the submatrix
$(R_{m,n})_{m > \ftmu, n > \ftnu}$ is lower triangular, with nonzero
diagonal coefficients, so that the system can be solved by forward
substitution.
As explained in §\ref{sec:approximate-series-solutions},
the output is a truncation of a solution of~$\phi(L)$.
By Lemma~\ref{lem:build-mat}, the cost in ops of
step~\ref{item:init-portion-nonsingular-part}
is $\bigO((r + d) \, n)$,
and each row of~$S$ contains at most $r + 2 d$ nonzero entries.
Therefore, step~\ref{item:solve-nonsingular-part:solve} costs
$\bigO((r + d) \, n)$~ops.
\end{proof}

\begin{algop}
\inputs{
  A linear Mahler operator~$L$ of order~$r$.
}
\outputs{
  A basis $(f_1, \dots, f_\sigma)$ of approximate series solutions of~$L$.
}
\caption{Approximate series solutions.}
\label{algo:solve-singular-part}
\raggedright
Let $\mu$, $\nu$ be as defined by~\eqref{eq:def-sizes}.
If $\nu<0$, return $()$.
Otherwise, call Algorithm~\ref{algo:solve-prescribed-part} with $\phi = \id$,
  \[
    h = \lfloor \mu \rfloor + 1,
    \quad
    w = \lfloor \nu \rfloor + 1,
    \quad
    E = \bigl( \min_k (v_k + n b^k) \bigr)_{0 \leq n < w},
  \]
and return the result.
\end{algop}

\begin{algop}
\inputs{
  A linear Mahler operator~$L$ of order~$r$. An integer~$w \in \bN$.
}
\outputs{
  A basis $(f_1, \dots, f_\sigma)$ of the polynomial solutions of~$L$
  of degree less than~$w$.
}
\caption{Polynomial solutions of bounded degree.}
\label{algo:poly-bounded-degree}
\raggedright
Let $\mu$, $\nu$ be as defined by~\eqref{eq:def-sizes}.
If $\nu<0$, return $()$.
Otherwise, call Algorithm~\ref{algo:solve-prescribed-part} with $\phi = \id$,
  \begin{equation*}
    h = \max_k d_k + (w-1) b^r + 1,
    \quad
    w,
    \quad
    E = \bigl( \max_k (d_k + n b^k) \bigr)_{0 \leq n < w},
  \end{equation*}
and return the result.
\end{algop}

\begin{algop}
\inputs{
  A linear Mahler operator~$L$ of order~$r$.
}
\outputs{
  A basis $(f_1, \dots, f_\sigma)$ of all polynomial solutions of~$L$.
}
\caption{Basis of polynomial solutions.}
\label{algo:poly-sols-basis}
\raggedright
Call Algorithm~\ref{algo:poly-bounded-degree} with
  $w = \Bigl\lfloor \frac {\max_k d_k} {b^{r-1} (b-1)} \Bigr\rfloor + 1$
  and return the result.
\end{algop}

\subsection{Polynomial solutions}\label{sec:poly}

Our goal in this subsection is Algorithm~\ref{algo:poly-sols-basis},
which computes a basis of all polynomial solutions.
Lemma~\ref{lem:admissible-degrees} provides us with an upper bound
$d/(b^r - b^{r-1}) + 1 = \bigO(d/b^r)$ for
the degree of any polynomial solution.
Before we take this into account,
we provide an algorithm
to compute polynomial solutions with degree bounded by~$w \geq 0$,
which runs in a complexity that is sensitive to~$w$.

In the same way as in Proposition~\ref{prop:solve-singular-part-series},
to obtain candidate polynomial solutions
$f = f_0 + \dots + f_{w-1} x^{w-1}$,
we set $f_n = 0$ for~$n \geq w$ and then compute~$f_n$ for decreasing~$n$
by ``following'' the ``left profile'' of the matrix~$\IM$ (or, dually, the
upper Newton polygon).
The corresponding specialization of Algorithm~\ref{algo:solve-prescribed-part}
is formalized as Algorithm~\ref{algo:poly-bounded-degree}.

\begin{prop} \label{prop:poly-bounded-degree}
  Assume $\nu \geq 0$.
  Algorithm~\ref{algo:solve-prescribed-part}, called with
  $\phi = \id$ and
  \begin{equation}
    h = d + (w-1) b^r + 1,
    \qquad
    E = \bigl( \max_k (d_k + n b^k) \bigr)_{0 \leq n \leq w},
  \end{equation}
  returns a basis of the space of
  polynomial solutions of~\eqref{eq:mahler-eqn} of degree less than~$w$.
  For~$w = \bigO(d/b^r)$,
  the algorithm runs in $\softO(w d + \Mult(d))$~ops.
\end{prop}

\begin{proof}
The proof is similar to that of
Proposition~\ref{prop:solve-singular-part-series}:
the extracted submatrix of~$\IM$ is now upper triangular;
the zeros on its diagonal correspond to the admissible nonpositive
integer slopes of the upper Newton polygon;
the number of such zeros is not more than~$r$.
Both preconditions of Algorithm~\ref{algo:solve-prescribed-part}
are therefore satisfied
and Lemma~\ref{lem:solve-prescribed-part} applies.
Additionally, the choice of~$h$ in terms of~$w$ is such that $\deg(Ly) < h$
whenever $\deg y < w$ for a polynomial~$y$.
So, the basis returned is that of the kernel of the map induced by~$L$
from~$\bK[x]_{<w}$ to~$\bK[x]$,
as announced.

For the complexity result,
the hypothesis on~$w$ implies
$h = \bigO(d)$ and $r = \bigO(\log_b d)$,
so that the conclusion of Lemma~\ref{lem:solve-prescribed-part}
specializes to $\softO(w d + \Mult(d))$~ops.
\end{proof}

\begin{rem}
The loose bound on~$w$, namely~$w = \bigO(d/b^r)$,
permits in particular to obtain a result
when $d$~is not the maximal degree of the~$\ell_k$,
but only bounds them up to a multiplicative constant.
In this case,
the complexity announced by Proposition~\ref{prop:poly-bounded-degree}
specializes to the same complexity
as in Corollary~\ref{cor:poly-sols-basis}.
This will be used
for the numerators of rational-function solutions in~\S\ref{sec:num}.
\end{rem}

By Lemma~\ref{lem:admissible-degrees},
the degree of any polynomial solution is bounded above by
$\delta_0 = d/(b^r - b^{r-1}) + 1$.
Specializing Proposition~\ref{prop:poly-bounded-degree}
to~$w = \lfloor \delta_0 \rfloor$,
we obtain a bound for the complexity of computing
the whole space of polynomial solutions.

\begin{cor}\label{cor:poly-sols-basis}
  Assuming $\nu \geq 0$, Algorithm~\ref{algo:solve-prescribed-part},
  called with $\phi = \id$,
  \[
    h = 3 d + 1,
    \qquad
    w = \Bigl\lfloor \frac d {b^{r-1} (b-1)} \Bigr\rfloor + 1,
    \qquad
    E = \bigl( \max_k (d_k + n b^k) \bigr)_{0 \leq n \leq w},
  \]
  computes a basis of the polynomial solutions of~\eqref{eq:mahler-eqn}
  in $\softO(d^2/b^r + \Mult(d))$~ops.
\end{cor}

\begin{proof}
Observe that the choice for~$w$ induces that $h$,
as defined in Algorithm~\ref{algo:poly-sols-basis},
satisfies~$h \leq 3d + 1$.
The result follows from this fact and $w = \bigO(d/b^r)$.
\end{proof}

\subsection{Puiseux series solutions}
\label{sec:Puiseux}

We now discuss the computation of solutions of~\eqref{eq:mahler-eqn}
in~$\Puiseux$.
Even though Proposition~\ref{prop:nonzero-identity-term} does not apply,
we still assume that the coefficient~$\ell_0$ of~$L$ is nonzero.
There is no loss of generality in doing so:
if $L = L_1 M^w$ for some $w \in \bN$,
then the Puiseux series solutions of~$L$ are exactly the
$y(x^{b^{-w}})$ where $y$~ranges over the Puiseux series solutions of~$L_1$.
Additionally, the order of $L_1$ is bounded by that of~$L$, so that the
complexity estimates depending on it will still hold
(and equations of order zero that result from the transformation when $r = w$
have no nontrivial solutions).

The computation of solutions $y \in \Ser{1/N}$
with a given ramification index~$N$ is similar to that
of power series solutions.
In order to compute a full basis of solutions in $\Puiseux$, however, we need a
bound on the ramification index necessary to express them all.
Lemma~\ref{lem:ExponentDenominatorAndb}, communicated to us by Dreyfus and Roques, and
Proposition~\ref{prop:puiseux-slopes-denominators-lcm} below
provide constraints on the possible ramification indices.

\begin{lem}\label{lem:ExponentDenominatorAndb}
If $y \in \Puiseux$ is a Puiseux series such that $L y \in \bK((x^{1/q'}))$
where $q'$~is coprime with~$b$,
then $y \in \bK((x^{1/q}))$ for some~$q$ coprime with~$b$.
\end{lem}
\begin{proof}
Let $q_0$ be the smallest positive integer such that $y \in \bK((x^{1/q_0}))$.
Set $g = q_0 \wedge b$ and~$q'' = q_0 / g$, so that
$M y \in \bK((x^{b/q_0})) \subset \bK((x^{1/q''}))$.
The expression
\begin{equation*}
y = \ell_0^{-1} \, \bigl(L y - (\ell_1 + \dotsb + \ell_r M^{r-1}) M y \bigr)
\end{equation*}
shows that $y \in \bK((x^{1/q_1}))$ where $q_1 = q' q''$.
By minimality of~$q_0$, we have $q_1 = k q_0$ for some~$k \in \bN$,
which simplifies to~$q' = k g$.
Since $q'$~was assumed to be coprime with~$b$, this implies~$g=1$.
\end{proof}

\begin{rem} \label{rem:non-Puiseux}
Some non-Puiseux formal series solutions
of Mahler equations with $\ell_0 \neq 0$
do involve ramifications of order divisible by~$b$:
perhaps the simplest example, akin to~\cite[p.~64]{Chevalley-1951-ITAF} (see also~\cite{Abhyankar-1956-TNFPS}), is
$y = x^{1/b} + x^{1/b^2} + x^{1/b^3} + \dotsb$,
which satisfies
$(M-x^{b-1}) (M - 1) \, y = 0$.
\end{rem}

The following proposition formalizes,
as a consequence of Lemma~\ref{lem:ExponentDenominatorAndb}
and the properties of Newton polygons discussed in~\S\ref{sec:structure},
that no ramification is needed
beyond those present in the candidate leading terms given by the Newton polygon.
Call~$\mathcal N$ the lower Newton polygon of~$L$, and
let $Q$ denote the set of denominators~$q$ of slopes (written in lowest terms)
of admissible edges of~$\mathcal N$ such that $q \wedge b = 1$.

\begin{prop} \label{prop:puiseux-slopes-denominators-lcm}
Any Puiseux-series solution~$y$ of~$L y = 0$ belongs to
$\mathcal V = \sum_{q \in Q} \Ser{1/q}$.
In particular, the space of solutions of~$L$ in $\Puiseux$ is contained
in\/~$\Ser{1/N}$,
where $N\leq b^r - 1$ denotes the lcm of the elements of~$Q$.
\end{prop}

\begin{proof}
Let $y \in \Puiseux$ satisfy $L y = 0$,
and suppose by contradiction that
$y$~contains a nonzero term of exponent $p_1/q_1$
where $p_1 \wedge q_1 = 1$
and $q_1$ does not divide any element of~$Q$.
Choose $p_1/q_1$ minimal with these properties.
Write $y = y_0 + y_1$ where~$y_0$ consists of
the terms of~$y$ with exponent strictly less than~$p_1/q_1$,
so that $y_0 \in \mathcal V$ and $y_1$~has valuation~$p_1/q_1$.
Then $g = L y_0$ belongs to $\mathcal V$,
so that there exists $q' \in \bN$
for which $q' \wedge b = 1$ and $g \in \bK((x^{1/q'}))$.
Since~$L y_1 = -g$,
Lemma~\ref{lem:ExponentDenominatorAndb} implies that
$y_1 \in \bK((x^{1/q}))$ for some~$q$ coprime with~$b$.
In particular, $q_1$~is coprime with~$b$.

Since $p_1/q_1$~is the valuation of a solution of the equation $L z = -g$,
its opposite~$s=-p_1/q_1$ is the slope of an admissible edge~$\mathcal E$
of the lower Newton polygon~$\mathcal N_g$ of~$(L, -g)$
(see~\S\ref{sec:non-hom}).
On the other hand,
because of the definition of~$Q$
and the properties $q_1 \wedge b = 1$ and~$q_1 \not\in Q$,
the edge~$\mathcal E$ cannot be an edge of~$\mathcal N$.
Therefore,
by the description in~\S\ref{sec:non-hom},
$g$~must be nonzero and
the edge~$\mathcal E$ must be the leftmost edge of~$\mathcal N_g$.
The valuation of~$g \in \mathcal V$ is thus a rational number $p_0/q_0$
(not necessarily in lowest terms) with~$q_0 \in Q$,
so that in particular $q_0 \wedge b = 1$.
As $s$~is the slope of~$\mathcal E$ in~$\mathcal N_g$,
it is of the form
$(q_0 v_k - p_0)/(q_0 b^k)$ for some $k \in \{0, \dots, r\}$.
Then, $q_1$~divides~$q_0 b^k$.
As it is coprime with~$b$,
this implies that $q_1$~divides~$q_0 \in Q$, a contradiction.
We have proved that $y$ belongs to~$\mathcal V$.

Next, it is clear that $\mathcal V$~is contained in~$\Ser{1/N}$.
Finally, letting $(b^{k_i}, v_i)$ denote the vertices of~$\mathcal N$
(sorted from left to right as $i$~increases),
the lcm~$N$ satisfies
$N \leq  \prod_i (b^{k_{i+1} - k_i} - 1) < b^r$,
as claimed.
\end{proof}

\begin{rem}
The bound $N < b^r$ is tight, as shown by the example of $M^r - x$,
which admits the solution $x^{1/(b^r - 1)}$.
\end{rem}

In order to obtain an algorithm that computes a basis of the space of Puiseux series solutions, there remains to generalize the results
of~§\ref{sec:approximate-series-solutions}--\ref{sec:series-sols}
to the case of solutions lying in $\Ser{1/N}$ where~$N$ is given.
Motivated by the structure of the space~$\mathcal V$ described in
Proposition~\ref{prop:puiseux-slopes-denominators-lcm},
we do not require here that $N$~be equal to the lcm of all elements of~$Q$:
setting it to the lcm of any subset of these elements also makes sense.
For the most part, the algorithms searching for power series solutions
apply \emph{mutatis mutandis} when the indices $m$~and~$n$ are allowed to take
negative and noninteger rational values.
Nevertheless, some care is needed in the complexity analysis, so we explicitly
describe a way to reduce the computation of ramified solutions of~$L$ to that
of power series solutions of an operator~$\tilde L$.

Denote $x = t^\beta$, and consider
the change of unknown functions $y(x) = t^\alpha z(t)$,
for $\alpha \in \bZ$ and
$\beta \in \bN_{>0}$ to be determined.
Observe that $M t = t^b$.
If~$y(x)$ is a solution of~$Ly = 0$, then $z(t)$~is annihilated~by
\begin{equation*}
\tilde L = t^{-\gamma} L \, t^\alpha
  = t^{-\gamma} \sum_{k=0}^r t^{\alpha b^k} \ell_k(t^\beta) M^k
  = \sum_{k=0}^r \tilde\ell_k(t) M^k
\end{equation*}
where $\gamma \in \bZ$~can be adjusted so that the~$\tilde{\ell}_k$ belong to~$\bK[t]$.
We then have $\tilde L = \phi(L)$ where $\phi$ is the $\bK$-linear map, already
introduced in~§\ref{sec:power-series-sols}, that sends $x^j M^k$ to
\begin{equation}\label{eq:change-var}
\phi(x^j M^k)
  = t^{-\gamma} t^{\beta j} M^k t^\alpha
  = t^{-\gamma + \beta j + \alpha b^k} M^k .
\end{equation}
Viewing monomials~$x^j M^k$ as points in the plane of the Newton diagram,
the map~$\phi$ induces an affine shearing
\begin{equation}\label{eq:shear-map}
[\phi]: \begin{pmatrix}b^k \\ j\end{pmatrix} \mapsto
\begin{pmatrix}1 & 0 \\ \alpha & \beta\end{pmatrix}
\begin{pmatrix}b^k \\ j\end{pmatrix} +
\begin{pmatrix}0 \\ -\gamma\end{pmatrix}.
\end{equation}
As in~§\ref{sec:power-series-sols},
denote by $\tilde v_k$ and $\tilde d_k$ the valuations and degrees of the
coefficients of~$\tilde L$, and by $\tilde\mu$ and $\tilde\nu$ the quantities
defined by~\eqref{eq:def-sizes} with $v_k$ replaced by~$\tilde v_k$.

\begin{figure}
  \begin{center}

    \begin{tikzpicture}[xscale=0.471579,yscale=0.2]

      \draw[->,>=latex] (0,0) -- (10.5,0) ;
      \draw(10,0) node[above]{$U$};
      \draw[->,>=latex] (0,0) -- (0,20) ;
      \draw(0,19) node[right]{$V$};

      \draw (1,-0.24) -- (1,0.24) ;
      \draw (1,-1.44) node {$1$} ;
      \draw (3,-0.24) -- (3,0.24) ;
      \draw (3,-1.44) node {$3$} ;
      \draw (9,-0.24) -- (9,0.24) ;
      \draw (9,-1.44) node {$9$} ;
      \draw (-0.212054,0) -- (0.212054,0) ;
      \draw (-0.636161,0) node {$0$} ;
      \draw (-0.212054,8) -- (0.212054,8) ;
      \draw (-0.636161,8) node {$8$} ;
      \draw (-0.212054,16) -- (0.212054,16) ;
      \draw (-0.636161,16) node {$16$} ;

      \fill[fill=black!5!white]  (1,6) -- (3,0) -- (9,3)--(9,17)--(1,17)--(1,6) ;
      \fill (9,10) node {\tiny$\bullet$} ;
      \fill (9,9) node {\tiny$\bullet$} ;
      \fill (9,6) node {\tiny$\bullet$} ;
      \fill (1,7) node {\tiny$\bullet$} ;
      \fill (1,6) node {\tiny$\bullet$} ;
      \fill (9,3) node {\tiny$\bullet$} ;
      \fill (3,0) node {\tiny$\bullet$} ;

      \draw [style=dashed,-]  (1,6) -- (3,0) -- (9,3) ;

    \end{tikzpicture}
    \hfil
    \begin{tikzpicture}[xscale=0.471579,yscale=0.2]

      \draw[->,>=latex] (0,0) -- (10.5,0) ;
      \draw(10,0) node[above]{$U$};
      \draw[->,>=latex] (0,0) -- (0,20) ;
      \draw(0,19) node[right]{$V$};

      \draw (1,-0.48) -- (1,0.48) ;
      \draw (1,-1.44) node {$1$} ;
      \draw (3,-0.48) -- (3,0.48) ;
      \draw (3,-1.44) node {$3$} ;
      \draw (9,-0.48) -- (9,0.48) ;
      \draw (9,-1.44) node {$9$} ;
      \draw (-0.212054,0) -- (0.212054,0) ;
      \draw (-0.636161,0) node {$0$} ;
      \draw (-0.212054,16) -- (0.212054,16) ;
      \draw (-0.636161,16) node {$16$} ;

    \fill[fill=black!5!white] (1,14) -- (3,0) -- (9,0) -- (9,17) -- (1,17) -- (1,14) ;
      \fill (9,14) node {\tiny$\bullet$} ;
      \fill (1,16) node {\tiny$\bullet$} ;
      \fill (9,12) node {\tiny$\bullet$} ;
      \fill (1,14) node {\tiny$\bullet$} ;
      \fill (9,6) node {\tiny$\bullet$} ;
      \fill (9,0) node {\tiny$\bullet$} ;
      \fill (3,0) node {\tiny$\bullet$} ;

      \draw [style=dashed,-]  (1,14) -- (3,0) -- (9,0) ;

    \end{tikzpicture}

  \end{center}
  \caption{\label{fig:ChangeOfVariable}
    The transformation in Example~\ref{ex:ChangeOfVariable}
    puts the edge with slope~$1/2$ of the lower Newton polygon of~$L$ (left)
    onto the $U$-axis (Newton polygon of~$\tilde L$, right).}
\end{figure}
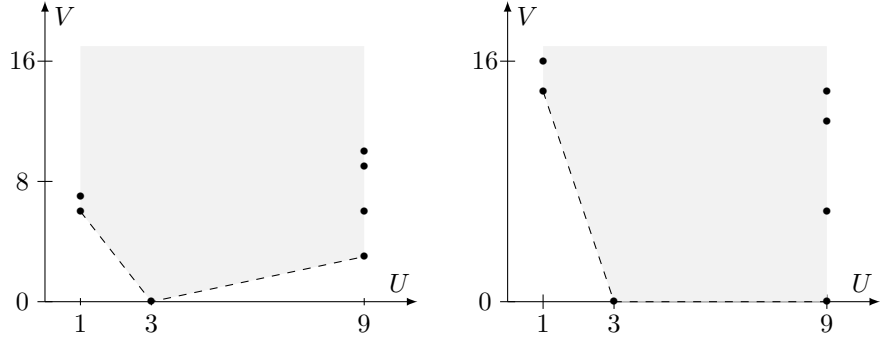

\begin{lem}\label{lem:ChangeOfVariableAndShearing}
Fix an edge~$S_0$ of the lower Newton polygon of~$L$,
of slope~$-p/q$ for (not necessarily coprime) $p\in\bZ$ and $q\in\bN$.
Let~$c$ be the $V$\!-intercept of the line supporting~$S_0$.
Set $\alpha = p$, $\beta = q$, and $\gamma = q c$ in~\eqref{eq:change-var}.
Then:
\begin{enumerate}[label=(\alph*)]
  \item\label{item:chvar-shearing:operator}
    the operator~$\tilde L = \phi(L)$ has polynomial coefficients;
  \item\label{item:chvar-shearing:diagram}
    its Newton diagram is the image of that of~$\tilde L$ by~$[\phi]$,
    with the edge~$S_0$ being mapped to a segment of the $U$\!-axis;
  \item\label{item:chvar-shearing:tilde-v-tilde-d-from-v-d}
    in terms of those of~$L$,
    the parameters associated to~$\tilde L$ satisfy
    \begin{gather*}
      \tilde d_k = -q c + p b^k + q d_k \geq
      \tilde v_k = -q c + p b^k + q v_k \geq 0, \\
      \tilde \nu = q \nu -p \geq 0,
      \qquad
      \tilde \mu = q (\mu - c) \geq 0.
    \end{gather*}
\end{enumerate}
\end{lem}

\begin{proof}
Observe that $q c$~is equal to the common value on~$S_0$ of $p U + qV$.
Since the endpoints of~$S_0$ have integer coordinates, this value is an integer,
and hence the coefficients of~$\tilde L$ are Laurent polynomials.
The transformation~$[\phi]$ of the Newton plane maps segments of slope~$s$ to
segments of slope
$(\alpha + \beta s)/(1 + 0 \cdot s) = p + q s$,
and in particular maps~$S_0$ to a horizontal segment.
By the choice of~$c$, that segment lies on the $U$-axis.
Since $q > 0$, images by $[\phi]$ of points above~$S_0$ lie above $[\phi](S_0)$.
As monomials of~$L$ correspond to points lying on or above~$S_0$, their images
by~$\phi$ are monomials of nonnegative degree.
This proves assertion~\ref{item:chvar-shearing:operator}.
It follows that $\tilde L$~has a Newton diagram in the sense of our definition,
and it is then clear this Newton diagram is as stated
by~\ref{item:chvar-shearing:diagram}.
The expressions of $\tilde v_k$ and $\tilde d_k$
in~\ref{item:chvar-shearing:tilde-v-tilde-d-from-v-d}
are a consequence of~\eqref{eq:change-var},
using again the positivity of~$\beta$.
Those of $\tilde \nu$~and~$\tilde \mu$ follow.
We already observed that~$\tilde v_k \geq 0$.
Finally, $-\tilde \nu$ and~$\tilde \mu$ are, respectively, the slope and
$V$\!-intercept of the leftmost edge of
the lower Newton polygon~$\tilde{\mathcal N}$ of~$\tilde L$.
Since $\tilde{\mathcal N}$~has a horizontal edge, $\tilde \nu$ and $\tilde \mu$
are nonnegative.
\end{proof}

\begin{ex}\label{ex:ChangeOfVariable}
Consider again the Mahler operator~$L$
in~\eqref{eq:running-example}
treated for~$b = 3$ in~\S\ref{sec:ex}.
We already observed that
the slopes of the Newton polygon of~$L$ are $-3$ and~$1/2$
and that they are admissible, and,
in~\S\ref{sec:ex}, we performed
the transformation~\eqref{eq:change-var}
for the parameters $\alpha = -1$, $\beta = 2$, and $\gamma = -3$,
to obtain the operator~$\tilde L$ in~\eqref{eq:running-example-transfd}.
The slopes of the Newton polygon of~$\tilde L$ are $-7$ and~0
and are both admissible.
\end{ex}

\begin{algo}
\inputs{
  A linear Mahler operator~$L$ as in~\eqref{eq:mahler-opr}.
  A ramification index $N \in \bN_{>0}$.
  A truncation order~$n \in \bN$.
}
\outputs{
  A vector $(\hat y_1, \dots, \hat y_\sigma)$ of truncated Puiseux series.
}
\caption{Solving a Mahler equation in~$\Ser{1/N}$.}
\label{algo:PuiseuxSolutions}
\begin{enumerate}
  \item \label{step:PuiseuxSolutions:edge}
    Compute the slope~$s$ and $V$-intercept~$c$ of the rightmost admissible
    edge of the lower Newton polygon of~$L$ with slope in $N^{-1} \bZ$.
  \item \label{step:PuiseuxSolutions:change-var}
    Define $\phi$ and $\tilde L = \phi(L)$ according to~\eqref{eq:change-var},
    with
    $\alpha = -N s$,
    $\beta = N$,
    and $\gamma = N c$.
  \item\label{step:PuiseuxSolutions:solve-wrt-t}
    Call Algorithm~\ref{algo:solve-prescribed-part} on $L$ and $\phi$, with
    \[
      h = \lfloor \tilde\mu \rfloor + 1,
      \quad
      w = \lfloor \tilde\nu \rfloor + 1,
      \quad
      E = \bigl( \min_k (\tilde v_k + n b^k) \bigr)_{0 \leq n < w},
    \]
    where $\tilde \mu$, $\tilde \nu$ and $\tilde v_k$ are given by
    Lemma~\ref{lem:ChangeOfVariableAndShearing}\ref{item:chvar-shearing:tilde-v-tilde-d-from-v-d},
    to compute a vector $(f_1, \dots, f_\sigma)$ of
    approximate power series solutions of $\tilde L z = 0$.
  \item \label{step:PuiseuxSolutions:extend}
    For $i = 1, \dots, \sigma$, call Algorithm \ref{algo:solve-nonsingular-part}
    to compute
    $\tilde n = \max(0, N (s + n) - \lfloor\tilde\nu\rfloor)$
    additional terms of~$f_i$, thus extending it to a truncated power series
    solution
    $\hat z_i = z_0 + \dots + z_{N (s + n)} x^{N (s + n)}$
    of~$\tilde L$.
  \item\label{step:PuiseuxSolutions:chvar}
    Return $(\hat y_1, \dots, \hat y_\sigma)$ where
    $\hat y_i = z_0 x^{-s} + z_1 x^{-s + 1/N} + \dots + z_{N (s + n)} x^{n}$.
\end{enumerate}
\end{algo}

\begin{thm}\label{thm:PuiseuxSeriesSolutionsCost}
Algorithm~\ref{algo:PuiseuxSolutions}
runs in
\[ \bigO(r^2 \Mult(N d) + r N (d^2 + (r+d) \, n))
  = \softO(r^2 N d \, (d + n))~\text{ops} \]
(assuming a softly linear-time polynomial multiplication)
and computes the truncation to order~$\bigO(x^{n+1})$
of a basis of solutions of~\eqref{eq:mahler-eqn} in $\Ser{1/N}$.
\end{thm}

\begin{proof}
The discussion at the beginning of this section shows that $z(x) \in \Puiseux$
is a solution of the operator~$\tilde L$ computed at
step~\ref{step:PuiseuxSolutions:change-var} if and only if
$y(x) = x^{-s} z(x^{1/N})$ is a solution of~$L$.
By Lemma~\ref{lem:valuation-at-0} and the choice of~$s$, solutions of~$L$
in~$\Ser{1/N}$ have valuation at least~$-s$, and hence
correspond to solutions of~$\tilde L$ lying in $\bK[[x]]$.
Since the mapping $z \mapsto y$ is linear and invertible, a basis of solutions
of~$\tilde L$ in $\bK[[x]]$ provides a basis of solutions of~$L$ in~$\Ser{1/N}$.

Let~$S_0$ be the edge of the Newton polygon of~$L$ considered at
step~\ref{step:PuiseuxSolutions:edge}, so that the notation of the algorithm
agrees with that of Lemma~\ref{lem:ChangeOfVariableAndShearing}.
Lemma~\ref{lem:ChangeOfVariableAndShearing}\ref{item:chvar-shearing:tilde-v-tilde-d-from-v-d}
then provides expressions various parameters associated to~$\tilde L$
in terms of $s$, $c$, and quantities that can be read off~$L$.
Since $\tilde \nu$~is nonnegative,
Proposition~\ref{prop:solve-singular-part-series} applies and shows that
step~\ref{step:PuiseuxSolutions:solve-wrt-t} computes a basis $(f_1,\dots,f_{\sigma})$ of the space of
approximate solutions of~$\tilde L$ in~$\bK[[x]]$ in
$\bigO(r d \tilde v_0 + r^2 \Mult(\tilde v_0))$ ops.
Denote by $(z_1, \dots, z_\sigma)$ the basis of power series solutions
of~$\tilde L$ such that each~$z_i$ extends~$f_i$.
Then, according to Proposition~\ref{prop:formal-series-sols}, the
series~$\hat z_i$ computed at step~\ref{step:PuiseuxSolutions:extend}
satisfy $z_i = \hat z_i + \bigO(x^{N(s+n)+1})$,
and their computation takes
$\bigO(\sigma (r + d) \tilde n)$~ops.
Finally, the truncated Puiseux series returned by the algorithm satisfy
$\hat y_i = x^{-s} \hat z_i(x^{1/N})$, hence are truncations of elements of a
basis of solutions of~$\tilde L$ in~$\Ser{1/N}$.

Steps other than
\ref{step:PuiseuxSolutions:solve-wrt-t}~and~\ref{step:PuiseuxSolutions:extend}
do not perform any operation in~$\bK$, so that the cost in ops of the algorithm
is concentrated in those two steps.
Let $(b^{k_1}, v_{k_1})$ and $(b^{k_2}, v_{k_2})$ with $k_1 < k_2$ be the
endpoints of~$S_0$, so that
\begin{equation} \label{eq:qc}
  q c = p b^{k_1} + q v_{k_1} = p b^{k_2} + q v_{k_2}.
\end{equation}
Lemma~\ref{lem:ChangeOfVariableAndShearing}\ref{item:chvar-shearing:tilde-v-tilde-d-from-v-d} gives
$\tilde v_0 = q v_0 + p - q c$.
If $p \geq 0$, then \eqref{eq:qc}~implies $q c \geq p$ and hence
$\tilde v_0 \leq q v_0 \leq N d$.
If, now, $p < 0$, first observe that since $b^{k_2} \geq 2 b^{k_1}$, we have
$-p b^{k_1} \leq -p (b^{k_2} - b^{k_1}) = q (v_{k_2} - v_{k_1})$.
It follows that
$-q c = - p b^{k_1} - q v_{k_1} \leq q v_{k_2}$,
whence
$\tilde v_0 \leq q (v_0 + v_{k_2}) \leq 2 N d$.
In both cases, we have proved that $\tilde v_0 = \bigO(N d)$.
The complexity estimate for step~\ref{step:PuiseuxSolutions:solve-wrt-t}
thus rewrites as $\bigO(r N d^2 + r^2 \Mult(N d))$~ops.
As $s \leq d$ (because all slopes of the Newton polygon are bounded
by~$d$ in absolute value) and $\sigma \leq r$,
that of step~\ref{step:PuiseuxSolutions:extend}
becomes $\bigO(r N (r + d) (d + n))$~ops.
The total running time is therefore
$\bigO(r^2 \Mult(N d) + r N (d^2 + (r+d) \, n))$~ops.
\end{proof}

Recall that $Q$~denotes the set of denominators~$q$
of slopes, written in lowest terms,
of admissible edges of~$\mathcal N$ such that $q \wedge b = 1$.

\begin{cor}\label{coro:Puiseux}
Algorithm~\ref{algo:PuiseuxSolutions} with
$N$~set to the lcm of elements in $Q$, returns
the truncation to order~$\bigO(x^{n+1})$
of a basis of solutions of~\eqref{eq:mahler-eqn} in $\Puiseux$
in $\softO(r^2b^r d (d + n))$~ops, assuming $\Mult(k) = \softO(k)$.
\end{cor}

\begin{proof}
This follows by combining
Proposition~\ref{prop:puiseux-slopes-denominators-lcm} with
Theorem~\ref{thm:PuiseuxSeriesSolutionsCost}.
\end{proof}

\begin{ex}\label{example:GorgeousPuiseuxExample2:-)}
With $b = 3$, let us consider the order~$r = 11$ Mahler operator
\begin{multline*}
  L =
  x^{568}
  -  (x^{1218} + x^{1705})  M
  +  x^{3655}  M^2
  -  (x^{162} - x^{10962})  M^3 \\
  +  (1+x^{487}-x^{4104}-x^{4536}-x^{32887})  M^{4}
  -  (x - x^{11826} - x^{12313} - x^{13122} - x^{13609})  M^5 \\
  -  (1 + x^{35479} + x^{39367})  M^6
  +  (x+x^{95634}-x^{106434}-x^{118098})  M^7 \\
  -  (x^{286416} + x^{286903} - x^{319303} - x^{354295})  M^8
  +  x^{859249}  M^9 \\
  +  x^{2577744}  M^{10}
  -  x^{7733233}  M^{11}
  .
\end{multline*}
Its associated parameters are $w = 0$, $v_0 = 568$, and a Newton polygon made from five segments, all admissible, with slopes $-203/13$, $-3$, $0$, $1/1458$, and $221/5$. Except for $1458 = 2 \cdot 3^6$, the denominators are coprime with $b = 3$ and their lcm is~$N = 65$. The rightmost slope is $s = 221/5$ and we perform the change of variables of Algorithm~\ref{algo:PuiseuxSolutions} with
$\alpha = -2873$, $\beta = 65$, hence $\gamma = -6283186$ and this provides us with the new operator
\begin{multline*}
  \tilde L =
  t^{6317233}
  -  (t^{6353737} + t^{6385392})  {M}
  +  t^{6494904}  {M}^2
  -  (t^{6216145} - t^{6918145})  {M}^3 \\
  +  (t^{6050473}+t^{6082128}-t^{6317233}-t^{6345313}-t^{8188128})  {M}^{4} \\
  -  (t^{5585112} - t^{6353737} - t^{6385392} - t^{6437977} - t^{6469632})  {M}^5 \\
  -  (t^{4188769} -t^{6494904}-t^{6747624})  {M}^6
  +  (1+t^{6216145}-t^{6918145}-t^{7676305})  {M}^7 \\
  -  (t^{6050473} + t^{6082128} - t^{8188128} - t^{10462608}) {M}^8
  +  t^{5585112}  {M}^9
  +  t^{4188769}  {M}^{10}
  -  {M}^{11}
  .
\end{multline*}
We want to find a basis of Puiseux solutions for~$L$ with a precision~$\bigO(x^{n})$ where $n = 10^6$. According to Algorithm~\ref{algo:PuiseuxSolutions}, this leads us to compute a basis of formal series solutions for~$\tilde L$ with a precision~$\bigO(x^{\tilde n})$ where $\tilde n = 65002873$. We first apply
Algorithm~\ref{algo:solve-singular-part} with $\tilde \nu = 3888$, $\tilde \mu = 6321121$. The computation shows that the space of solutions has dimension~$2$. We extend the solutions to the requested precision by Algorithm~\ref{algo:solve-nonsingular-part}
 and we obtain a basis of formal series solutions
\begin{multline*}
  \tilde f_1(t) =
  1+{t}^{28080}+{t}^{657072}+{t}^{2274480}+{t}^{2302560}+{t}^{17639856}
  +{t}^{53222832}\\
  +{t}^{53250912}+{t}^{62068032}+\bigO \left( {t}^{65002873}
   \right),
\end{multline*}
\begin{multline*}
  \tilde f_2 (t) =
  {t}^{3888}+{t}^{314928}+{t}^{343008}+{t}^{9160128}+{t}^{25509168}+{t}
  ^{25537248}\\
  +{t}^{27783648}+{t}^{27811728}+\bigO \left( {t}^{65002873}
   \right).
\end{multline*}
Reversing the change of variable, we find the basis
\begin{multline*}
  f_1(x) =
  x^{-{\frac{221}{5}}}+x^{{\frac{1939}{5}}}+x^{{\frac{50323}{5}}}+
  x^{{\frac{174739}{5}}}+x^{{\frac{176899}{5}}}+x^{{\frac{1356691}
  {5}}}+x^{{\frac{4093843}{5}}}\\
  +x^{{\frac{4096003}{5}}}+x^{{\frac{
  4774243}{5}}}+\bigO \left( x^{1000000} \right),
\end{multline*}
\begin{multline*}
  f_2(x) =
  x^{{\frac{203}{13}}}+x^{{\frac{62411}{13}}}+x^{{\frac{68027}{13}
  }}+x^{{\frac{1831451}{13}}}+x^{{\frac{5101259}{13}}}+x^{{\frac{
  5106875}{13}}}\\+x^{{\frac{5556155}{13}}}+x^{{\frac{5561771}{13}}}+\bigO \left( x^{1000000} \right).
\end{multline*}
These truncated series satisfy $Lf_1 = \bigO(x^{e})$, $Lf_2 = \bigO(x^{e})$
with $e = v_0 + n = 1000568$.

\end{ex}

\section{Rational solutions}\label{sec:rat-sols}

We now turn to the computation of rational function solutions
of Mahler equations of the form~\eqref{eq:mahler-eqn}.
Our algorithm follows a classical pattern:
it first computes a \emph{denominator bound}, that is, a polynomial that
the denominator of any (irreducible) rational solution must divide.
Then it makes a change of unknown functions and computes the possible
numerators using the algorithm of~\S\ref{sec:poly}.
As is usual with other functional equations,
the denominator bound is obtained by analyzing the action
of the operator~$L$ on zeros and poles of the functions it is applied to.

\subsection{Denominator bounds: setting}
\label{sec:den-intro}

We will call a rational function~$p/(x^{\bar v} q)$
\emph{in lowest terms}
if it satisfies the following conditions:
$\bar v \geq 0$; $p,q\in\bK[x]$ are coprime polynomials;
$q(0) \neq 0$;
and $p(0)$~can be zero only if~$\bar v = 0$.

Consider a rational solution~$p/(x^{\bar v} q)$ of~\eqref{eq:mahler-eqn},
written in lowest terms.
We already know from Lemma~\ref{lem:valuation-at-0} that
$\bar v \leq v_r/(b^r-b^{r-1})$,
so we are left with the problem of finding a multiple of~$q$.

Write $T a = \bigvee_{i=0}^{r-1} M^i a$.
We will freely use the fact that $T(ab) \mid (T a) \, (T b)$ for all $a$~and~$b$.
For any~$j$ between 0 and~$r$, multiplying the equation
\[
  \ell_r(x) M^r y + \cdots + \ell_1(x) \, M y + \ell_0(x) y = 0,
\]
by
$(M^r x^{\bar v}) \, (M^j q) \, \bigvee_{i\neq j} M^i q$
and reducing modulo $M^j q$ yields
\begin{equation} \label{eq:Mdiv-general}
  M^j q
  \mid  x^{(b^r - b^j) {\bar v}} \ell_j \, (M^j p) \, \bigvee_{i \neq j} M^i q.
\end{equation}
As $q$ is coprime with $p$ and $q(0) \neq 0$,
Equation~\eqref{eq:Mdiv-general} with $j=r$ implies
\begin{equation} \label{eq:Mdiv-coprime-r}
  M^r q \mid \ell_r \, T q.
\end{equation}
This relation is our starting point for computing a polynomial $q^\star$,
depending only on~$\ell_r$, such that $q \mid q^\star$.

The algorithm for this task, presented in~\S\ref{sec:den-algo}, operates
with polynomials over~$\bK$, but
it may be helpful in order to get an intuition
to first consider the case $\bK = \bC$.
Assume for simplicity that $q$~is squarefree.
Equation~\eqref{eq:Mdiv-coprime-r} then says that, if $\alpha$~is a
zero of~$q$, each of its $b^r$th roots is either a $b^k$th root
with $k < r$ of some zero of~$q$ or a zero of~$\ell_r$.
Thus, when $\alpha$ is not a root of unity, its $b^r$th roots are either
zeros of~$\ell_r$ or roots of lower order of some \emph{other} zero
of~$q$, whose $b^r$th roots then satisfy the same property.
(Compare Lemma~\ref{lem:out-of-cycle-for-G} below.)
As $q$~has finitely many zeros,
this cannot continue indefinitely, so, in this case, we will eventually
find a zero~$\alpha$ whose $b^r$th roots are zeros of~$\ell_r$.
A difficulty arises when $\alpha$~is a root of unity, but then
at most one of its $b$th roots can be part of a cycle
of the map $\zeta \mapsto \zeta^b$
(cf.\ Lemma~\ref{lem:M-exits-cycles}),
and a closer examination shows that
the $b-1$ other roots behave essentially like non-roots of unity.

\subsection{Properties of the Mahler and Gräffe operators}
\label{sec:MG}

Going back to the general case,
and before making the reasoning sketched above more precise,
let us state a few properties of the action of~$M$ on polynomials.
Besides~$M$, we consider the \emph{Gräffe operator} defined by
\[
G : \bK [x] \rightarrow \bK [x], \hspace{1em} p \mapsto
\operatorname{Res}_y (y^b - x, p (y)).
\]
In other words, $Gp$~is the product
$p(x^{1/b}) p(\zeta x^{1/b}) \dotsb p(\zeta^{b-1} x^{1/b})$
for any primitive $b$th root of unity~$\zeta$.
While $M$~maps a polynomial~$p$ to a polynomial whose complex zeros
are the $b$th roots of the zeros of~$p$,
the zeros of~$Gp$ are the $b$th powers of the zeros of~$p$.

As a direct consequence of the definitions,
$M$ and~$G$ act on degrees by:
\[ \deg Mp = b \deg p, \qquad \deg Gp = \deg p . \]
Some other elementary properties that will be useful in the sequel
are as follows.

\begin{lem}
  \label{lemma:MG}
  For any nonzero~$i \in \bN$,
  the following relations between $M$~and~$G$ hold for all
  $p, q \in \bK [x]$:
  \begin{enumerate}[label=(\alph*)]
    \item \label{item:GM}
    $G^i M^i p = p^{b^i}$,

    \item \label{item:pGp}
    $p \divides M^iG^ip$,

    \item \label{item:Mdiv}
    $p \divides q \Longleftrightarrow M^i p \divides M^i q$.
  \end{enumerate}
\end{lem}

\begin{proof}
  The case~$i>1$ reduces to the case~$i=1$ by changing the radix,
  since $M^i$ (resp.~$G^i$) is nothing but the Mahler (resp.\ Gräffe)
  operator of radix~$b^i$;
  so we set~$i=1$.
  The assertions \ref{item:GM} and~\ref{item:pGp} are direct
  consequences of the definition of~$G$ as a resultant.
  The direct implication in~\ref{item:Mdiv} is clear.
  For the converse, write the Euclidean division $q = up + v$.
  If $M q = sM p$ for some $s \in \bK [x]$,
  then $(M u) \, (M p) + (M v) = sM p$,
  whence $M v =0$ since $\deg Mv < \deg Mp$.
\end{proof}

\begin{lem} \label{lemma:G-irred}
  If $p \in \bK [x]$ is monic irreducible and $i \in \bN$,
  then $G^i p = q^e$ for some monic irreducible $q \in \bK [x]$
  and $e \in \bN$.
  Furthermore, $G^i p = p$ if and only if $p$ divides~$M^i p$.
  If this holds for~$i>0$,
  $G^j p$~is monic irreducible for any~$j \in \bN$.
\end{lem}

\begin{proof}
  To prove the first point, consider the factorization
  $G^{i} p = c q_1^{e_1} \cdots q_s^{e_{s}}$
  of $G^{i}p$ for monic irreducible and pairwise coprime~$q_j$
  and a nonzero~$c \in \bK$.
  Because of Lemma~\ref{lemma:MG}\ref{item:Mdiv},
  the polynomials
  $M^{i}  q_1^{e_1}, \dots, M^{i} q_s^{e_{s}}$ are pairwise coprime.
  We have
  \[ M^iG^i p = c \, (M^i q_1^{e_1}) \cdots (M^i q_s^{e_s}), \]
  and,
  by Lemma~\ref{lemma:MG}\ref{item:pGp},
  $p \divides M^{i}q_j^{e_{j}}$ for some~$j$.
  It follows that $G^{i} p \divides G^{i}M^{i} q_j^{e_{j}} = q_j^{e_{j}b^{i}}$
  by Lemma~\ref{lemma:MG}\ref{item:GM},
  proving the first point.

  Now if $p \divides M^i p$, then $G^i p \divides p^{b^i}$,
  and necessarily there is~$e \in \bN$ such that $G^i p = p^e$.
  In fact, $e = 1$ and~$G^i p = p$
  as $G^i p$ and~$p$ have the same degree and $p$~is irreducible.
  Conversely, if $G^i p = p$, then $p$ divides $M^i p$
  by~Lemma~\ref{lemma:MG}\ref{item:pGp}.

  Assume~$G^i p = p$ for some~$i>0$.
  Let~$j \in \bN$ and $m \in \bN$ such that~$mi \geq j$.
  Then $p = G^{mi} p = G^{mi-j} (G^j p)$ is monic irreducible,
  so that $G^j p$~is monic irreducible too.
\end{proof}

\begin{lem}
\label{lem:f-neq-x}
Let $f\in\bK[x]$ be a nonconstant polynomial with $f(0) \neq 0$.
If $f$ and its derivative $f'$ are coprime,
so are $Mf$ and $(Mf)'$.
\end{lem}

\begin{proof}
Assume $f \wedge f'=1$.  Applying~$M$ to a Bézout relation shows
that~$Mf \wedge M(f')=1$.  Now, $(Mf)' = bx^{b-1} M(f')$, so a common
factor~$s$ of $Mf$ and~$(Mf)'$ must divide~$x$. As $x$~cannot divide~$Mf$
because~$x\nmid f$, the only possibility is that $s$~be a constant.
\end{proof}

The following lemma generalizes the fact that the iterated $b$th roots of a
complex number~$\alpha \neq 0$ are all distinct, except in some cases where
$\alpha$ is a root of unity.

\begin{lem}\label{lem:out-of-cycle-for-G}
  Let $p \in \bK [x]$ be monic and irreducible.
  For general\/~$\bK$, $M^i p$ and $M^j p$ are coprime for all $i > j \geq 0$
  if none of the $G^i p$ for $i \geq 1$ is equal to~$p$.
  When\/ $\bK=\bQ$, the same conclusion holds if\/ $G p$~is not equal to~$p$.
\end{lem}

\begin{proof}
  We proceed by contraposition, assuming the negation of the common conclusion:
  for monic irreducible~$p$, assume
  $M^i p \wedge M^j p \neq 1$ for some $i > j \geq 0$.
  Set~$k = i - j \geq 1$.
  Lemma~\ref{lemma:MG}\ref{item:Mdiv} implies that
  $M^k p$ and~$p$ are not coprime.
  Then $p$~divides~$M^k p$ and
  Lemma~\ref{lemma:G-irred} implies that~$G^k p = p$.
  This proves the result for general~$\bK$.
  For~$\bK=\bQ$, a further consequence is that
  the map $\alpha \mapsto \alpha^{b^k}$ is a permutation of the roots
  of~$p$ in~$\bar \bQ$.
  Hence, all roots of~$p$ satisfy $\alpha^B = \alpha$
  for some power~$B = b^e$ of~$b$, with~$e>0$.
  This means that $p$~divides~$x^B - x$.
  If $p = x$, $Gp = p$;
  otherwise, $p$~is a cyclotomic polynomial~$\Phi_a$
  with $a \mid b^e-1$, so $a \wedge b = 1$.
  Applying the formula in~\cite[Prop.~4 p.~14]{Dumas-1993-RMS} yields
  $M \Phi_a = \prod_{b' \mid b} \Phi_{a b'}$,
  so that $p$~divides~$M p$.
  Lemma~\ref{lemma:G-irred} now implies $G p = p$ again,
  completing the proof.
\end{proof}

\begin{rem}
  Over a general subfield $\bK \subset \bC$, the cyclotomic
  polynomial $\Phi_a$ factors as $\Phi_a = \Psi_1 \cdots \Psi_s$ and $G$
  acts as a cyclic permutation of the $\Psi_i$.
  See also~\cite[Chap.~1]{Dumas-1993-RMS} for a detailed description of
  the case $a \wedge b \neq 1$.
\end{rem}

Lemma~\ref{lem:out-of-cycle-for-G} states a result for polynomials~$p$
that are not part of a cycle of the map~$G$.
As a matter of fact, a related graph whose structure plays a crucial role
in what follows is that of the map~$\sqrt G$
that maps a monic irreducible~$p$ to the unique monic irreducible~$q$
such that $G p$~is some power of~$q$:
we call this map the \emph{radical\/} of~$G$,
as it ignores the exponent generally introduced by~$G$.
An immediate degree argument shows that the cycles of~$G$
are exactly the cycles of~$\sqrt G$,
and consist of monic irreducible polynomials only.

To find a kind of generalization of Lemma~\ref{lem:out-of-cycle-for-G}
that applies to polynomials on cycles of~$\sqrt G$,
we can always reduce to its hypothesis $G^i p \neq p$ for nonzero~$i$,
by ``stepping back one step'' in the graph of~$\sqrt G$,
thus leaving the cycle.

\begin{lem}
\label{lem:M-exits-cycles}
Let $f\in\bK[x]$ be a nonconstant polynomial with $f(0) \neq 0$.
There exists a monic irreducible
factor~$q\in\bK[x]$ of~$Mf$ such that $G^kq\neq q$ for all
nonzero~$k\in\bN$.
\end{lem}

\begin{proof}
Choose a monic irreducible factor~$p$ of~$f$ and write $Mp = q_1\dotsm q_s$
for monic irreducible~$q_i$.  By contradiction, assume that for each~$i$,
there is some nonzero~$k_i$ for which~$G^{k_i}q_i = q_i$.  It follows that
for~$k=k_1\dotsm k_s$ and all~$i$, $G^kq_i = q_i$.
Lemma~\ref{lemma:MG}\ref{item:GM} implies $p^b = (G q_1) \dotsm (G q_s)$,
and because of Lemma~\ref{lemma:G-irred}, for all~$i$,
$G q_i$ is irreducible.
Hence, there exist nonzero~$e_i\in\bN$ such that $G q_i = p^{e_i}$, with
$b=e_1+\dots+e_s$.
Therefore, for each~$i$, $q_i = G^{k-1}p^{e_i}$, so that,
as $q_i$~is irreducible, $e_i=1$,
and thus all~$q_i$ are equal to some same monic irreducible~$\tilde q$.
It follows that~$Mp=\tilde q^b$.
As $p$~is irreducible,
Lemma~\ref{lem:f-neq-x} applies to show that~$Mp \wedge (Mp)' = 1$,
which is impossible.
The result follows by setting $q=q_i$ for a suitable~$i$.
\end{proof}

\begin{figure}
\begin{scriptsize}

\centerline{%
\begin{tikzpicture}[level distance=6mm]
\tikzstyle{every node}=[inner sep=1pt]
\tikzstyle{level 1}=[sibling distance=40mm]
\tikzstyle{level 2}=[sibling distance=20mm]
\tikzstyle{level 3}=[sibling distance=14mm]
\tikzstyle{level 4}=[sibling distance=8mm]
\tikzstyle{level 5}=[sibling distance=4mm]
\node (phi_1) {$\Phi_a$} [<-]
  child {node {$\Phi_{2a}$}
    child {node {$\Phi_{4a}$}
      child {node {$\Phi_{8a}$}
        child {node {$\Phi_{16a}$}
          child {node {$\vdots$}}
        }
        child {node {$\Phi_{48a}$}
          child {node {$\vdots$}}
        }
      }
      child {node {$\Phi_{24a}$}
        child {node {$\Phi_{144a}$}
          child {node {$\vdots$}}
        }
      }
    }
    child {node {$\Phi_{12a}$}
      child {node {$\Phi_{72a}$}
        child {node {$\Phi_{432a}$}
          child {node {$\vdots$}}
        }
      }
    }
  }
  child {node {$\Phi_{3a}$}
    child {node {$\Phi_{9a}$}
      child {node {$\Phi_{27a}$}
        child {node {$\Phi_{81a}$}
          child {node {$\vdots$}}
        }
        child {node {$\Phi_{162a}$}
          child {node {$\vdots$}}
        }
      }
      child {node {$\Phi_{54a}$}
        child {node {$\Phi_{324a}$}
          child {node {$\vdots$}}
        }
      }
    }
    child {node {$\Phi_{18a}$}
      child {node {$\Phi_{108a}$}
        child {node {$\Phi_{648a}$}
          child {node {$\vdots$}}
        }
      }
    }
  }
  child {node {$\Phi_{6a}$}
    child {node {$\Phi_{36a}$}
      child {node {$\Phi_{216a}$}
        child {node {$\Phi_{1296a}$}
          child {node {$\vdots$}}
        }
      }
    }
  }
;
\path[->] (phi_1) edge [loop above] node {} ();
\path[->] node (x) at (-6, 0) {$x$} edge [loop above] node {} ();
\end{tikzpicture}
}

\centerline{%
\begin{tikzpicture}[level distance=6mm]
\tikzstyle{every node}=[inner sep=1pt]
\tikzstyle{level 1}=[sibling distance=4mm]
\tikzstyle{level 2}=[sibling distance=15mm]
\tikzstyle{level 3}=[sibling distance=15mm]
\tikzstyle{level 4}=[sibling distance=25mm]
\tikzstyle{level 5}=[sibling distance=4mm]
\node at (7, .5) {$\vdots$} [<-]
  child {node {$x-2^{216}$}
    child {node {$x-2^{36}$}
      child {node {$x-2^6$}
        child {node {$x-2$}
          child {node {$x^6-2$}
            child {node {$\vdots$}}
          }
        }
        child {node {$x+2$}
          child {node {$x^6+2$}
            child {node {$\vdots$}}
          }
        }
        child {node {$x^2-2x+4$}
          child {node {$x^{12}-2x^6+4$}
            child {node {$\vdots$}}
          }
        }
        child {node {$x^2+2x+4$}
          child {node {$x^{12}+2x^6+4$}
            child {node {$\vdots$}}
          }
        }
      }
      child {node {$\dots$}}
      child {node {$\dots$}}
      child {node {$\dots$}}
    }
    child {node {$\dots$}}
    child {node {$\dots$}}
    child {node {$\dots$}}
  }
;
\end{tikzpicture}
}

\end{scriptsize}

\caption{Graph of the radical~$\sqrt G$ of the Gräffe operator
  for~$b = 6$ in~$\bQ[x]$.
  Here, $a$~is a positive integer, coprime to~$b$.
  In general, the graph of~$\sqrt G$ consists of
  a loop rooted at~$x$ (top left),
  bi-infinite trees (bottom),
  and cycles between cyclotomic polynomials
  with infinite trees rooted at them (top right).
}
\label{fig:radical-of-Gräffe}
\end{figure}
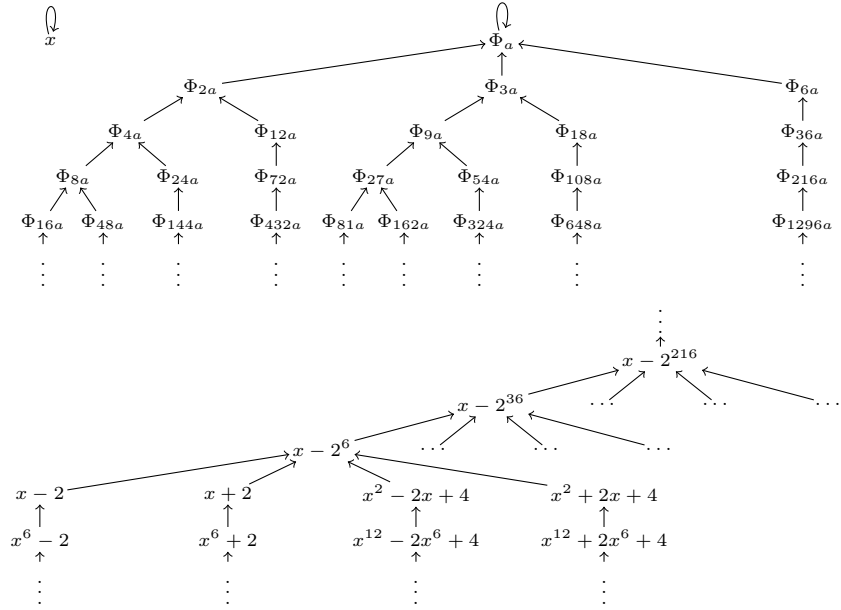

\begin{ex}
To suggest the graph structures induced by the Mahler and Gräffe operators,
we depict on Figure~\ref{fig:radical-of-Gräffe}
the graph of the radical~$\sqrt G$.
Applying~$M$ to some vertex~$p$ in the graph
results in the product of all antecedents under the map.
For example, $M (x-2^6) = (x-2)(x+2)(x^2-2x+4)(x^2+2x+4)$, and
$M \Phi_a = \Phi_a \Phi_{2a} \Phi_{3a} \Phi_{6a}$.
In the second example, $\Phi_a$~appears to the right
as a consequence of it being mapped to itself by~$G$.

The depicted case, $b = 6$, is typical for~$\bQ[x]$.
In particular, all cycles have length~1
as a consequence of the second part of Lemma~\ref{lem:out-of-cycle-for-G}.
\end{ex}

\subsection{Denominator bounds: algorithm}
\label{sec:den-algo}

Armed with the previous lemmas, we can now prove the key result that leads to
our main denominator bound.
Still, to avoid repetitions
in the proof of Proposition~\ref{prop:find-an-Mfactor} below,
we first state two intermediate lemmas.

The following lemma can be expressed more intuitively as follows:
for any~$\tilde f$ that is not on a cycle of~$\sqrt G$,
any~$g$ that appears on the tree rooted at~$\tilde f$
of antecedents under~$\sqrt G$
is also not on a cycle.

\begin{lem} \label{lem:M-cannot-lead-to-G-cycle}
Let $\tilde f \in \bK[x]$ be monic irreducible and satisfy
$G^i \tilde f \neq \tilde f$ for all~$i > 0$.
Further, let $g \in \bK[x]$~be monic irreducible and divide~$M^j \tilde f$
for some~$j \geq 0$.
Then $G^i g \neq g$ for all~$i > 0$.
\end{lem}

\begin{proof}
Suppose $G^i g = g$ for some~$i \geq 1$.
By Lemma~\ref{lemma:G-irred}, $G^j g$~is monic irreducible,
and since $G^j g \divides G^j M^j \tilde f = \tilde f^{b^j}$,
it must be~$\tilde f$.
Thus, $G^i \tilde f = G^{i+j} g = G^j g = \tilde f$,
in contradiction with the definition of~$\tilde f$.
\end{proof}

\begin{lem} \label{lem:drop-Tq}
Let $s \geq r-1$ and~$m \geq 1$~be integers,
and let
$f \in \bK[x]$ be monic irreducible,
$q \in \bK[x]$ be nonconstant,
and~$\ell \in \bK[x]$ be nonzero,
and such that $x \nmid q$,
$M^s f^m \divides M^r q \divides \ell \, T q$,
and $M^{s-i} f \wedge q = 1$ whenever $0 \leq i < r$.
Then $M^s f^m$~divides~$\ell$.
\end{lem}

\begin{proof}
Let~$h^k \divides M^s f^m$ for a monic irreducible~$h \in \bK[x]$ and~$k > 0$,
so that $h^k \divides \ell \, T q$.
We prove by contradiction that $h$~is coprime with~$T q$:
suppose there exists some~$i$ satisfying $0 \leq i <r$
such that $h$~divides~$M^i q$.
Then, $G^i h$~divides both $G^i M^i q$ and~$G^i M^s f$, which,
upon applying Lemma~\ref{lemma:MG}\ref{item:GM},
are equal to powers of $q$ and~$M^{s-i} f$, respectively.
This contradicts the coprimality of $q$~and $M^{s-i} f$.
We conclude that~$h^k \divides \ell$, and
the conclusion follows upon considering all~$h^k \divides M^s f^m$.
\end{proof}

The following proposition will be used implicitly as a termination test
in Algorithm~\ref{algo:bound-from-lr}:
as long as there exists a nonpolynomial rational solution~$p/q$,
the nonconstant polynomial~$u$ proved to exist
contains (potential) factors of~$q$
and can be used to change unknowns
in a way that lessens the degree of~$\ell_r$.
An interpretation of the structure of the proof is as follows:
\begin{itemize}
\item
  If some factor of~$q$ appears out of all cycles of~$\sqrt G$,
  there exists such a factor~$u$ with no other factor of~$q$
  in the tree rooted at~$u$,
  and this~$u$ satisfies~$M^r u \divides \ell$.
\item
  Otherwise, each factor~$f$ of~$q$ is on a cycle
  and leads to some antecedent~$\tilde f$ under~$\sqrt G$ that is on no cycle,
  for which $f$~divides~$G \tilde f$.
  Considering all possible~$f$ and taking multiplicities into account,
  we construct a polynomial~$u$
  such that $M^{r-1} u \divides \ell$ and~$q \divides G u$.
\end{itemize}

\begin{prop} \label{prop:find-an-Mfactor}
Let $\ell \in \bK[x]$~be a nonzero polynomial
and $q \in \bK[x]$~be a nonconstant polynomial
such that $x \nmid q$ and
$M^r q \mid \ell \, T q$.
Then there exists a nonconstant~$u \in \bK[x]$ such that:
\begin{itemize}
\item either $M^r u \divides \ell$,
\item or $M^{r-1} u \divides \ell$ and $q \divides G u$.
\end{itemize}
\end{prop}

\begin{proof}
We consider two cases, the first one being when
there exists a monic irreducible~$f$ dividing~$q$
such that $G^if \neq f$ for all~$i>0$.
In this case, we first prove that we can also assume without loss of generality
that $M^j f \wedge q = 1$ for all~$j > 0$.
Assume the contrary: that the gcd is nontrivial for at least one~$j > 0$.
By Lemma~\ref{lem:out-of-cycle-for-G}, the~$M^j f$
for~$j\in\bN$ are pairwise coprime,
and since $q$~has finitely many factors, $M^j f \wedge q \neq 1$ for at most
finitely many~$j$.
Set $j$~to the maximal possible value
and $g$~to a monic irreducible factor of~$M^j f \wedge q$.
Lemma~\ref{lem:M-cannot-lead-to-G-cycle} applied to $g$ and~$\tilde f = f$
implies that $G^i g \neq g$ for all~$i > 0$,
and $g$~can replace~$f$ with the added property on the~$M^j g$.
At this point, Lemma~\ref{lem:drop-Tq} applies with $s = r$ and~$m = 1$,
proving that $M^r f$~divides~$\ell$.
The proposition is proved in this case by choosing~$u = f$.

In the second case,
let $q = c \prod_k f_k^{m_k}$ be the irreducible factorization of~$q$,
for a nonzero constant~$c$
and two-by-two distinct monic irreducible~$f_k$,
and with, for each~$k$, some~$i_k>0$ satisfying $G^{i_k} f_k = f_k$.
Fix any~$k$.
Lemma~\ref{lem:M-exits-cycles}
provides a monic irreducible factor~$\tilde f_k\in\bK[x]$ of~$M f_k$ such that
$G^i \tilde f_k \neq \tilde f_k$ for all~$i > 0$.
If $M^i \tilde f_k \wedge q$ was nontrivial for some~$i \in \bN$,
this gcd would contain some monic irreducible factor~$g$,
necessarily equal to some~$f_{k'}$,
and Lemma~\ref{lem:M-cannot-lead-to-G-cycle} would contradict
the existence of~$i_{k'}$.
So the polynomials $M^j \tilde f_k$ are coprime with~$q$ for all~$j \in \bN$.
Upon setting $s = r - 1$, $m = m_k$, and~$g = \tilde f_k$,
$M^s g^m = M^{r-1} \tilde f_k^{m_k} \divides M^r f_k^{m_k} \divides M^r q$,
and $M^{s-i} g = M^{r-1-i} \tilde f_k$ is coprime with~$q$
for all~$i$ satisfying $0 \leq i < r$,
so that Lemma~\ref{lem:drop-Tq} proves
that $M^{r-1} \tilde f_k^{m_k} = M^s g^m$~divides~$\ell$.
Additionally, $G g = G \tilde f_k \divides G M f_k = f_k^b$,
so that $G g$ is a power of~$f_k$,
hence $f_k \divides G g = G \tilde f_k$, and next
$f_k^{m_k} \divides G \tilde f_k^{m_k}$.
Gathering the results over all~$k$,
the~$\tilde f_k$ are pairwise coprime because the~$f_k$ are;
it follows that all~$M^{r-1} \tilde f_k^{m_k}$
divide~$\ell$ and are pairwise coprime,
so that, finally,
the product $u = \prod_k \tilde f_k^{m_k}$ satisfies
$M^{r-1} u \divides \ell$ and $q \divides G u$.
\end{proof}

\begin{rem}
In the first case of the proof,
which builds~$u$ satisfying~$M^r u \divides \ell$,
it is of interest to compare the construction
with that in the case of usual recurrences~\cite{Abramov-1989-RSL}.
The obtained~$u$ is extremal,
in the sense that no other factor of~$q$ can be found
in the tree rooted at it,
that is to say by iterating~$\sqrt G$ backward from it;
this is used to compute~$u$ from the leading coefficient~$\ell$
of the Mahler operator.
In the case of usual recurrences,
the shift operator~$S$ (with respect to the variable~$n$)
and its inverse~$S^{-1}$
play roles similar to $M$ and~$\sqrt G$, respectively.
In Abramov's algorithm for denominator bounds,
poles are searched for by considering
poles that are extremal in a class~$\alpha + \bZ$:
in particular,
a pole~$\beta \in \alpha + \bZ$ with minimal real part
corresponds to a monic irreducible factor~$u = n - \beta$
such that $S^r u$~divides the leading coefficient~$\ell$
of the recurrence operator.
\end{rem}

\begin{cor} \label{cor:d-vs-r}
When $d < b^{r-1}$, Eq.~\eqref{eq:mahler-eqn} has no
nonconstant rational solution.
\end{cor}

\begin{proof}
With the notation above,
Lemma~\ref{lem:valuation-at-0} implies $\bar v = 0$.
If a nonconstant~$q$ could satisfy Eq.~\eqref{eq:Mdiv-coprime-r},
Proposition~\ref{prop:find-an-Mfactor} would apply,
inducing the contradiction~$b^{r-1} \leq \deg\ell_r \leq d$.
So $q$~is constant,
and Lemma~\ref{lem:admissible-degrees} applies
and proves $p$~is constant.
\end{proof}

Proposition~\ref{prop:find-an-Mfactor} forms the basis of
Algorithm~\ref{algo:bound-from-lr}, which repeatedly searches for factors
of the form $M^r u$ to ``be removed'' from~$\ell_r$ (while ``adding back''
other factors of strictly smaller degree) and accumulates the corresponding~$u$
into the denominator bound.
The update of $\ell$ at step~\ref{step:bound-from-lr:newl} of each loop
iteration can be viewed as a change of unknown functions of the form $y =
\tilde y/u_k$ in~\eqref{eq:mahler-eqn}.
The search for factors of the form~$M^r u$, respectively~$M^{r-1} \tilde u$,
uses the following property
(for radix $b^r$, resp.~$b^{r-1}$).

\begin{algo}
\inputs{
  A linear Mahler equation of the form~\eqref{eq:mahler-eqn}.
}
\outputs{
  A polynomial~$q^\star\in\bK[x]$.
}
\caption{Obtain a denominator bound from~$\ell_r$.\label{algo:bound-from-lr}}
\begin{enumerate}
\item \label{step:bound-from-lr:loop}
  Set~$\ell:=\ell_r$, then repeat for $k = 1, 2, \dots$:
  \begin{enumerate}
  \item \label{step:bound-from-lr:slice}
    write
    $\ell = \sum_{i=0}^{b^r-1} x^{i} M^r f_i$
    with $f_i \in \bK[x]$;
  \item \label{step:bound-from-lr:biggcd}
    set $u_k := \bigwedge_{i=0}^{b^r-1} f_i$;
  \item \label{step:bound-from-lr:newl}
    set $\ell := (\ell/M^r u_k) \bigvee_{i=0}^{r-1} M^i u_k$
  \end{enumerate}
  until~$\deg u_k = 0$, at which point set $t = k - 1$.
\item \label{step:bound-from-lr:tilde-u}
  Set $\tilde u := \bigwedge_{i=0}^{b^{r-1}-1} f_i$
  where
  $\ell = \sum_{i=0}^{b^{r-1}-1} x^{i} M^{r-1} f_i$.
\item \label{step:bound-from-lr:prod}
  Return $u_1\dotsm u_t \, (G \tilde u)$.
\end{enumerate}
\end{algo}

\begin{lem} \label{lemma:Mr-factor}
Let $f_0, \dots, f_{b-1},u \in \bK[x]$.
The polynomial
$\ell = Mf_0 + x \, Mf_1 + \dots + x^{b-1} \, Mf_{b-1}$
is divisible by~$M u$ if and only if $f_0, \dots, f_{b-1}$ are all
divisible by~$u$.
\end{lem}

\begin{proof}
The ``if'' part is clear.
Conversely, fix $i < b$, and assume that $Mu \mid \ell$.
Let~$\omega$ be a primitive $b$th root of unity.
Then, $Mu = (Mu)(\omega^j x) \mid \ell(\omega^j x)$ for all~$j$,
hence $Mu$ divides
\[
  \sum_{j=0}^{b-1} \omega^{-ij} \ell(\omega^j x)
  = b \, x^i M f_i.
\]
As $Mu \in \bK[x^b]$ and $i < b$, this implies $Mu \mid Mf_i$, and
$u \mid f_i$ by Lemma~\ref{lemma:MG}\ref{item:Mdiv}.
\end{proof}

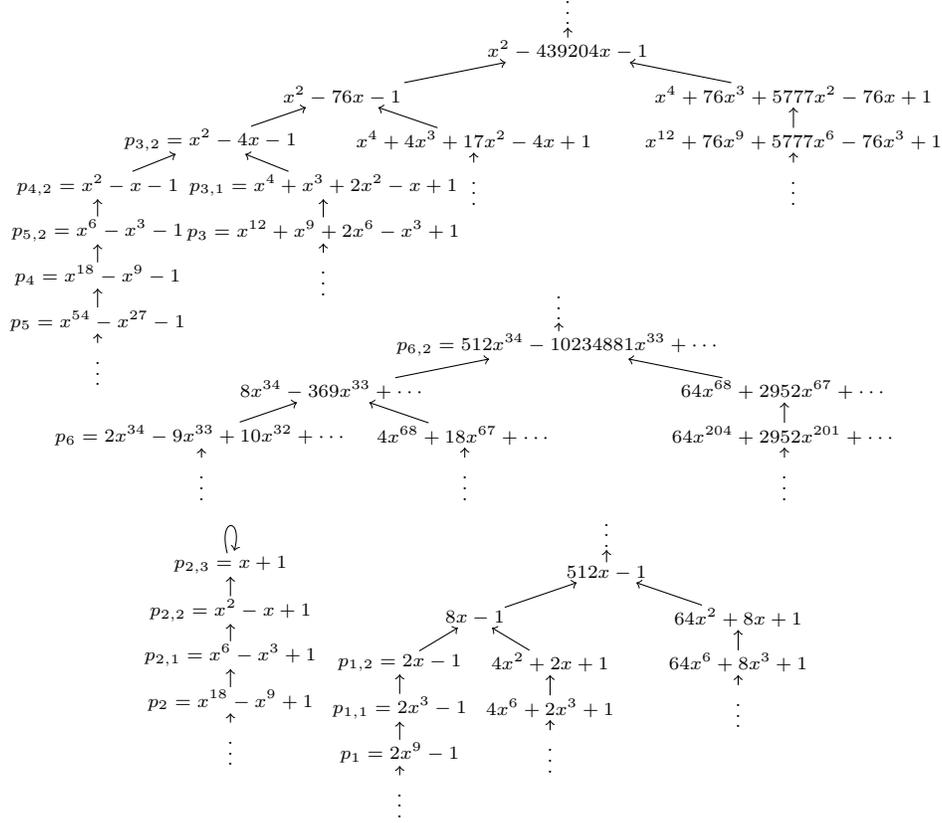
\begin{figure}
\begin{scriptsize}

\begin{tikzpicture}[level distance=6mm]
\tikzstyle{every node}=[inner sep=1pt]
\tikzstyle{level 1}=[sibling distance=20mm]
\tikzstyle{level 2}=[sibling distance=60mm]
\tikzstyle{level 3}=[sibling distance=35mm]
\tikzstyle{level 4}=[sibling distance=30mm]
\node {\vdots} [<-]
  child {node {$x^2-439204x-1$}
    child {node {$x^2-76x-1$}
      child {node {$p_{3,2} = x^2-4x-1$}
        child {node {$p_{4,2} = x^2-x-1$}
          child {node {$p_{5,2} = x^6-x^3-1$}
            child {node {$p_4 = x^{18}-x^9-1$}
              child {node {$p_5 = x^{54}-x^{27}-1$}
                child {node {\vdots}}
              }
            }
          }
        }
        child {node {$p_{3,1} = x^4+x^3+2x^2-x+1$}
          child {node {$p_3 = x^{12}+x^9+2x^6-x^3+1$}
            child {node {\vdots}}
          }
        }
      }
      child {node {$x^4+4x^3+17x^2-4x+1$}
        child {node {\vdots}}
      }
    }
    child {node {$x^4+76x^3+5777x^2-76x+1$}
      child {node {$x^{12}+76x^9+5777x^6-76x^3+1$}
        child {node {\vdots}}
      }
    }
  }
;
\end{tikzpicture}

\vskip-15mm

\begin{tikzpicture}[level distance=6mm]
\tikzstyle{every node}=[inner sep=1pt]
\tikzstyle{level 1}=[sibling distance=20mm]
\tikzstyle{level 2}=[sibling distance=60mm]
\tikzstyle{level 3}=[sibling distance=35mm]
\node {\vdots} [<-]
  child {node {$p_{6,2} = 512x^{34}-10234881x^{33}+\dotsb$}
    child {node {$8x^{34}-369x^{33}+\dotsb$}
      child {node {$p_6 = 2x^{34}-9x^{33}+10x^{32}+\dotsb$}
        child {node {\vdots}}
      }
      child {node {$4x^{68}+18x^{67}+\dotsb$}
        child {node {\vdots}}
      }
    }
    child {node {$64x^{68}+2952x^{67}+\dotsb$}
      child {node {$64x^{204}+2952x^{201}+\dotsb$}
        child {node {\vdots}}
      }
    }
  }
;
\end{tikzpicture}

\begin{tikzpicture}[level distance=6mm]
\tikzstyle{every node}=[inner sep=1pt]
\tikzstyle{level 1}=[sibling distance=20mm]
\tikzstyle{level 2}=[sibling distance=17mm]
\tikzstyle{level 3}=[sibling distance=10mm]
\node (phi_2) {$p_{2,3} = x+1$} [<-]
  child {node {$p_{2,2} = x^2-x+1$}
    child {node {$p_{2,1} = x^6-x^3+1$}
      child {node {$p_2 = x^{18}-x^9+1$}
        child {node {\vdots}}
      }
    }
  }
;
\path[->] (phi_2) edge [loop above] node {} ();
\tikzstyle{level 2}=[sibling distance=35mm]
\tikzstyle{level 3}=[sibling distance=20mm]
\node at (5, .5) {\vdots} [<-]
  child {node {$512x-1$}
    child {node {$8x-1$}
      child {node {$p_{1,2} = 2x-1$}
        child {node {$p_{1,1} = 2x^3-1$}
          child {node {$p_1 = 2x^9-1$}
            child {node {\vdots}}
          }
        }
      }
      child {node {$4x^2+2x+1$}
        child {node {$4x^6+2x^3+1$}
          child {node {\vdots}}
        }
      }
    }
    child {node {$64x^2+8x+1$}
      child {node {$64x^6+8x^3+1$}
        child {node {\vdots}}
      }
    }
  }
;
\end{tikzpicture}

\end{scriptsize}

\caption{Portion of the graph of the radical~$\sqrt G$
  of the Gräffe operator used for the resolution in Example~\ref{ex:rat-sol}.}
\label{fig:rat-sol}
\end{figure}

\begin{ex}\label{ex:rat-sol}
In this example, we let $b = 3$ and use Algorithm~\ref{algo:bound-from-lr}
to analyze the potential poles in rational-function solutions
of an operator
\[ L = \bigl(p_1(x) \dotsm p_6(x)\bigr) \, M^2 + \dotsb , \]
where the~$p_i$ are polynomials to be found in Figure~\ref{fig:rat-sol}
and the coefficients of $M^1$ and~$M^0$ will be disclosed below.
In the figure and this example, polynomials of large size are truncated
to their first few monomials,
and in most cases, we write them in factored form,
although polynomials are manipulated in expanded form in the actual algorithm.

Following Algorithm~\ref{algo:bound-from-lr}, we set~$\ell = p_1 \dotsm p_6$.
Step~\ref{step:bound-from-lr:loop} is motivated
by the first case in Proposition~\ref{prop:find-an-Mfactor}:
it strives to solve~\eqref{eq:Mdiv-coprime-r}
by finding a factor~$u$ of~$q$ such that~$M^2 u \mid \ell$.
For each~$i$, the only monic irreducible candidate factor of~$u$
that can ``cover''~$p_i$ upon application of~$M^2$
is the polynomial~$p_{i,2}$ in the figure.
However, $M^2 p_{i,2}$~consists of \emph{all} factors on the level of~$p_i$
with same ancestor~$p_{i,2}$.
So, for example, $M_2 p_{1,2} = p_1$ and $p_{1,2}$~can be part of~$u$,
whereas $M_2 p_{6,2}$~is a strict multiple of~$p_6$
so that $p_{6,2}$~cannot be made part of~$u$.
As a matter of fact, for~$k=1$ in the loop, the algorithm finds
$u_1 = p_{1,2} p_{2,2} p_{4,2} p_{5,2}$ at step~\ref{step:bound-from-lr:biggcd},
after rewriting~$\ell$ in the form
\begin{multline*}
\ell =
  - M^2\bigl( (2x-1)(x^2-x+1)(x^2-x-1)(x^6-x^3-1)(9x^5-133x^4+\dotsb) \bigr) + \dotsb + \\
  x^8 M^2\bigl( 2(2x-1)(x^2-x+1)(x^2-x-1)(x^6-x^3-1)(5x^4-74x^3+\dotsb) \bigr)
\end{multline*}
at step~\ref{step:bound-from-lr:slice}.
Step~\ref{step:bound-from-lr:newl} resets~$\ell$ to a polynomial
that factors into
\[ p_{1,1} \, p_{1,2} \, p_{2,1} \, p_{2,2} \, p_3 \, p_4 \, p_{5,2} \, p_{4,2} \, p_6. \]
Following the same approach for~$k=2$,
a new phenomenon occurs because of the loops in the graph:
the candidate factor~$p_{2,3}$ that would ``cover''~$p_{2,1}$
appears in its own tree on the same level as~$p_{2,1}$,
and thus has to be rejected.
It follows that the algorithm finds~$u_2 = p_{3,2} p_{4,2}$
at step~\ref{step:bound-from-lr:biggcd},
after rewriting~$\ell$ in the form
\begin{multline*}
\ell = - M^2\bigl( (x^2-4x-1)(x^2-x-1)(248x^5-5615x^4+\dotsb) \bigr) + \dotsb + \\
  x^8 M^2\bigl( (x^2-4x-1)(x^2-x-1)(532x^4-6211x^3+\dotsb) \bigr)
\end{multline*}
at step~\ref{step:bound-from-lr:slice}.
Step~\ref{step:bound-from-lr:newl} resets~$\ell$ to a polynomial
that factors into
\[ p_{1,1} \, p_{1,2} \, p_{2,1} \, p_{2,2} \, p_{3,1} \, p_{5,2} \, p_{4,2} \, p_{3,2} \, p_6. \]
Following the same approach for~$k=3$ leads to~$u_3 = 1$:
no further factor~$u$ of~$q$ exists
and helps solving~Eq.~\eqref{eq:Mdiv-coprime-r}
by ensuring~$M^2 u \mid \ell$.

This leads to step~\ref{step:bound-from-lr:tilde-u},
which is motivated by the second case in Proposition~\ref{prop:find-an-Mfactor}:
Eq.~\eqref{eq:Mdiv-coprime-r}~now implies~$M^2 q \mid q \wedge M q$,
which is solved by finding~$\tilde u$ such that $M \tilde u \mid \ell$.
A difference to step~\ref{step:bound-from-lr:loop}
is that at step~\ref{step:bound-from-lr:tilde-u}, candidates are looked for
just $2 - 1 = 1$ level above the factors to be ``covered''.
A similar calculation as previously explains that the algorithm finds
$\tilde u = p_{1,2} p_{2,2} p_{3,2} p_{4,2}$,
after rewriting~$\ell$ in the form
\begin{multline*}
\ell = M^2\bigl( (2x-1)(x^2-x+1)(x^2-4x-1)(x^2-x-1)(181x^{13}-1198x^{12}+\dotsb) \bigr) \\
  - x M^2\bigl( (2x-1)(x^2-x+1)(x^2-4x-1)(x^2-x-1)(44x^{13}-623x^{12}+\dotsb) \bigr) \\
  + x^2 M^2\bigl( (2x-1)(x^2-x+1)(x^2-4x-1)(x^2-x-1)(4x^{13}-382x^{12}+\dotsb) \bigr) .
\end{multline*}
From these factors, only~$p_{2,2}$ is cyclotomic.
But as the algorithm does not factor polynomials,
the other factors cannot be discarded.

At step~\ref{step:bound-from-lr:prod}, the algorithm returns the bound
\begin{equation*}
q^\star = u_1 u_2 G \tilde u =
  {p_{1,2}}^{1+1} {p_{2,2}}^2 {p_{3,2}}^{1+1} {p_{4,2}}^{2+1} p_{5,2} ,
\end{equation*}
where the~``$+1$'' indicate factors that could have been saved
if a cyclotomic test had been available.
The operator~$L$ was indeed constructed
so as to admit the two explicit rational solutions
\begin{equation*}
\frac{2x}{(2x-1)(x^2-x-1)}
\qquad\text{and}\qquad
\frac{x-3}{(x^2-x+1)(x^2-4x-1)(x^6-x^3-1)} ,
\end{equation*}
whose denominators are effectively ``covered'' by~$q^\ast$.

We remark that, during the steps of the algorithm,
the degree of~$\ell$ has dropped from its initial value~145
down to~84, then to~62.

\end{ex}

\begin{ex}\label{ex:rat}
Let $b=3$ and let us consider the Mahler equation
\begin{multline*}
L = (2x^4-x^3-x+3)(2x^9-1)(x^{18}-x^9-1) \, M^2 \\
    -(x^2+1)(2x^3-1)(x^4+1)(x^6-x^3-1)(2x^{10}-x^9-x+3) \, M \\
    +x^2(2x-1)(x^2+x+1)(x^2-x+1)(x^2-x-1)(2x^{12}-x^9-x^3+3) .
\end{multline*}
Following Algorithm~\ref{algo:bound-from-lr}, we expand
$(2x^4-x^3-x+3)(2x^9-1)(x^{18}-x^9-1)$ to get~$\ell$, which
step~\ref{step:bound-from-lr:slice} rewrites
\begin{multline*}
\ell = M^2(6x^3-9x^2-3x+3) + x M^2(-2x^3+3x^2+x-1) + {} \\
  x^3 M^2(-2x^3+3x^2+x-1) + x^4 M^2(4x^3-6x^2-2x+2) .
\end{multline*}
(That is, $f_2 = f_5 = f_6 = f_7 = f_8 = 0$.)
We get $u_1 = 2x^3-3x^2-x+1$, which factors into~$(2x-1)(x^2-x-1)$.
Step~\ref{step:bound-from-lr:newl} resets~$\ell$
to a polynomial that factors into
$(2x-1)(x^2-x-1)(2x^3-1)(2x^4-x^3-x+3)(x^6-x^3-1)$.
Expanding~$\ell$ as in step~\ref{step:bound-from-lr:slice},
we now find
\begin{multline*}
\ell =
M^2(3-10x) + x M^2(-4-15x) + x^2 M^2(-8-19x) + {} \\
  x^3 M^2(5+40x) + x^4 M^2(5-10x) + x^5 M^2(2x+9) + {} \\
  x^6 M^2(-25-16x) + x^7 M^2(15+8x) + x^8 M^2(23) ,
\end{multline*}
so that~$u_2=1$.
We pass to step~\ref{step:bound-from-lr:tilde-u},
which expands~$\ell$ in the form
\begin{multline*}
\ell = M(-16x^5+40x^4-10x^3-25x^2+5x+3) + {} \\
  x M(8x^5-10x^4-15x^3+15x^2+5x-4) + {} \\
  x^2 M(2x^4-19x^3+23x^2+9x-8) ,
\end{multline*}
and $\tilde u = (2x-1)(x^2-x-1)$.
So, $q^\star = u_1 \, G \tilde u = (2x-1)(x^2-x-1)(8x-1)(x^2-4x-1)$.
This means that if $y=p/(x^{\bar v} q)$ is solution of~$Ly=0$,
where $\bar{v} \geq 0$ and $p,q\in\bK[x]$ satisfy
$x\wedge q=p\wedge q=p\wedge x^{\bar v}=1$,
then $q$~divides~$q^\star$.
Using the results of~\S\ref{sec:structure},
we find that 0~could not be a pole of a solution in~$\bK(x)$
and therefore $\bar v=0$.
Consequently, $q^\star$ is a denominator bound.
\end{ex}

\begin{prop} \label{prop:bound-from-lr}
Algorithm~\ref{algo:bound-from-lr} runs
in $\bigO((\deg \ell_r) \, \Mult(d) \log d)$~ops if~$b = 2$,
resp.\ in $\bigO(b^{-r} \, (\deg \ell_r) \, \Mult(d) \log d)$~ops if~$b \geq 3$,
and computes a polynomial~$q^{\star}$ of degree
at most~$\deg \ell_r$ if~$b = 2$,
resp.\ at most~$(\deg \ell_r)/b^{r-1}$ if~$b \geq 3$,
such that any rational function solution~$y$ of~\eqref{eq:mahler-eqn}
can be written in the form $y = p/(x^{\bar v} q^\star)$ for some $p \in \bK[x]$
and ${\bar v} \in \bN$.
\end{prop}

\begin{proof}
For each~$k \geq 1$ reached by the loop~\ref{step:bound-from-lr:loop},
let $\tilde\ell_k$ denote the value of~$\ell$
considered at step~\ref{step:bound-from-lr:slice}, so that
the value assigned at step~\ref{step:bound-from-lr:newl} is~$\tilde\ell_{k+1}$.
(In particular, $\tilde\ell_1 = \ell_r$.)

First, observe that,
after step~\ref{step:bound-from-lr:biggcd} in each loop iteration,
$u_k$~is by Lemma~\ref{lemma:Mr-factor}
a polynomial of maximal degree such that $M^r u_k \mid \tilde\ell_k$.
In particular, the next value, $\tilde\ell_{k+1}$,
computed at step~\ref{step:bound-from-lr:newl}, is a polynomial.
Set $\rho = b^r - \frac{b^r - 1}{b - 1}$, which is at least~1.
Step~\ref{step:bound-from-lr:newl} decreases the degree of~$\ell$ by
\[
  \deg \tilde\ell_k - \deg \tilde\ell_{k+1}
  \geq \deg M^r u_k - \deg T u_k
  \geq \rho \deg u_k
  \geq \rho .
\]
In particular, the loop terminates after at most
$\rho^{-1} (1 + \deg \ell_r)$ iterations,
and therefore the whole algorithm terminates as well.
Second, after step~\ref{step:bound-from-lr:tilde-u},
$\tilde u$~is similarly a polynomial of maximal degree such that
$M^{r-1} \tilde u \mid \tilde\ell_{t+1}$.
Therefore, $b^{r-1} \deg \tilde u$ is bounded above
by the degree of~$\tilde\ell_{t+1}$,
so that
\[
  \biggl( \sum_{k=1}^t \rho \deg u_k  \biggr) + b^{r-1} \deg \tilde u
  \leq \biggl( \sum_{k=1}^t \deg \tilde\ell_k - \deg \tilde\ell_{k+1} \biggr)
       + \deg \tilde\ell_{t+1}
  \leq \deg \ell_r ,
\]
where $t$ denotes, as in Algorithm~\ref{algo:bound-from-lr},
the last value of~$k$ for which~$\deg u_k > 0$.
The output from the algorithm is $q^\star = u_1\dotsm u_t \, (G \tilde u)$.
If $b = 2$, then $\rho = 1$
and $\deg q^\star = \bigl(\sum_k \deg u_k\bigr) + \deg \tilde u$
is bounded by $\deg \ell_r$;
if $b \geq 3$, then
\[ \rho = b^{r-1} \left(b-2 + \frac{b-2}{b-1}\right) + \frac1{b-1} \geq b^{r-1} \]
and $\deg q^\star$
is bounded by $b^{-(r-1)} \deg \ell_r$.

Assume that $p/(x^{\bar v} q)$ is a solution written in lowest terms.
Set $\tilde q_0 = q$ and, for~$k$ between $1$ and~$t$, define the polynomials
$\tilde q_k = \tilde q_{k-1} / (u_k \wedge \tilde q_{k-1})$.
Let us prove by an induction on~$k$ that, for~$1 \leq k \leq t+1$:
\emph{(i)}~$x \nmid \tilde q_{k-1}$;
\emph{(ii)}~$M^r \tilde q_{k-1} \mid \tilde\ell_k \, T \tilde q_{k-1}$;
and \emph{(iii)}~$q \mid u_1 \dotsm u_{k-1} \tilde q_{k-1}$.
Initially when~$k = 1$, we have $\tilde q_0 = q$ and~$\tilde\ell_1 = \ell_r$,
so the three properties hold by our assumption on a solution
and Equation~\eqref{eq:Mdiv-coprime-r}.
Assume now that
$x \nmid \tilde q_{k-1}$,
$M^r \tilde q_{k-1} \mid \tilde\ell_k \, T \tilde q_{k-1}$
and $q \mid u_1 \dotsm u_{k-1} \tilde q_{k-1}$.
It follows from
$\tilde q_{k-1} = (u_k \wedge \tilde q_{k-1}) \, \tilde q_k \mid u_k \tilde q_k$
that
$x \nmid \tilde q_k$ and
$T \tilde q_{k-1} \mid (T u_k) \, (T \tilde q_k)$.
Furthermore,
\begin{equation}\label{eq:induction-divisibility-before-division}
(M^r (u_k \wedge \tilde q_{k-1})) \, (M^r \tilde q_k)
= M^r \tilde q_{k-1}
\mid \tilde\ell_k \, T \tilde q_{k-1}
\mid \tilde\ell_k \, (T u_k) \, (T \tilde q_k) .
\end{equation}
Write $M^r u_k = a_k \, M^r (\tilde q_{k-1} \wedge u_k)$
and $\tilde \ell_k = b_k \, M^r u_k$,
for suitable polynomials $a_k$ and~$b_k$.
Upon division by~$M^r (u_k \wedge \tilde q_{k-1})$,
Equation~\eqref{eq:induction-divisibility-before-division} becomes
\begin{equation}\label{eq:divisibility-with-a-b}
M^r \tilde q_k \mid a_kb_k \, (T u_k) \, (T \tilde q_k) .
\end{equation}
By construction, $a_k$~and~$M^r \tilde q_k$ are coprime,
as they are the cofactors of~$M^r(u_k \wedge \tilde q_{k-1})$
in, respectively, $M^r u_k$~and~$M^r \tilde q_{k-1}$,
so Equation~\eqref{eq:divisibility-with-a-b} finally becomes
\begin{equation*}
M^r \tilde q_k \mid
\frac{\tilde\ell_k}{M^r u_k} \, (T u_k) \, (T \tilde q_k) =
\tilde\ell_{k+1} \, T \tilde q_k .
\end{equation*}
By the divisibility assumption on~$q$ and the definition of~$\tilde q_k$,
\begin{equation*}
q \mid u_1 \dotsm u_{k-1} \tilde q_{k-1}
= u_1 \dotsm u_{k-1} \, (u_k \wedge \tilde q_{k-1}) \, \tilde q_k
\mid u_1 \dotsm u_k \tilde q_k ,
\end{equation*}
completing the proof by induction.

The loop terminates when $\ell$~no longer has
any nonconstant factor of the form~$M^r u$,
with $\ell = \tilde\ell_{t+1}$.
At this point, $M^r \tilde q_t \mid \tilde \ell_{t+1} \, T \tilde q_t$
and $q \mid u_1 \dotsm u_t \tilde q_t$.
If $\tilde q_t$~is constant, then $q \mid u_1 \dotsm u_t \mid q^\star$.
On the other hand, if $\tilde q_t$~is not constant,
Proposition~\ref{prop:find-an-Mfactor} applies, as~$x \nmid \tilde q_t$,
which implies that
$\tilde\ell_{t+1}$~admits a factor of the form~$M^{r-1} u$
such that $\tilde q_t \mid Gu$.
By Lemma~\ref{lemma:Mr-factor}, step~\ref{step:bound-from-lr:tilde-u}
computes a polynomial~$\tilde u$ such that $M^{r-1} u \mid M^{r-1} \tilde u$.
It follows by Lemma~\ref{lemma:MG}\ref{item:Mdiv} that $u \mid \tilde{u}$,
next that $\tilde q_t \mid G \tilde u$,
so that $q$~divides~$q^\star$, again.

Let us turn to the complexity analysis.
Applying~$M$ to a polynomial requires no arithmetic operation.
Each execution of step~\ref{step:bound-from-lr:biggcd} amounts
to $b^r-1$~gcds of polynomials of degree less than or equal to~$d/b^r$,
for a total cost of $\bigO(\Mult(d) \log d)$~ops.
The same argument applies to step~\ref{step:bound-from-lr:tilde-u}.
Similarly, the chain of lcms at step~\ref{step:bound-from-lr:newl}
requires
\[
  \bigO\biggl(\sum_{i=0}^{r-1} \Mult(b^i \deg u_k) \log(b^i \deg u_k)\biggr)
  = \bigO(\Mult(d) \log d)~\text{ops} ,
\]
as $(\sum_{i=0}^{r-1} b^i) \deg u_k = \bigO(d)$.
Since there are at most $\rho^{-1} (1+\deg\ell_r)$ iterations
of steps~\ref{step:bound-from-lr:biggcd} and \ref{step:bound-from-lr:newl},
the cost of step~\ref{step:bound-from-lr:loop}
is $\bigO(\rho^{-1} \, (\deg \ell_r) \, \Mult(d) \log d)$.
If $b = 2$, then $\rho = 1$~and the cost of step~\ref{step:bound-from-lr:loop}
is $\bigO((\deg \ell_r) \, \Mult(d) \log d)$.
If $b\geq 3$, then $\rho \geq b^{r-1}$ and the cost is
$\bigO(b^{-r} \, (\deg \ell_r) \, \Mult(d) \log d)$.

The computation of~$G \tilde u$ from~$\tilde u$
at step~\ref{step:bound-from-lr:prod}
can be performed in $\bigO(\Mult(bd))$~ops~\cite{BostanFlajoletSalvySchost-2006-FCS,Henrici-1986-ACC}
and the final product can be computed in $\bigO(\Mult(d) \log d)$~ops
using a product tree.
\end{proof}

Proposition~\ref{prop:bound-from-lr} implicitly provides a bound on
$\deg q$ that essentially (when $\bar v = 0$ and $\tilde u = 1$, exactly)
matches that of Bell and Coons~\cite[Proposition~2]{BellCoons-2017-TTM}.
However, a tighter bound holds, especially for $b=2$.

\begin{prop} \label{prop:den-deg-bound}
With the notation above, $q$~has degree at most $3 \deg \ell_r /b^r$.
\end{prop}

\begin{proof}
Let $g = M^r q \wedge Tq$.
On the one hand, \eqref{eq:Mdiv-coprime-r}~implies $M^r q \mid \ell_r g$,
so that $b^r \deg q \leq \deg \ell_r + \deg g$.
On the other hand, $Mg$ divides $h = M^r q \vee Tq$ by definition of~$T$,
hence $g \, Mg$ divides $gh = (M^r q) \, (Tq)$, whence
\[ (b + 1) \deg g \leq \frac{b^{r+1} - 1}{b-1} \deg q. \]
Comparing the two inequalities leads to
\[
    \deg q
    \leq \frac{(b^2 - 1) \deg \ell_r}{b^{r+2} - b^{r+1} - b^r + 1}
    \leq \frac{(b^2 -1) \deg \ell_r}{b^r (b^2 - b - 1)}
    \leq \frac{3 \deg \ell_r}{b^r}
\]
since $(b^2-1)/(b^2-b-1) \leq 3$ for $b \geq 2$.
\end{proof}

\begin{rem}
The previous discussion to find~$q^\star$
is entirely based on~\eqref{eq:Mdiv-general} in the case~$j = r$
and on expressing the solution~$y$ with a minimal denominator~$x^{\bar v} q$.
Noting that \eqref{eq:Mdiv-general}~actually holds
also for~$j \neq r$ and even if $p \wedge q \neq 1$,
we may apply it with $0 \leq j \leq r-1$
to a potential solution written in the form
$p/(x^{\bar v} q^\star)$
to get additional constraints involving
$\ell_0, \dots, \ell_{r-1}$
that can be used to remove some factors
from~$q^\star$.
\end{rem}

\subsection{An alternative bound}
\label{sec:den-alt}

We now describe an alternative method for computing denominator bounds.
While it yields coarser bounds, our estimate for its computational cost is
better, so that it may be a superior choice in some cases.
The results of this subsection are not used in the sequel.

\begin{prop} \label{prop:alt-denom-bound}
If $x^{\bar v} q \in \bK[x]$
is the denominator of a
rational solution of~\eqref{eq:mahler-eqn} written in lowest terms, then
it holds that
\[
  q \mid (G^r \ell_r) \, (G^{r+1} \ell_r) \dotsm (G^{r + K} \ell_r),
  \qquad
  K = \lfloor \log_b (3 \deg\ell_r) \rfloor - r.
\]
\end{prop}

\begin{proof}
Suppose $f$~is monic irreducible and $m$~is positive such that $f^m \mid q$,
and consider the condition
\begin{equation} \label{eq:altbnd-div}
  M^{r+j} f \mid \bigvee_{i=0}^r M^i q.
\end{equation}
Clearly, \eqref{eq:altbnd-div} is satisfied for~$j=0$, while it requires
\[ b^{r+j} \deg f \leq \frac{b^{r+1} - 1}{b - 1} \deg q, \]
which in turn implies $j \leq \log_b \deg q$.
Plugging in the bound from Proposition~\ref{prop:den-deg-bound},
we obtain $j \leq \log_b (3 \deg\ell_r) - r$.

Choose $j$ maximal such that \eqref{eq:altbnd-div}~holds.
Then $M^{r+j} f$ cannot divide~$Tq$,
and by Lemma~\ref{lem:f-neq-x}, $M^{r+j} f$ is squarefree.
Let~$h$ be a monic irreducible factor of~$M^{r+j} f$ not dividing~$Tq$.
In the rest of the proof,
we write~$\sqrfree p$ for the squarefree part of any polynomial~$p$.
For all~$k \geq 0$, set $h_k = \sqrfree(G^k h)$,
and denote by~$m_k$ the multiplicity of~$h_k$ as a factor of~$Tq$.
Thus $m_0$ is zero by definition of~$h$.
Continuing with~$k \geq 0$,
Lemma~\ref{lemma:MG}\ref{item:pGp} implies $G^k h \mid MG^{k+1} h$, so that
$h_k \mid \sqrfree(MG^{k+1} h) \mid M \sqrfree(G^{k+1} h) = M h_{k+1}$.
As $h_k^{m_k} \mid Tq$, we deduce that
$h_k^{m_{k+1}} \mid M h_{k+1}^{m_{k+1}} \mid M Tq \mid Tq \wedge M^r q$,
then, by using~\eqref{eq:Mdiv-coprime-r}, $h_k^{m_{k+1}} \mid \ell_r \, Tq$.
The definition of~$m_k$
then yields $h_k^{\delta_k} \mid \ell_r$
for $\delta_k = \max(m_{k+1} - m_k, 0)$.

Now, restrict~$k$ to the interval~$j < k \leq r+j$.
Then, by Lemma~\ref{lemma:MG}\ref{item:pGp},
\begin{equation}\label{eq:h_k-bounded-by-f-power}
h_k \mid G^k h \mid G^k M^{r+j} f = (M^{r+j-k} f)^{b^k} ,
\end{equation}
and as $h_k$~is squarefree, $h_k$~divides~$M^{r+j-k} f$.
Since~$f^m \mid q$ and $0 \leq r+j-k < r$, $h_k^m$~divides~$T q$,
implying~$m_k \geq m$.

By Lemma~\ref{lemma:MG}\ref{item:pGp}
and Equation~\eqref{eq:h_k-bounded-by-f-power},
$G^{r+j-k} h_k$ is~$f^{b^{r+j}}$,
so that $f$~divides the former.
Then,
\[ f^{\delta_k} \mid G^{r+j-k} h_k^{\delta_k} \mid G^{r+j-k} \ell_r . \]
Forming the product of these bounds for $k$ ranging from~0 to~$j$, we get
$f^m \mid \prod_{k=0}^j G^k \ell_r$,
as $m \leq m_{j+1}$ and~$m_0 = 0$.
The result follows by considering all possible~$(f,m)$
such that~$f^m \mid q$.
\end{proof}

\begin{prop}
One can compute a polynomial $q^\ast \in \bK[x]$ of degree at most
${d \, (\log_b d - r + 2)}$
and such that $q \mid q^\ast$ in
$\bigO(\Mult(d \log d) \log d)$ ops.
\end{prop}

\begin{proof}
If $\deg \ell_r < b^{r-1}$, return~$1$.
This is a valid bound by Corollary~\ref{cor:d-vs-r}.
Otherwise, return the bound from Proposition~\ref{prop:alt-denom-bound}.
As with the previous bound, the $G^k \ell_r$ up to
$k = r+K = \bigO(\log d)$
can be computed for a total of $\bigO(\Mult(b d) \log d)$
ops~\cite{BostanFlajoletSalvySchost-2006-FCS,Henrici-1986-ACC}.
The product then takes $\bigO(\Mult(d \log d) \log d)$ ops.
\end{proof}

\subsection{Computing numerators}
\label{sec:num}

In order to obtain a basis of rational solutions~$y$
of~\eqref{eq:mahler-eqn},
it suffices to obtain a bound~$x^{\bar v} q^\star$ on denominators
as in \S\ref{sec:den-algo},
to construct an auxiliary equation corresponding
to the change of unknown functions~$y = \tilde y / (x^{\bar v} q^\star)$,
and to search for its polynomial solutions~$\tilde y$.
We first note the following consequence of
Lemma~\ref{lem:admissible-degrees},
already proved by Bell and Coons~\cite[Prop.~2]{BellCoons-2017-TTM}.

\begin{prop} \label{prop:deg-num}
If $p, q \in \bK[x]$, not necessarily coprime, satisfy $L (p/q) = 0$,
then $\deg p$ is at most $\deg q + \lfloor d/(b^r - b^{r-1}) \rfloor$.
\end{prop}

\begin{algo}
\inputs{
  A linear Mahler equation of the form~\eqref{eq:mahler-eqn}.
}
\outputs{
  A basis of its space of rational function solutions.
}
\caption{Rational solutions}
\label{algo:rat-sol}
\begin{enumerate}
  \item \label{step:rat-sol:d}
    Set $\delta = \max \deg \ell_k$.
  \item \label{step:rat-sol:trivial-case}
    If $\delta < b^{r-1}$:
      return the basis~$(1)$ if $L(1) = 0$, and the empty basis~$()$ otherwise.
  \item \label{step:rat-sol:den-bound}
    Compute $q^\star$ using Algorithm~\ref{algo:bound-from-lr}.
    Set $\bar v = \lfloor \delta/(b^r - b^{r-1}) \rfloor$.
  \item \label{step:rat-sol:change-unknown}
    For $0 \leq k \leq r$, set
    $e_k = \lfloor b \delta/(b-1) \rfloor - b^k \bar v$
    and
    \[
      \tilde \ell_k = x^{e_k} \ell_k \;
      \prod_{\mathclap{0 \leq i \leq r, \ i \neq k}} M^i q^\star.
    \]
    Set $\tilde L = \tilde \ell_r M^r + \dots + \tilde \ell_0$.
  \item \label{step:rat-sol:solve}
    Call Algorithm~\ref{algo:poly-bounded-degree}
    on the equation
    $\tilde L p = 0$,
    with
    $w = \deg q^\star + 2 \bar v + 1$,
    to compute a basis $(p_1, \dots, p_\sigma)$ of its polynomial
    solutions of degree less than~$w$.
  \item \label{step:rat-sol:return}
    Return $( p_k/(x^{\bar v} q^\star) )_{1 \leq k \leq \sigma}$.
\end{enumerate}
\end{algo}

The procedure to obtain rational solutions
is summarized in Algorithm~\ref{algo:rat-sol}.

\begin{prop}\label{res:rat-sol}
Algorithm~\ref{algo:rat-sol} computes a basis of rational solutions
of its input equation.
Assuming~$d \geq b^{r-1}$, it runs in
$\softO(d \Mult(d) + 2^r d^2 + \Mult(2^r d))$~ops when~$b = 2$
and $\softO(b^{-r} d \Mult(d))$~ops when~$b \geq 3$.
Assuming further\/ $\Mult(n) = \softO(n)$,
it runs in $\softO(2^r d^2) = \softO(d^3)$~ops when~$b = 2$
and in $\softO(b^{-r}d^2)$~ops when~$b \geq 3$.
\end{prop}

\begin{proof}
Define $\delta$ as in step~\ref{step:rat-sol:d}, so that~$\delta \leq d$.
If $\delta < b^{r-1}$, the algorithm will stop
after step~\ref{step:rat-sol:trivial-case}.
In this case, Corollary~\ref{cor:d-vs-r} states
that there are no nonconstant rational solution.
Therefore, the vector space of rational solutions
is~$\bK$ when $L(1) = 0$ and $\{0\}$~otherwise.

Otherwise, the algorithm continues with~$d \geq  b^{r-1}$.
Assume that $y \in \bK(x)$ is a rational solution of $L y = 0$,
and let $p = x^{\bar v} q^\star y$ for
$q^\star$ and~$\bar v$ computed as in step~\ref{step:rat-sol:den-bound}.
By Proposition~\ref{prop:bound-from-lr} combined with
Lemma~\ref{lem:valuation-at-0}, $p$~is a polynomial.
By Proposition~\ref{prop:deg-num} combined with
Lemma~\ref{lem:admissible-degrees}, it has degree at most
$\deg (x^{\bar v} q^\star) + \bar v = \deg q^\star + 2 \bar v$.
Plugging $y = p/(x^{\bar v} q^\star)$ into $L y = 0$
and multiplying the resulting equation by the polynomial
$x^{\lfloor b \delta/(b-1) \rfloor} \prod_{i=0}^r M^i q^\star$,
we see that $p$ satisfies $\tilde L p = 0$,
where $\tilde L$~is defined as in step~\ref{step:rat-sol:change-unknown}.
As $b^k \bar v \leq b \delta/(b-1)$ for $k \leq r$,
the~$e_k$ are nonnegative
and the~$\tilde \ell_k$ are polynomials.
Thus Algorithm~\ref{algo:poly-bounded-degree} applies and,
by Proposition~\ref{prop:poly-bounded-degree},
$p$~belongs to the span of the~$p_k$
computed at step~\ref{step:rat-sol:solve} of Algorithm~\ref{algo:rat-sol}.
Conversely, for all~$k$, the fraction $p_k/(x^{\bar v} p^\star)$ is a
solution of $L y = 0$.

After step~\ref{step:rat-sol:trivial-case}, we have $b^r = \bigO(d)$,
that is, $r = \softO(1)$.
By Proposition~\ref{prop:bound-from-lr},
the cost of step~\ref{step:rat-sol:den-bound} is
$\softO(d \Mult(d))$~ops when $b = 2$ and
$\softO(b^{-r} d \Mult(d))$~ops when $b \geq 3$.
Define
\begin{equation}\label{eq:tilde-d}
  \tilde d
  = \frac{2b-1}{b-1} d + \frac{b^{r+1}-1}{b-1} \deg q^\star
  = \begin{cases}
      \bigO(2^r d), & b = 2 , \\
      \bigO(d),     & b \geq 3 ,
    \end{cases}
\end{equation}
where the asymptotic bounds follow from Proposition~\ref{prop:bound-from-lr}.
Each polynomial~$\tilde\ell_k$ defined at step~\ref{step:rat-sol:change-unknown}
then satisfies
\[
  \deg \tilde \ell_k
  \leq e_k + \delta + \frac{b^{r+1} - 1}{b-1} \deg q^\star
  \leq \tilde d ,
\]
so its computation as a product of $r+1$~factors can be done in
$\bigO( r \Mult(\tilde d) )$~ops.
This makes a total of
$\bigO( r^2 \Mult(\tilde d) ) = \softO( \Mult(\tilde d) )$~ops
to compute the~$\ell_k$'s.
Observe as well that $1 \leq w = \bigO(\tilde d / b^r)$.
According to Proposition~\ref{prop:poly-bounded-degree},
step~\ref{step:rat-sol:solve} thus requires
$\softO(b^{-r} \tilde d^2 + \Mult(\tilde d))$~ops,
which dominates the cost of step~\ref{step:rat-sol:change-unknown}.
Taking the bounds~\eqref{eq:tilde-d} into account, we get that
step~\ref{step:rat-sol:solve} is dominated by step~\ref{step:rat-sol:den-bound}
when~$b\geq3$,
so that the total cost is
$\softO(d \Mult(d) + 2^r d^2 + \Mult(2^r d))$~ops when~$b = 2$
and $\softO(b^{-r} d \Mult(d))$~ops when~$b \geq 3$.
With fast multiplication, $\Mult(n) = \softO(n)$, this simplifies
to the announced complexity estimates.
\end{proof}

\begin{ex}\label{ex:rat-suite}
We continue Example \ref{ex:rat}.
We have seen that the denominator bound is
$q^\star = (2x-1)(x^2-x-1)(8x-1)(x^2-4x-1)$.
We set $\tilde y = q^\star y$,
so that $Ly=0$ if and only if
$\tilde L \tilde y = 0$,
where $\tilde{L}=\tilde{\ell_2}M^2+\tilde{\ell_1}M^{1}+\tilde{\ell_0}$ for
\begin{align*}
\tilde\ell_2 &= (2x-1)(8x-1)(x^2-x-1)(x^2-4x-1) \times {} \\
  &\qquad (4x^2+2x+1)(2x^4-x^3-x+3)(x^4+x^3+2x^2-x+1) ,
  \displaybreak[0]\\
\tilde\ell_1 &= -(8x-1)(x^2+1)(x^2-4x-1)(2x^3-1)(x^4+1) \times {} \\
  &\qquad (x^6-x^3-1)(4x^6+2x^3+1)(2x^{10}-x^9-x+3)(x^{12}+x^9+2x^6-x^3+1) ,
  \displaybreak[0]\\
\tilde\ell_0 &= x^2(2x-1)(x^2+x+1)(x^2-x+1)(x^2-x-1) \times {} \\
  &\qquad (4x^2+2x+1)(2x^3-1)(x^4+x^3+2x^2-x+1) \times {} \\
  &\qquad (x^6-x^3-1)(4x^6+2x^3+1)(x^{12}+x^9+2x^6-x^3+1)(2x^{12}-x^9-x^3+3) .
\end{align*}

We have to compute the complete set of polynomial solutions of
$\tilde{L}\tilde{y}=0$.
The degree of $\tilde{\ell_2},\tilde{\ell_1},\tilde{\ell_0}$
are respectively 16, 46, 54.
Using Lemma~\ref{lem:admissible-degrees}, we find that
the degree of a nonzero polynomial solution is necessarily $4$ or~$5$.
Following Algorithm~\ref{algo:poly-sols-basis},
we equate the coefficients on both sides of $\tilde{L}\tilde{y}=0$
up to degree~54,
and we obtain that $\tilde y = \tilde y_0 + \dots + x^5 \tilde y_5$
is solution of $\tilde{L}\tilde{y}=0$
if and only if the vector $(\tilde y_0,\dots,\tilde y_5)$
is solution of a system of $h=163$~equations.
A basis of solutions turns out to consist of
$(2x-1)(8x-1)(x^2-4x-1)$ and $(x^2-x-1)(8x-1)(x^2-4x-1)$.
Consequently, a basis of rational-function solutions of~$Ly=0$
consists of
\[ \frac{1}{2x-1} \hbox{ and } \frac{1}{x^2-x-1}. \]
\end{ex}

\begin{rem}
When Mahler equations are considered in difference Galois theory%
~\cite{DreyfusHardouinRoques-2015-HSM,Roques-2018-ARB},
the interest tends to be in base fields on which
$M$~acts as an automorphism,
such as $\Puiseux$ and~$\ratram=\bigcup_{n=1}^{+\infty} \bK (x^{1/n})$.
By combining the strategy of Algorithm~\ref{algo:rat-sol}
with Proposition~\ref{prop:puiseux-slopes-denominators-lcm}
about possible ramifications,
we obtain an algorithm that computes a basis of
solutions of~\eqref{eq:mahler-eqn} in $\ratram$.
Assuming $\Mult(n) = \softO(n)$,
it runs in $\softO(2^{3r} d^3)$~ops when~$b = 2$
and in $\softO(b^{r}d^2)$~ops when~$b \geq 3$.
Note that, as in \S\ref{sec:Puiseux}, these complexity bounds hold
even if $\ell_0$~is zero.
\end{rem}

\subsection{Testing transcendence}\label{sec:transcendence}

As was announced in the introduction,
solving Mahler equations relates
to testing the transcendence of Mahler functions.
In particular, when computing the rational solutions
of a Mahler equation~\eqref{eq:mahler-eqn}
shows that there are no nonzero rational solutions,
this is a proof that all solutions to~\eqref{eq:mahler-eqn}
are transcendental.
We compare in this section the complexity of
the transcendence test by Bell and Coons~\cite{BellCoons-2017-TTM}
with that of a test by our rational solving.

To this end, we briefly sketch
Bell and Coons' “universal” transcendence test \cite{BellCoons-2017-TTM}
and do a complexity analysis of their approach,
using our notation and the same level of sophistication with regard to
algorithms for subtasks.
Define
\begin{gather*}
  \kappa_1 = \left\lfloor \frac{(b-1) \, d}{b^{r+1}-2b^r+1} \right\rfloor , \
  \kappa_2 = \left\lfloor \frac{d / (b-1)}{b^{r-1}} \right\rfloor , \
  \kappa = \kappa_1 + \kappa_2 + 1 , \
  B = d + \kappa \frac{b^{r+1}-1}{b-1} .
\end{gather*}
Bell and Coons \cite[Proposition~2 and Lemma~1]{BellCoons-2017-TTM}
show that any rational solution~$p/q$ to~\eqref{eq:mahler-eqn}
without pole at~0
satisfies $\deg q \leq \kappa_1$, $\deg p \leq \kappa_1 + \kappa_2$,
and that if a series~$y \in \bK[[x]]$ solves~\eqref{eq:mahler-eqn},
then either $y - p/q \neq \bigO(x^{B+1})$ or $y = p/q$ as series.
Then, given~$y = y_0 + y_1x + \dotsb$, Bell and Coons consider the matrix
$M = (y_{i+j})_{0\leq i\leq\kappa, \ 0\leq j\leq B}$,
whose $i$th row represents the truncation
up to~$\bigO(x^{B+1})$
of the non-singular part of~$y/x^i$.
To any nonzero~$\tilde q$ in the left kernel of~$M$,
they associate the polynomial $q = \tilde q_\kappa + \dots + \tilde q_0 x^\kappa$
and find a polynomial~$p$ of degree at most~$\kappa$
such that $y - p/q = \bigO(x^{B+1})$,
therefore such that $y = p/q$.
This leads to the equivalence that $M$~is full rank if and only if
$y$~is transcendental.
Bell and Coons' test therefore consists
in computing the truncation of~$y$
up to~$\bigO(x^{B+\kappa+1})$,
in forming the matrix~$M$,
and in determining if $M$~is of full rank, $\kappa+1$.
Only considering the linear-algebra task, which will dominate the complexity,
Bell and Coons' approach takes $\bigO(B\kappa^{\omega-1})$~ops,
by the algorithm of Ibarra, Moran and Hui~\cite{IbarraMoranHui-1982-GFL}.
When $b = 2$, we get $\kappa = \bigO(d)$, $B = \bigO(2^r d)$,
and a complexity~$\bigO(2^r d^\omega)$;
for~$b \geq 3$, we get $\kappa = \bigO(d/b^r)$, $B = \bigO(d)$,
and a complexity~$\bigO(d^\omega / b^{(\omega-1)r})$.
In either case, the dependency in~$d$ is in~$\bigO(d^\omega)$,
being not as good as the~$\softO(d^2)$ that can be obtained
by Algorithm~\ref{algo:rat-sol},
as Proposition~\ref{res:rat-sol} justifies.

In situations where \eqref{eq:mahler-eqn}~has nonzero rational solutions,
a given series solution $y \in \bK[[x]]$ can easily be tested
to be one of them, in $\bigO(r^\omega d) + \softO(rd)$~ops,
because only $\lfloor\nu\rfloor+1 = \bigO(d)$ initial coefficients
of solutions identify them (see~\S\ref{sec:approximate-series-solutions}).
So in all cases our Algorithm~\ref{algo:rat-sol} induces a transcendence test
in better complexity with respect to~$d$
than with the approach of~\cite{BellCoons-2017-TTM}.

\section{\texorpdfstring{The case $\ell_0=0$}{The case ℓ₀ = 0} and an algorithm for computing gcrd's}
\label{sec:ell0neq0}

In this section, we drop the assumption $\ell_0\neq 0$.
More precisely, we consider
a linear Mahler equation of the form \eqref{eq:mahler-eqn},
with $\ell_0=\dots=\ell_{w-1}=0$ and $\ell_r\ell_w\neq 0$.
We call the integer~$w$ the \emph{$M$-valuation of~\eqref{eq:mahler-eqn}}
and $d=\max_{k=w,\dots,r}\deg \ell_k$ its \emph{degree}.
We define the \emph{$M$-valuation} and the \emph{degree}
of the corresponding operator~\eqref{eq:mahler-opr} similarly.
The goal of this section is to compute a linear Mahler equation
with $M$-valuation equal to~0,
such that the new equation and~\eqref{eq:mahler-eqn}
have the same set of series solutions in~$\bK((x))$.

The algorithm proposed here, Algorithm~\ref{algo:ell0},
can be seen as an improvement
over an algorithm given by Dumas in his thesis
\cite[\S3.2.1]{Dumas-1993-RMS}.
In particular, Algorithm~\ref{algo:split-ell0},
borrowed from~\cite{Dumas-1993-RMS},
performs the subtask
of splitting an operator of positive $M$-valuation
into a system of operators of zero $M$-valuation
while preserving the solution set in~$\bK((x))$.
Dumas's algorithm next makes use of the right Euclidean structure of the
algebra~$\mathcal M(\bK)$ of linear Mahler operators
with coefficients in~$\bK(x)$,
and transforms the system into a single, equivalent equation
by computing a gcrd (greatest common right divisor) via Euclidean divisions.
The problem of this approach
is that the degree of the obtained equation explodes
in the process.
To avoid this, we change
the second step of algorithm in~\cite{Dumas-1993-RMS}
so as to reuse Algorithm~\ref{algo:split-ell0}
and cancellations of trailing instead of leading coefficients.

The splitting process of Algorithm~\ref{algo:split-ell0}
is explained in terms of \emph{section maps~$S_i$},
each of which maps a polynomial in $x$ and~$M$ to a polynomial in $x$ and~$M$,
and whose collection plays the role of a partial inverse for~$M$:
for $0 \leq i < b$, let $S_i$ be the $\bK$-linear map
that sends $x^jM^{k+1}$ to $x^{(j-i)/b}M^k$ if $(j-i)/b$ is an integer
and to~0 otherwise.

\begin{algo}
\inputs{
  A linear Mahler operator $L$ with coefficients in $\bK[x]$.
}
\outputs{
  A set of linear Mahler operators with coefficients in $\bK[x]$
  and $M$-valuation zero.
}
\caption{\label{algo:split-ell0}Split of $\eqref{eq:mahler-opr}$.}
\begin{enumerate}
\item If $L=0$, return $\emptyset$.
\item If $L$ has $M$-valuation~0, return $\{L\}$.
\item Return the union of the results of calling the algorithm recursively
  on each section $S_i(L)$ for $0 \leq i < b$.
\end{enumerate}
\end{algo}

\begin{lem}
Let $L$ be a linear Mahler operator~$L$ of the form~\eqref{eq:mahler-opr}
and have degree~$d$ and positive $M$-valuation.
Then, whenever~$0 \leq i < b$,
the section~$S_i(L)$ has degree at most~$d/b$.
Additionally, $L$~can be reconstructed from its sections by
\begin{equation}\label{eq:recon-from-sections}
L = \sum_{i=0}^{b-1} x^i M \, S_i(L) .
\end{equation}
\end{lem}

\begin{proof}
The degree bound and relation~\eqref{eq:recon-from-sections}
are shown by immediate calculations.
\end{proof}

\begin{lem}\label{lem:norm}
  Let $L$ be a linear Mahler operator of the form~\eqref{eq:mahler-opr},
  with order~$r$, $M$-valuation~$w$,
  and degree~$d$.
  Then, Algorithm \ref{algo:split-ell0} returns a set
  of nonzero linear Mahler operators
  of order at most~$r-w$, $M$-valuation~0,
  and degree at most~$ db^{-w}$.
\end{lem}

\begin{proof}
This is shown by a straightforward induction on~$w$.
\end{proof}

\begin{algo}
\inputs{
  A nonzero linear Mahler operator~$L$ of the form \eqref{eq:mahler-opr},
  order $r$, $M$-valuation $w$, and degree $d$.
}
\outputs{
  A linear Mahler operator~$\tilde L$
  of order $\tilde{r}\leq r-w$, $M$-valuation~0,
  and degree $\tilde{d}\leq db^{-w}$.
}
\caption{\label{algo:ell0}Normalization to $\ell_0\neq 0$.}
\begin{enumerate}
\item\label{it:ell0:initial-split}
Let $\mathcal L$ be the result of applying Algorithm~\ref{algo:split-ell0} to~$L$.
\item\label{it:ell0:while}
While $\mathcal L$ has at least two elements:
  \begin{enumerate}
  \item\label{it:ell0:while:choice}
    choose $L_1$ with highest order in~$\mathcal L$,
    then $L_2$ from~$\mathcal L \setminus \{L_1\}$;
  \item\label{it:ell0:while:norm}
    compute the result~$\mathcal L'$
    of applying Algorithm~\ref{algo:split-ell0}
    to the interreduction~$\Spoly(L_1, L_2)$;
  \item\label{it:ell0:while:update}
    replace $\mathcal L$ by $(\mathcal L \setminus \{L_1\}) \cup \mathcal L'$.
  \end{enumerate}
\item\label{it:ell0:return}
Return the element~$\tilde L$ of the singleton~$\mathcal L$.
\end{enumerate}
\end{algo}

Instead of considering usual Euclidean divisions according to decreasing powers,
which would compute a gcrd as in~\cite{Dumas-1993-RMS},
we use in Algorithm~\ref{algo:ell0}
linear combinations that kill constant terms:
given two nonzero Mahler operators $L_1$ and~$L_2$
with coefficients in $\bK[x]$, $M$-valuation zero,
and coefficient of~$M^0$ respectively $c_1$ and~$c_2$,
we write $\Spoly(L_1, L_2)$ for the operator $c_2 L_1 - c_1 L_2$,
whose coefficient of~$M^0$ is zero.
We call this operator the \emph{interreduction} of $L_1$ and~$L_2$
and a step of the algorithm that replaces an operator~$L_1$
by an interreduction~$\Spoly(L_1, L_2)$ a \emph{reduction step}.

\begin{lem}\label{lem:replace}
Let $\mathcal L$ be a system of Mahler operators.
Replacing an element~$L$ of~$\mathcal L$ by its sections $S_0(L), \dots, S_{b-1}(L)$
does not change the set of solutions of~$\mathcal L$ in~$\bK((x))$.
Nor does replacing $L_1$ by the interreduction~$\Spoly(L_1, L_2)$
where $L_1, L_2$ are distinct elements of~$\mathcal L$.
\end{lem}

\begin{proof}
The second claim is obvious.
Regarding the first one (already in~\cite[\S3.2.1]{Dumas-1993-RMS}),
the decomposition~\eqref{eq:recon-from-sections} shows
that any common solution of the $S_i(L)$ is a solution of~$L$.
If, conversely, $y$~is an \emph{unramified} solution of~$L$, then the
$x^i M S_i(L) \, y$, $0 \leq i < b$, have disjoint support,
hence $S_i(L) \, y = 0$ for all~$i$.
\end{proof}

Here, the degree of~$\Spoly(L_1, L_2)$ may well be
the sum of the degrees of $L_1$ and~$L_2$,
but having generated a multiple of~$M$
makes it possible to apply splitting and keep degrees under control.
This leads to Algorithm~\ref{algo:ell0},
whose correctness and complexity are given in the following proposition.

It is worth mentioning that, in general,
the equation $\tilde{L}(y)=0$ returned by Algorithm~\ref{algo:ell0}
does not have the same set of solutions in~$\Puiseux$
as the equation $L(y)=0$.
As an example, let $b=2$ and consider $L=M^2-xM$.
We have $\tilde{L}=1$, and
the solution space in~$\Puiseux$ of $\tilde{L}(y)=0$ is $\{0\}$.
On the other hand, the solution space in~$\Puiseux$ of $L(y)=0$
is the $\bK$-vector space spanned by~$x^{1/2}$.

\begin{prop}\label{propo:ell0}
The operator~$L$ has the same set of solutions in~$\bK((x))$
as the operator~$\tilde{L}$ returned by Algorithm \ref{algo:ell0}.
This operator has order $\tilde{r}\leq r-w$, $M$-valuation~0,
and degree $\tilde{d}\leq db^{-w}$.
Furthermore, Algorithm~\ref{algo:ell0} runs in~$\bigO(r b^r \Mult(d/b^w))$~ops.
\end{prop}

\begin{proof}
Because~$L\neq0$ and by construction of Algorithm~\ref{algo:split-ell0},
the initial set~$\mathcal L$ is nonempty.
Next, by construction of Algorithm~\ref{algo:ell0},
at any time of a run,
$\mathcal L$~is nonempty and
contains only elements of outputs from Algorithm~\ref{algo:split-ell0},
so that, by Lemma~\ref{lem:norm},
if Algorithm~\ref{algo:ell0} terminates,
its output must be nonzero and of $M$-valuation zero.
Lemma~\ref{lem:replace} implies that
the original operator~$L$,
the system~$\mathcal L$ at any time of the run,
and therefore the final operator~$\tilde L$,
all share the same set of solutions in~$\bK((x))$.

Let us prove the bound on the order and the degree of $\tilde{L}$.
By Lemma~\ref{lem:norm}, the set~$\mathcal L$ computed
at step~\ref{it:ell0:initial-split}
consists of Mahler operators with
orders bounded by~$r-w$ and degrees bounded by~$db^{-w}$.
These bounds keep on holding
after each run of the loop body at step~\ref{it:ell0:while}:
As the operators $L_1$ and~$L_2$ chosen at step~\ref{it:ell0:while:choice}
satisfy the property,
their combination~$\Spoly(L_1, L_2)$ (including the case it is zero) has
order bounded by~$r-w$, degree bounded by~$2db^{-w}$,
and positive valuation.
By Lemma~\ref{lem:norm}, the set~$\mathcal L'$ computed
at step~\ref{it:ell0:while:norm}
consists of Mahler operators with
orders bounded by~$r-w-1$ and degrees bounded by~$2db^{-(w+1)}$.
As $2/b \leq 1$, the set~$\mathcal L$ retains the property
after the update at step~\ref{it:ell0:while:update}.
Therefore, if the algorithm terminates,
it returns at step~\ref{it:ell0:return} an element of~$\mathcal L$,
therefore with the announced order and degree bounds.

We finally prove termination and complexity by a joint argument.
To this end, we represent the process of Algorithm~\ref{algo:ell0}
by an oriented tree labeled by operators~$L_w^n$,
for integers~$n$ and words~$w$ on the alphabet~$\{0,\dots,b-1\}$.
These operators~$L_w^n$ will be the operators considered
during the execution of the algorithm.
This tree is rooted at the node labeled~$L_\epsilon^0 = L$,
and evolves by following the execution of Algorithm~\ref{algo:ell0}.
Each time a section of an operator~$L_w^n$ is computed
by the subtask of Algorithm~\ref{algo:split-ell0},
whether it be
at step~\ref{it:ell0:initial-split} or at step~\ref{it:ell0:while:norm},
the tree is augmented by new edges from~$L_w^n$
to its subsection~$L_{wj}^n = S_j(L_w^n)$.
For each choice of $L_1 = L_w^n$ and~$L_2 = L_{w'}^{n'}$
at step~\ref{it:ell0:while:choice},
the tree is augmented by a new edge labeled~$L_{w'}^{n'}$,
from~$L_w^n$ to~$L_\epsilon^{m+1}$,
if $m$~is the larger upper index in the tree before reduction.
Thus, one obtains that at each stage of the execution,
the set~$\mathcal L$ is equal to the collection of nonzero leaves of the current tree.
Now, by construction of the tree and by design of the algorithm,
a reduction step results either in a zero operator
or in an operator with positive $M$-valuation
that is immediately split to its sections.
Therefore, following a path from the root to a leaf,
two reduction edges can only appear if separated by at least one section edge.
As section edges reduce orders by at least~1,
while reduction edges do not increase orders,
the tree has to be finite
and the algorithm terminates.
The only arithmetic operations of the algorithm are the polynomial products
involved in the computation of the~$\Spoly(L_1, L_2)$
at step~\ref{it:ell0:while:norm}.
It was proved above that
any operator of~$\mathcal L$ has degree bounded by~$d/b^w$.
Because operators all have order at most~$r$ and
as the size of the tree bounds the number of reductions,
the algorithm has total complexity $\bigO(r b^r \Mult(d/b^w))$.
\end{proof}

\begin{rem}
A slightly better complexity can be obtained
by a variant of Algorithm~\ref{algo:ell0},
in which the~$L_1$ at step~\ref{it:ell0:while:choice}
is not chosen as having maximal order,
but according to a notion of depth
in the tree introduced for the proof of Proposition~\ref{propo:ell0}.
Doing so guarantees a better behavior of degrees,
with a geometric decrease with depth,
as opposed to the uniform bound~$d/b^w$ used in the proof above.

Define the depth~$\beta$ of a node~$L_w^n$ in the tree
as the number of section edges from the root~$L_\epsilon^0$ to~$L_w^n$,
and change the strategy at step~\ref{it:ell0:while:choice}
to choose~$L_1$ among the elements of~$\mathcal L$
of lowest depth.
By another induction,
$L_w^n$~has order not more than~$r - \beta$, as in the proof above,
but its degree is
not more than $d/b^w$ if~$\beta \leq w$,
and not more than $(2/b)^{\beta-w} (d/b^w)$ if~$\beta > w$.
A bound on the complexity becomes
\begin{equation*}
\sum_{\beta=w}^r (r+1) b^\beta \Mult\left( \frac{2^{\beta-w} d}{b^\beta} \right)
  \leq \bigO\left( r \Mult(2^r d/b^w) \right) .
\end{equation*}
This bound is better than the original complexity $\bigO(r b^r \Mult(d/b^w))$
when~$b \geq 3$.
For~$b=2$, the new bound is not tight and the variant algorithm
has the same complexity bound as Algorithm~\ref{algo:ell0}.
\end{rem}

\begin{ex}\label{ex:reduction}
\begin{figure}
\begingroup
\tiny
\tikzstyle{operator}=[draw,rectangle,text width=9.5mm,text centered]
\tikzstyle{section}=[->,>=latex,color=blue]
\tikzstyle{reduction}=[->,>=latex,color=red]
\tikzstyle{helper}=[right]
\begin{tikzpicture}[xscale=3.5,yscale=1.2]
  \node[operator] (L1) at (0,-1) {$L_\epsilon^0 = L$ $(4, 147)$} ;
  \node[operator] (L3) at (-1,-2) {$L_0^0$ $(3, 49)$} ;
  \draw[section] (L1) -- (L3) ;
  \node[operator] (L4) at (0,-2) {$L_1^0 = 0$} ;
  \draw[section] (L1) -- (L4) ;
  \node[operator] (L5) at (1,-2) {$L_2^0$ $(3, 45)$} ;
  \draw[section] (L1) -- (L5) ;
  \node[operator] (L15) at (1/2,-4) {$L_{20}^0$ $(2, 15)$} ;
  \draw[section] (L5) -- (L15) ;
  \node[operator] (L16) at (1,-4) {$L_{21}^0$ $(2, 12)$} ;
  \draw[section] (L5) -- (L16) ;
  \node[operator] (L17) at (3/2,-4) {$L_{22}^0$ $(2, 13)$} ;
  \draw[section] (L5) -- (L17) ;
  \node[operator] (L18) at (-1,-3) {$L_\epsilon^1$ $(3, 58)$} ;
  \draw[reduction] (L3) to node[helper] {$L_{21}^0$} (L18) ;
  \node[operator] (L54) at (-3/2,-4) {$L_0^1$ $(2, 18)$} ;
  \draw[section] (L18) -- (L54) ;
  \node[operator,thick] (L55) at (-1,-4) {$L_1^1$ $(2, 19)$} ;
  \draw[section] (L18) -- (L55) ;
  \node[operator] (L56) at (-1/2,-4) {$L_2^1$ $(2, 16)$} ;
  \draw[section] (L18) -- (L56) ;
  \node[operator] (L57) at (1,-5) {$L_\epsilon^2 = 0$} ;
  \draw[reduction] (L16) to node[helper] {$L_{22}^0$} (L57) ;
  \node[operator] (L58) at (3/2,-5) {$L_\epsilon^3 = 0$} ;
  \draw[reduction] (L17) to node[helper] {$L_{20}^0$} (L58) ;
  \node[operator] (L59) at (1/2,-5) {$L_\epsilon^4 = 0$} ;
  \draw[reduction] (L15) to node[helper] {$L_2^1$} (L59) ;
  \node[operator] (L60) at (-1/2,-5) {$L_\epsilon^5 = 0$} ;
  \draw[reduction] (L56) to node[helper] {$L_0^1$} (L60) ;
  \node[operator] (L61) at (-3/2,-5) {$L_\epsilon^6 = 0$} ;
  \draw[reduction] (L54) to node[helper] {$L_1^1$} (L61) ;
\end{tikzpicture}
\endgroup
\caption{\label{fig:reduction}Execution of Algorithm~\ref{algo:ell0}
  on the operator of Example~\ref{ex:reduction}.
  Each nonzero operator is given with a corresponding pair (order, degree).
  Operators are generated in the following order:
  $L_\epsilon^0 = L$, $L_0^0$, $L_1^0$, $L_2^0$, $L_{20}^0$, $L_{21}^0$, $L_{22}^0$, $L_\epsilon^1$, $L_0^1$, $L_1^1$, $L_2^1$, $L_\epsilon^2$, $L_\epsilon^3$, $L_\epsilon^4$, $L_\epsilon^5$, $L_\epsilon^6$.
  Blue and red arrows respectively represent section and reduction steps.
  Labels on (red) arrows provide the auxiliary operators used for reduction.
  The process starts with~$L_\epsilon^0 = L$ and ends with~$L_1^1$.
  Observe the strict decrease of orders along blue edges
  and large decrease along red edges.
  Also observe that degrees are divided by at least~3 on blue edges
  and, for the only nontrivial red edge of this example,
  how the reduction of~$L_0^0$ by~$L_{21}^0$ induces an increase of the degree
  from~49 to~58, which is not more than~$49 + 12$.}
\end{figure}
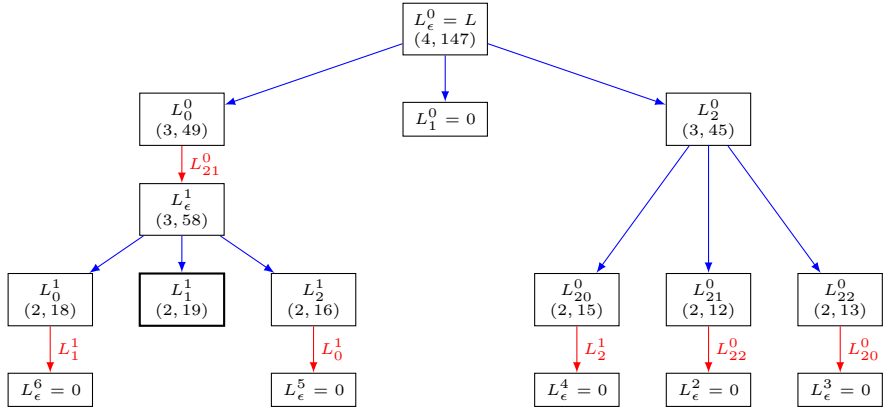
We apply Algorithm~\ref{algo:ell0} with $b = 3$ and the operator
\[
  L = \ell_1 M + \ell_2 M^2 + \ell_3 M^3 + \ell_4 M^4
\]
with
\begin{align*}
  \ell_1 & = x^9 (1-x^{15}+x^{51}+x^{54}-x^{87}+x^{108}) (1-x^{12}+x^{24}) ,
  \displaybreak[0]\\
  \ell_2 & =
  - x^3 \left(1+x^6-x^{20}-x^{21}+x^{30}+x^{32}+x^{33}+x^{36}-x^{44}-x^{45}+x^{54}+x^{56} \right. \\
  & \left. \rule{4em}{0ex} +x^{57}+x^{60}-x^{68}-x^{69}+x^{80}+x^{81}+x^{84}+x^{90}-x^{92}-x^{93}+x^{104} \right. \\
  & \left. \rule{4em}{0ex} +x^{105}+x^{108}+x^{114}-x^{116}-x^{117}+x^{138}+x^{144} \right) ,
  \displaybreak[0]\\
  \ell_3 & =
  \left(1+x^3-x^5+x^{17}+x^{18}+x^{21}-x^{23}-x^{29}+x^{35}+x^{36}+x^{39}-x^{47}+x^{54} \right. \\
  & \left. \rule{2em}{0ex} +x^{57}+x^{72}+x^{75}+x^{90}+x^{93}-x^{95}+x^{107}+x^{108}+x^{111}-x^{113}-x^{119}
 \right.
  \displaybreak[0]\\
  & \left. \rule{2em}{0ex} +x^{125}+x^{126}+x^{129}-x^{137}+x^{144}+x^{147} \right) ,\\
  \ell_4 & = - (1+x^{27}+x^{54}) (1-x^{27}+x^{54}) (1-x^5+x^{17}+x^{18}-x^{29}+x^{36}).
\end{align*}
Starting from $L_\epsilon^0 = L$, we compute its sections
(see Fig.~\ref{fig:reduction}, blue edges):
first, $L_0^0 = S_0(L_\epsilon^0)$,
which has $M$-valuation~0
so that the process of splitting stops for~it;
next, $L_1^0 = S_1(L_\epsilon^0)$,
which is zero and is dropped;
last, $L_2^0 = S_2(L_\epsilon^0)$,
which has $M$-valuation~1.
Splitting continues for the latter and provides
$L_{20}^0 = S_0(L_0^0)$, $L_{21}^0 = S_1(L_2^0)$, $L_{22}^0 = S_2(L_0^0)$,
all with $M$-valuation~0.
Note that during this splitting,
the operators~$L_\epsilon^0$, $L_1^0 = 0$, and~$L_2^0$ disappear.
A reduction is made (see Fig.~\ref{fig:reduction}, red edges)
where $\Spoly(L_0^0,L_{21}^0) = L_\epsilon^6$ replaces~$L_0^0$.
The process continues and, at the end, there only remains
\begin{align*}
  L_1^1 & = x^5 (1+x+x^2) (1-x+x^2) (1-x^{4}+x^8) \\
  & - x^3 (1+x+x^2) (1-x+x^2) (1-x^2+x^{4}-x^6+x^8) (1+2\,x^2+x^{4}) M \\
  & + x^3 (1+x+x^2) (1-x+x^2) (1+x^3+x^6) (1-x^{3}+x^6) M^2.
\end{align*}
It is worth noting that~$L_1^1$ has a content $c = x^3 (1+x+x^2) (1-x+x^2)$,
so that we can write $L_1^1 = c \bar L_1^1$
where~$\bar L_1^1$ is a primitive polynomial (with respect to~$M$).
The computation shows that $L$~is
in the left ideal generated by~$\bar L_1^1$
in the algebra~$\mathcal M(\bQ)$.
This and exhibiting the $M$-valuation~$w = 1$ of~$L$ provides factorizations
$L = L' M = L'' M \bar L_1^1$.
We can say that $M^w$~has been pushed as much as possible to the left.
Using Algorithm~\ref{algo:rat-sol}, we find that a basis of solutions
of~$L_1^1$ in~$\bK(x)$ is given by 1~and $\frac{x}{x^2-1}$.
Since $L_1^1$~has order two,
this also forms a basis of solutions of~$L_1^1$ in~$\bK((x))$,
as a consequence of Proposition~\ref{prop:esp-sol},
and by Proposition~\ref{propo:ell0},
a basis of solutions of~$L$ in~$\bK((x))$.
\end{ex}

We now proceed to prove that Algorithm~\ref{algo:ell0} indeed computes a gcrd
with controlled degree.
This is proved in Theorem~\ref{thm:gcrd} below,
using the following lemmas.

\begin{lem}\label{lem:section-right-div}
For any operators $P_1$, $P_2$,
and any integer~$i$ such that~$0\leq i < b$,
$S_i(P_1 M P_2) = S_i(P_1 M) P_2$.
\end{lem}
\begin{proof}
By linearity, it is sufficient to consider
$P_1 = x^{j_1} M^{k_1}$ and $P_2 = x^{j_2} M^{k_2}$.
Then, $P_1 M P_2 = x^{j_1 + b^{k_1+1} j_2} M^{k_1+k_2+1}$.
Either $b$~divides~$j_1-i$ and
\begin{equation*}
S_i(P_1 M P_2) = x^{(j_1-i)/b + b^{k_1} j_2} M^{k_1+k_2}
  = x^{(j_1-i)/b} M^{k_1} x^{j_2} M^{k_2}
  = S_i(P_1 M) P_2 ,
\end{equation*}
or $b$~does not divide~$j_1-i$ and both extreme terms are zero,
thus equal again.
\end{proof}

\begin{lem}\label{lem:spoly-right-div}
For any operators $P_1$, $P_2$, and~$P$,
all of $M$-valuation~0,
let $c$~be the coefficient of~$M^0$ in~$P$.
Then, $\Spoly(P_1 P, P_2 P) = c \Spoly(P_1, P_2) P$.
\end{lem}
\begin{proof}
The property holds, as obviously
the coefficient of~$M^0$ in a product is
the product of the coefficients of~$M^0$ in the factors.
\end{proof}

\begin{thm}\label{thm:gcrd}
Steps \ref{it:ell0:while} and~\ref{it:ell0:return}
of Algorithm~\ref{algo:ell0} compute
a gcrd of the elements of the split~$\mathcal L$ of~$L$
obtained at step~\ref{it:ell0:initial-split}.
The degree of this particular gcrd is bounded
by the maximal degree of the elements of~$\mathcal L$.
\end{thm}
\begin{proof}
Let $I$ denote the left ideal $\mathcal M(\bK) \mathcal L$
generated by~$\mathcal L$ at any time in the run of the algorithm.
Call $G$ the monic gcrd of the elements of the set~$\mathcal L$
as obtained from~$L$ at the end of step~\ref{it:ell0:initial-split}.
By~\eqref{eq:recon-from-sections}, $G$~is a right factor of~$L$.
By the definition of~$\Spoly(\cdot,\cdot)$
and because of \eqref{eq:recon-from-sections} again,
the ideal~$I$ can only increase during the run of the algorithm,
so that during step~\ref{it:ell0:while},
$\mathcal M(\bK) L \subset \mathcal M(\bK) G \subset I$.

We show by induction that $G$~is a right factor of all elements of~$\mathcal L$
at any time in step~\ref{it:ell0:while},
in other words, that $I \subset \mathcal M(\bK) G$.
This is true by the definition of~$G$ when entering the loop.
The set~$\mathcal L$ contains only elements with $M$-valuation~0,
and it cannot be empty when entering the loop,
so $G$~has $M$-valuation~0 as well.
At any step~\ref{it:ell0:while:norm},
divisibility on the right by~$G$ is preserved for~$\Spoly(L_1, L_2)$,
by Lemma~\ref{lem:spoly-right-div}.
As $\Spoly(L_1, L_2)$~has positive $M$-valuation,
one can choose~$P_2 = G$ and find~$P_1$ so as to write
$\Spoly(L_1, L_2) = P_1 M P_2$.
By Lemma~\ref{lem:section-right-div},
it follows that divisibility on the right by~$G$
is also preserved for each element of~$\mathcal L'$,
then for each element of the next value of~$\mathcal L$.

As a consequence, during step~\ref{it:ell0:while},
$I$~constantly equals $\mathcal M(\bK) G$.
In particular, the final operator~$\tilde L$ is proportional to~$G$.

The degree bound was proved as part of Proposition~\ref{propo:ell0}.
\end{proof}

Note that the origin of the initial~$\mathcal L$ as a split of~$L$,
at step~\ref{it:ell0:initial-split} of Algorithm~\ref{algo:ell0}
plays no role in the proof of Theorem~\ref{thm:gcrd}.
Thus, Algorithm~\ref{algo:ell0} implicitly contains an algorithm
for computing the gcrd of any family~$\mathcal L$
of operators of $M$-valuation zero.

\begin{rem}
We developed Algorithm~\ref{algo:ell0} without targeting a gcrd
and realized Theorem~\ref{thm:gcrd} only a~posteriori.
As Algorithm~\ref{algo:ell0} indeed works by computing a gcrd
as the original algorithm in~\cite{Dumas-1993-RMS},
it is now instructive to compare the result of Proposition~\ref{propo:ell0}
with bounds on the size of gcrds of Mahler operators given by existing methods.
Such a bound can be computed using a variant of the subresultant argument given
by Grigor'ev~\cite[§5]{Grigoriev-1990-CIT} in the differential case.

Let $L_1, \dots, L_n$ be operators of respective order
$r_1 \geq r_2 \geq \dots \geq r_n \geq 1$ and degree
$d_1, \dots, d_n \leq \delta$.
Let $G = U_1 L_1 + \dots + U_n L_n$ be their greatest common right divisor.
We can assume that the order of each term~$U_i L_i$ is less
than~$t = r_1 + r_n$.
Indeed, for all~$i,j$ the linear equation $V_{i,j} L_i = V_{j,i} L_j$
with $V_{i,j}$, resp.~$V_{j,i}$, constrained to have degree
at most~$r_j$, resp. at most~$r_i$, has nontrivial solutions.
Via Euclidean divisions $U_i = Q_i V_{i,n} + R_i$, we obtain
$
G = \sum_i (Q_i V_{i,n} + R_i) L_i
  = \sum_i Q_i V_{n,i} L_n + \sum_i R_i L_i
  = \sum_i \tilde U_i L_i
$
where the~$\tilde U_i$
for $i \leq n - 1$ have order less than~$r_n$.
The $n-1$ first terms~$\tilde U_i L_i$ as well as~$G$ itself have order less
than~$r_1 + r_n$, hence the same must be true of~$\tilde U_n L_n$.

Consider a Sylvester-like matrix~$\Sylv \in \bK[x]^{s \times t}$ with rows
\[ \mathcal R(L_1), \mathcal R(M L_1), \dots, \mathcal R(M^{t-r_1-1} L_1),
   \dots ,
   \mathcal R(L_n), \mathcal R(M L_n), \dots, \mathcal R(M^{t-r_n-1} L_n), \]
where, for any operator $L = \sum_k \ell_k M^k$, we denote
$\mathcal R(L) = (\ell_{t-1}, \dots, \ell_{0})$.
Call $C_0, C_1, \dots, C_{t-1}$ the columns of~$\Sylv$, listed from right
to left (so that~$C_j$ contains the coefficients of~$M^j$ in~$M^k L_i$),
and $C_{j,0}, C_{j,1}, \dots, C_{j,s-1}$ the entries of~$C_j$.
Let $m$ denote the order of~$G$, and choose
$J \subseteq \{ m + 1, \dots, t - 1 \}$
of cardinality $|J| = \rk \Sylv - 1$ in such a way that the columns~$C_j$
with~$j \in J$ form a basis of the span of $C_{m+1}, \dots, C_{t-1}$,
while the~$C_j$ for $j \in \{ m \} \cup J$ form a basis of the full column space
of~$\Sylv$.
To see that such a~$J$ exists, consider a row echelon form of~$\Sylv$:
since $\mathcal R(G)$ belongs to the left image of~$\Sylv$ and $G$~has minimal
order among the nonzero elements of the ideal~$\sum_i \mathcal M(\bK) L_i$,
the rightmost pivot lies on column~$m$.
Further, let $I \subseteq \{0, \dots, s-1\}$ be such that the submatrix
$(C_{j,i})$, $i \in I$, $j \in J \cup \{ m \}$ of~$\Sylv$ is nonsingular.
Call~$D_m$ the corresponding minor, and more generally define
$D_k$ as the determinant of the submatrix~$(C_{j,i})$, $i \in I$,
$j \in J$, extended on the right by a copy of~$C_k$.
Expanding~$D_k$ along the last column yields
$D_k = \sum_{i=0}^{s-1} u_i C_{k,i}$,
where the~$u_i$ do not depend on~$k$.
For each $k > m$, the determinant~$D_k$ is zero,
as $C_k$~is in the span of the~$C_j$ for~$j \in J$.
It follows that the vector
$(0, \dots, 0, D_m, \dots, D_0)$ belongs to the left image of~$\Sylv$.
Thus, there is a gcrd of $L_1, \dots, L_n$ with polynomial coefficients whose
coefficients are minors of~$\Sylv$.

The entries of~$\mathcal S$ have degree bounded by
$\delta' = \max_{i=1}^n (b^{t-r_i-1} d_i)$.
Therefore, the degree of~$G$ is as most
$t \delta' \leq 2 r_1 b^{r_1 - 1} \delta \leq r_1 b^{r_1} \delta$.
Using fast polynomial linear algebra, it is plausible that one could
actually compute~$G$ based on this approach with a complexity of the type
$\softO(\delta' t^\omega) = \softO(b^{r_1} \delta)$.
Now, the gcrd in the algorithm of~\cite{Dumas-1993-RMS}
is that of a family of iterated sections
of the input operator~$L$.
In terms of the order~$r$ and degree~$d$ of~$L$, this family can involve
elements simultaneously of order~$r-1$ and degree~$d/2$.
Thus, Grigor'ev's approach (at least in a straightforward way) would lead
to a complexity bound similar to that of Proposition~\ref{propo:ell0},
but an exponentially worse bound on the degree of the output for large~$r$.
\end{rem}

This result leaves open the question of devising algorithms for computing
solutions of linear Mahler equations that run in polynomial time in
$r$~and~$d$, for all possible combinations of these parameters, even when the
trailing coefficient~$\ell_0$ of the equation is zero.
In particular, it would be interesting to see if the bounds on the size
of an operator equivalent to~$L$ implied by Algorithm~\ref{algo:split-ell0}
would be enough to extend
the algorithms of~§\ref{sec:series}--\ref{sec:rat-sols}
to the case where $\ell_0$~is zero, without going through the explicit
computation of such an operator.

We end the section by providing an extension of Algorithm~\ref{algo:ell0},
which computes a gcrd for a family of operators of arbitrary $M$-valuations.

\begin{algo}
\inputs{
  A finite family~$\{L_i\}_{i=1}^s$ of linear Mahler operators
  with polynomial coefficients,
  orders at most~$r$, minimal $M$-valuation $w$, and degrees at most~$d$.
}
\outputs{
  A linear Mahler operator~$\tilde L$
  of order $\tilde{r}\leq r$, $M$-valuation~$w$,
  and degree $\tilde{d}\leq d$.
}
\caption{\label{algo:gcrd}Computation of a gcrd of an arbitrary family.}
\begin{enumerate}
\item\label{it:gcrd:initial-factorizations}
Write each~$L_i$ in the form~$L'_i M^w$, for a polynomial~$L'_i$ in $x$ and~$M$.
\item\label{it:gcrd:initial-splits}
Let $\mathcal L$ be the union
of the results of applying Algorithm~\ref{algo:split-ell0} to the~$L'_i$'s.
\item\label{it:gcrd:while}
While $\mathcal L$ has at least two elements:
  \begin{enumerate}
  \item\label{it:gcrd:while:choice}
    choose $L_1$ with highest order in~$\mathcal L$,
    then $L_2$ from~$\mathcal L \setminus \{L_1\}$;
  \item\label{it:gcrd:while:norm}
    compute the result~$\mathcal L'$
    of applying Algorithm~\ref{algo:split-ell0}
    to the interreduction~$\Spoly(L_1, L_2)$;
  \item\label{it:gcrd:while:update}
    replace $\mathcal L$ by $(\mathcal L \setminus \{L_1\}) \cup \mathcal L'$.
  \end{enumerate}
\item\label{it:gcrd:return}
Write~$\tilde L$ for the single element of the singleton~$\mathcal L$
and return~$\tilde L M^w$.
\end{enumerate}
\end{algo}

\begin{thm}\label{thm:algo-gcrd}
Algorithm~\ref{algo:gcrd} computes
a gcrd of the input operators $L_1$, \dots, $L_s$.
\end{thm}
\begin{proof}
Observe that the minimal $M$-valuation of operators in a family
is the minimal $M$-valuation of elements
of the left ideal generated by the family,
in particular, the $M$-valuation of any gcrd of the family.
This justifies the general design of the algorithm,
with the factorization of~$M^w$ on the right
at step~\ref{it:gcrd:initial-factorizations}.

By construction, the~$L'_i$'s thus obtained
have orders at most~$r-w$ and degrees at most~$d$,
and at least one, say~$L'_1$, has $M$-valuation zero.
Let $G'$ denote the monic gcrd of the~$L'_i$,
which, as~$L'_1$, has $M$-valuation zero.
By Lemma~\ref{lem:section-right-div}, $G'$~is a right-hand factor
of all elements of the set~$\mathcal L$ computed
at step~\ref{it:gcrd:initial-splits}.
By a proof similar to the one for Theorem~\ref{thm:gcrd},
it remains so for all subsequent values of~$\mathcal L$,
so for the~$\tilde L$ of step~\ref{it:gcrd:return} as well.

As $\tilde L$ is also obviously a right-hand factor
of all previously computed operators,
including the~$L'_i$'s,
$\tilde L$~is a gcrd of the latter.
This concludes the proof.
\end{proof}

\bibliographystyle{amsplain}
\providecommand{\bysame}{\leavevmode\hbox to3em{\hrulefill}\thinspace}
\providecommand{\MR}{\relax\ifhmode\unskip\space\fi MR }
\providecommand{\MRhref}[2]{%
  \href{http://www.ams.org/mathscinet-getitem?mr=#1}{#2}
}
\providecommand{\href}[2]{#2}

\end{document}